\documentclass[acmsmall,screen,dvipsnames]{acmart}  
\setcopyright{none}
\acmPrice{}
\acmDOI{10.1145/3341692}
\acmYear{2019}

    \usepackage{booktabs}                               \usepackage{subcaption}

\renewcommand\footnotetextcopyrightpermission[1]{} 
\makeatletter
\renewcommand\@formatdoi[1]{\ignorespaces}
\makeatother

\usepackage{xspace}
\usepackage{etoolbox}
\usepackage{stmaryrd}
\usepackage{mathpartir}
\mprset{sep=1em}
\usepackage{booktabs}

\usepackage{enumitem}

\usepackage{mathtools}
\usepackage{braket}
\usepackage{multirow}
\usepackage{listings}
\usepackage{xstring}

\usepackage{framed}

\usepackage[]{tcolorbox}
\tcbuselibrary{breakable}
\tcbuselibrary{skins}
\tcbset{floatplacement=htbp}

\usepackage{supertabular}
\usepackage{geometry}
\usepackage{ifthen}
\usepackage{alltt}
\usepackage{semantic}
\usepackage{array}
\newcolumntype{C}[1]{>{\centering\let\newline\\\arraybackslash\hspace{0pt}}m{#1}}

\usepackage{checkend}	\usepackage{graphicx}	\usepackage{booktabs}

\usepackage{mathtools}
\usepackage{thm-restate}
\definecolor{greytext}{gray}{0.5}

\usepackage{etoolbox}

\usepackage{tabularx}

\newcommand{\oblset}[1]{\textsc{#1}}

\usepackage{amssymb}
\usepackage{color}

\definecolor{light-gray}{gray}{0.9}
\definecolor{dark-gray}{gray}{0.5}

\newcommand{\lang}{GDTL\xspace} \newcommand{\slang}{SDTL\xspace}

\definecolor{lightgray}{gray}{0.90}

\newcommand{\gbox}[1]{{\setlength{\fboxsep}{1pt}\colorbox{lightgray}{\tiny$#1$}}}

\newcommand{\lub}[1]{\sqcup}

\newcounter{obj-trt-line}

\mathlig{|}{\mid}

\newsavebox{\saveboxedarray}

\reservestyle{\command}{\textsf}
\command{fn[fn\;],main[main\;],unit[()],trait[trait\;],impl[impl],Fn[fn],
  type[type\;],ref[\&],let[let\;],in[\;in\;],as[\;as\;],struct[struct\;],
  static[static\;],pack[pack\;],unpack[unpack\;],axiom[axiom\;],sym[sym\;],
  where[where\;],spar[spar\;],farg[farg\;],fret[fret\;],objf[objf\;],
  objc-l[objc-l\;],objc-r[objc-r\;],for[for\;],coer-l[coer-l\;],
  coer-r[coer-r\;],coer[coer\;],obj-coer[obj-coer\;],extern[extern\;]}

\newcommand{\Pow}{\mathcal{P}}

\definecolor{bluekeywords}{rgb}{0.13,0.13,1}
\definecolor{greencomments}{rgb}{0,0.5,0}
\definecolor{redstrings}{rgb}{0.9,0,0}

\lstdefinelanguage{Rust}
{
    morekeywords={
    as, break, const, continue, crate, else, enum, extern, false, fn, for, if, impl, in, let, loop, match, mod, move, mut, pub, ref, return, Self, self, static, struct, super, trait, true, type, unsafe, use, where, while
  },
  otherkeywords={?, Unary, -, *, !, &, &mut, as, :, *, /, +, -, <<, >>, &, ^, |, ==, !=, <, >, <=, >=, &&, ||, .., ..., <-, =, +=, -=, *=, /=, &=, |=, ^=, <<=, >>=},
  sensitive=true,   morecomment=[l]{//},   morecomment=[s]{/*}{*/},   morestring=[b]",   commentstyle=\itshape\color{blue}
}

\lstnewenvironment{pseudocode}{
  \lstset{escapeinside={(*}{*)}
  , basicstyle=\rmfamily
  , emph={    Robinson, normalize, simplify, simplC, lowerCoercion    }
    ,morekeywords={if,then,else,let,in}
    ,emphstyle={\bfseries}    ,frame=none
  }
}{}

\lstnewenvironment{rustcode}
  {
    \lstset{
        language=C++,
        basicstyle=\ttfamily,
        breaklines=true,
        columns=fullflexible,
        frame = single,
        tabsize=2,
        morekeywords={
    as, break, const, continue, crate, else, enum, extern, false, fn, for, if, impl, in, let, loop, match, mod, move, mut, pub, ref, return, Self, self, static, struct, super, trait, true, type, unsafe, use, where, while
  },
  commentstyle=\rmfamily\itshape\color{darkgray}}
  }
  {
  }

  \lstnewenvironment{rustcodenoframe}
  {
    \lstset{
        language=C++,
        basicstyle=\ttfamily,
        breaklines=true,
        columns=fullflexible,
        frame = none,
        tabsize=2,
        morekeywords={
    as, break, const, continue, crate, else, enum, extern, false, fn, for, if, impl, in, let, loop, match, mod, move, mut, pub, ref, return, Self, self, static, struct, super, trait, true, type, unsafe, use, where, while
  },
        commentstyle=\rmfamily\itshape\color{darkgray}}
  }
  {
  }

\makeatletter

  \def\z#1{\sbox{0}{$#1$}  \ifdim\wd0>0.2\textwidth
     \expandafter\@firstoftwo
   \else
     \expandafter\@secondoftwo
    \fi
    {\rlap{\usebox{0}}\\\multicolumn{2}{r}{}}{\usebox{0}}}
\makeatother

\BeforeBeginEnvironment{verbatim}{\def\baselinestretch{1}}

\newcommand{\AllTerms}{\oblset{SCanonical}}
\newcommand{\AllGTerms}{\oblset{GCanonical}}
\newcommand{\pto}{\rightharpoonup}
\newcommand{\abe}{{\alpha\beta\eta}}

\newcommand{\CCw}{$\mathsf{CC}_\omega$}

\usepackage{hyperref}
\usepackage{cleveref}

\newcommand{\rrule}[1]{\rref*{#1}}

\usepackage[implicitLineBreakHack]{ottalt}

\newcommand{\ottdrule}[4][]{{\displaystyle\frac{\begin{array}{l}#2\end{array}}{#3}\quad\ottdrulename{#4}}}
\newcommand{\ottusedrule}[1]{\[#1\]}
\newcommand{\ottpremise}[1]{ #1 \\}
\newenvironment{ottdefnblock}[3][]{ \framebox{\mbox{#2}} \quad #3 \\[0pt]}{}

\newcommand{\ottnt}[1]{\mathit{#1}}

\newcommand{\ottkw}[1]{\mathbf{#1}}
\newcommand{\ottsym}[1]{#1}
\newcommand{\ottcom}[1]{\text{#1}}
\newcommand{\ottdrulename}[1]{\textsc{#1}}

\newcommand{\ottafterlastrule}{\\}

\usepackage{stmaryrd}
\usepackage{graphicx}
\usepackage{mathtools}
\usepackage{etoolbox}
\newcommand\leadsfrom{\reflectbox{$\leadsto$ } }

\newbool{ShowEmptyDot}
\setbool{ShowEmptyDot}{false}

\newbool{ShowArrowIndex}
\setbool{ShowArrowIndex}{true}

\newcommand\myepsilon\varepsilon

\newcommand\qm{\textbf{?} }

\newcommand{\gobble}[1]{}

\newcommand{\emptyspine}{\ifboolexpr{bool{ShowEmptyDot} }{\cdot}{} }

\newcommand{\gradualstyle}[3]{\color{RoyalBlue}\mathrm{ {#1}_{#2}#3 } }
\newcommand{\staticstyle}[3]{\color{BrickRed}\mathsf{ {#1}_{#2}#3 } }
\newcommand{\staticgreek}[3]{\color{BrickRed} {#1}_{#2}#3  }
\newcommand{\evstyle}[3]{\color{RoyalBlue}\mathrm{ {#1}_{#2}#3 } }

\newNTclass{gradual}
\newNTclass{evterm}
\newNTclass{static}

\newgradual[\NTCAPTURELOW{\gradualstyle }]{u}{u}
\newgradual[\NTCAPTURELOW{\gradualstyle }]{v}{v}
\newgradual[\NTCAPTURELOW{\gradualstyle }]{U}{U}
\newgradual[\NTCAPTURELOW{\gradualstyle }]{V}{V}
\newgradual[\NTCAPTURELOW{\gradualstyle }]{s}{s}
\newgradual[\NTCAPTURELOW{\gradualstyle }]{t}{t}
\newgradual[\NTCAPTURELOW{\gradualstyle }]{S}{S}
\newgradual[\NTCAPTURELOW{\gradualstyle }]{T}{T}
\newgradual[\NTCAPTURELOW{\gradualstyle }]{rr}{r}
\newgradual[\NTCAPTURELOW{\gradualstyle }]{RR}{R}
\newgradual[\NTCAPTURELOW{\gradualstyle }]{e}{\overline{s} }
\newgradual[\NTCAPTURELOW{\gradualstyle }]{Gamma}{\Gamma}
\newevterm[\NTCAPTURELOW{\evstyle }]{s}{s}
\newevterm[\NTCAPTURELOW{\evstyle }]{t}{e}
\newevterm[\NTCAPTURELOW{\evstyle }]{S}{S}
\newevterm[\NTCAPTURELOW{\evstyle }]{T}{E}
\newevterm[\NTCAPTURELOW{\evstyle }]{u}{u}
\newevterm[\NTCAPTURELOW{\evstyle }]{v}{v}
\newevterm[\NTCAPTURELOW{\evstyle }]{U}{U}
\newevterm[\NTCAPTURELOW{\evstyle }]{V}{V}
\newevterm[\NTCAPTURELOW{\evstyle }]{rv}{w}
\newevterm[\NTCAPTURELOW{\evstyle }]{rV}{W}
\newevterm[\NTCAPTURELOW{\evstyle }]{C}{\mathcal{C} }
\newstatic[\NTCAPTURELOW{\staticstyle }]{e}{\overline{s} }
\newstatic[\NTCAPTURELOW{ \UNDEFINED }]{s}{s}
\newstatic[\NTCAPTURELOW{\staticstyle }]{t}{t}
\newstatic[\NTCAPTURELOW{\UNDEFINED }]{S}{S}
\newstatic[\NTCAPTURELOW{\staticstyle }]{T}{T}
\newstatic[\NTCAPTURELOW{\staticstyle }]{u}{u}
\newstatic[\NTCAPTURELOW{\staticstyle }]{v}{v}
\newstatic[\NTCAPTURELOW{\staticstyle }]{U}{U}
\newstatic[\NTCAPTURELOW{\staticstyle }]{V}{V}
\newstatic[\NTCAPTURELOW{\staticstyle }]{rr}{r}
\newstatic[\NTCAPTURELOW{\staticstyle }]{RR}{R}
\newstatic[\NTCAPTURELOW{\staticstyle }]{C}{\mathcal{C} }
\newstatic[\NTCAPTURELOW{\staticgreek }]{Gamma}{\Gamma}

\newcommand{\TypeType}{\rev{\ottkw{Type} } }

\newcommand{\ottdruleSHsubPi}[1]{\ottdrule[#1]{%
\ottpremise{ {[  \static{u}  / { \mathit{x} } ]}^{ \static{U}  }  \static{U}_{{\mathrm{1}}} =  \static{U}'_{{\mathrm{1}}} }%
\ottpremise{ {[  \static{u}  / { \mathit{x} } ]}^{ \static{U}  }  \static{U}_{{\mathrm{2}}} =  \static{U}'_{{\mathrm{2}}} }%
\ottpremise{ \mathit{x}   \neq   \mathit{y} }%
}{
 {[  \static{u}  / { \mathit{x} } ]}^{ \static{U}  }  \ottsym{(}  \mathit{y}  \ottsym{:}  \static{U}_{{\mathrm{1}}}  \ottsym{)}  \rightarrow  \static{U}_{{\mathrm{2}}} =  \ottsym{(}  \mathit{y}  \ottsym{:}  \static{U}'_{{\mathrm{1}}}  \ottsym{)}  \rightarrow  \static{U}'_{{\mathrm{2}}} }{%
{\ottdrulename{SHsubPi}}{}%
}}

\newcommand{\ottdruleSHsubDiffNil}[1]{\ottdrule[#1]{%
\ottpremise{ \mathit{x}   \neq   \mathit{y} }%
}{
 {[  \static{u}  / { \mathit{x} } ]}^{ \static{U}  }   \mathit{y}  =   \mathit{y}  }{%
{\ottdrulename{SHsubDiffNil}}{}%
}}

\newcommand{\ottdruleSHsubDiffCons}[1]{\ottdrule[#1]{%
\ottpremise{ \mathit{x}   \neq   \mathit{y} }%
\ottpremise{ {[  \static{u}  / { \mathit{x} } ]}^{ \static{U}  }   \mathit{y} \static{e}  =   \mathit{y} \static{e}'  }%
\ottpremise{ {[  \static{u}  / { \mathit{x} } ]}^{ \static{U}  }  \static{u}_{{\mathrm{2}}} =  \static{u}_{{\mathrm{3}}} }%
}{
 {[  \static{u}  / { \mathit{x} } ]}^{ \static{U}  }   \mathit{y}  \static{e} \  \static{u}_{{\mathrm{2}}}   =   \mathit{y}  \static{e}' \  \static{u}_{{\mathrm{3}}}   }{%
{\ottdrulename{SHsubDiffCons}}{}%
}}

\newcommand{\ottdruleSHsubType}[1]{\ottdrule[#1]{%
}{
 {[  \static{u}  / { \mathit{x} } ]}^{ \static{U}  }   \TypeType_{ \ottnt{i} }  =   \TypeType_{ \ottnt{i} }  }{%
{\ottdrulename{SHsubType}}{}%
}}

\newcommand{\ottdruleSHsubLam}[1]{\ottdrule[#1]{%
\ottpremise{ {[  \static{u}  / { \mathit{x} } ]}^{ \static{U}  }  \static{u}_{{\mathrm{2}}} =  \static{u}_{{\mathrm{3}}} }%
\ottpremise{ \mathit{x}   \neq   \mathit{y} }%
}{
 {[  \static{u}  / { \mathit{x} } ]}^{ \static{U}  }  \ottsym{(}  \lambda  \mathit{y}  \ldotp  \static{u}_{{\mathrm{2}}}  \ottsym{)} =  \ottsym{(}  \lambda  \mathit{y}  \ldotp  \static{u}_{{\mathrm{3}}}  \ottsym{)} }{%
{\ottdrulename{SHsubLam}}{}%
}}

\newcommand{\ottdruleSHsubSpine}[1]{\ottdrule[#1]{%
\ottpremise{ {[  \static{u}_{{\mathrm{1}}}  / { \mathit{x} } ]}^{ \static{U}  } { \mathit{x}  \static{e} \  \static{u}_{{\mathrm{2}}}   }=\ { \static{u}_{{\mathrm{3}}} }\ :\ { \static{U}' } }%
}{
 {[  \static{u}_{{\mathrm{1}}}  / { \mathit{x} } ]}^{ \static{U}  }   \mathit{x}  \static{e} \  \static{u}_{{\mathrm{2}}}   =  \static{u}_{{\mathrm{3}}} }{%
{\ottdrulename{SHsubSpine}}{}%
}}

\newcommand{\ottdefnSHsub}[1]{\begin{ottdefnblock}[#1]{$ {[  \static{u}_{{\mathrm{1}}}  / { \mathit{x} } ]}^{ \static{U}  }  \static{u}_{{\mathrm{2}}} =  \static{u}_{{\mathrm{3}}} $}{\ottcom{Static Hereditary Substitution}}
\ottusedrule{\ottdruleSHsubPi{}}
\ottusedrule{\ottdruleSHsubDiffNil{}}
\ottusedrule{\ottdruleSHsubDiffCons{}}
\ottusedrule{\ottdruleSHsubType{}}
\ottusedrule{\ottdruleSHsubLam{}}
\ottusedrule{\ottdruleSHsubSpine{}}
\end{ottdefnblock}}

\newcommand{\ottdruleSHsubRHead}[1]{\ottdrule[#1]{%
}{
 {[  \static{u}  / { \mathit{x} } ]}^{ \static{U}  } { \mathit{x}  \emptyspine   }=\ { \static{u} }\ :\ { \static{U} } }{%
{\ottdrulename{SHsubRHead}}{}%
}}

\newcommand{\ottdruleSHsubRSpine}[1]{\ottdrule[#1]{%
\ottpremise{  {[  \static{u}_{{\mathrm{1}}}  / { \mathit{x} } ]}^{ \static{U}  } { \mathit{x} \static{e}  }=\ { \ottsym{(}  \lambda  \mathit{y}  \ldotp  \static{u}'_{{\mathrm{1}}}  \ottsym{)} }\ :\ { \ottsym{(}  \mathit{y}  \ottsym{:}  \static{U}'_{{\mathrm{1}}}  \ottsym{)}  \rightarrow  \static{U}'_{{\mathrm{2}}} }   \qquad   {[  \static{u}_{{\mathrm{1}}}  / { \mathit{x} } ]}^{ \static{U}  }  \static{u}_{{\mathrm{2}}} =  \static{u}_{{\mathrm{3}}}  }%
\ottpremise{  {[  \static{u}_{{\mathrm{3}}}  / { \mathit{y} } ]}^{ \static{U}'_{{\mathrm{1}}}  }  \static{u}'_{{\mathrm{1}}} =  \static{u}'_{{\mathrm{2}}}   \qquad   {[  \static{u}_{{\mathrm{3}}}  / { \mathit{y} } ]}^{ \static{U}'_{{\mathrm{1}}}  }  \static{U}'_{{\mathrm{2}}} =  \static{U}'_{{\mathrm{3}}}  }%
}{
 {[  \static{u}_{{\mathrm{1}}}  / { \mathit{x} } ]}^{ \static{U}  } { \mathit{x}  \static{e} \  \static{u}_{{\mathrm{2}}}   }=\ { \static{u}'_{{\mathrm{2}}} }\ :\ { \static{U}'_{{\mathrm{3}}} } }{%
{\ottdrulename{SHsubRSpine}}{}%
}}

\newcommand{\ottdefnSHsubR}[1]{\begin{ottdefnblock}[#1]{$ {[  \static{u}  / { \mathit{x} } ]}^{ \static{U}  } { \mathit{x} \static{e}  }=\ { \static{u}' }\ :\ { \static{U}' } $}{\ottcom{Static Atomic Hereditary Substitution}}
\ottusedrule{\ottdruleSHsubRHead{}}
\ottusedrule{\ottdruleSHsubRSpine{}}
\end{ottdefnblock}}

\newcommand{\ottdruleSWFEmpty}[1]{\ottdrule[#1]{%
}{
 \vdash  \cdot }{%
{\ottdrulename{SWFEmpty}}{}%
}}

\newcommand{\ottdruleSWFExt}[1]{\ottdrule[#1]{%
\ottpremise{   \vdash  \static{Gamma}   \qquad   \static{Gamma}  \vdash  \static{U}  : \TypeType    \qquad  \mathit{x} \, \not\in \, \static{Gamma} }%
}{
 \vdash  \ottsym{(}  \mathit{x}  \ottsym{:}  \static{U}  \ottsym{)}  \static{Gamma} }{%
{\ottdrulename{SWFExt}}{}%
}}

\newcommand{\ottdefnSWF}[1]{\begin{ottdefnblock}[#1]{$ \vdash  \static{Gamma} $}{\ottcom{Well Formed Static Environments}}
\ottusedrule{\ottdruleSWFEmpty{}}
\ottusedrule{\ottdruleSWFExt{}}
\end{ottdefnblock}}

\newcommand{\ottdruleStaticTypeType}[1]{\ottdrule[#1]{%
\ottpremise{\static{Gamma}  \vdash  \static{rr}  \Rightarrow   \TypeType_{ \ottnt{i} } }%
}{
 \static{Gamma}  \vdash  \static{rr}  : \TypeType }{%
{\ottdrulename{StaticTypeType}}{}%
}}

\newcommand{\ottdruleStaticTypePi}[1]{\ottdrule[#1]{%
\ottpremise{ \static{Gamma}  \vdash  \static{U}_{{\mathrm{1}}}  : \TypeType }%
\ottpremise{ \vdash  \ottsym{(}  \mathit{x}  \ottsym{:}  \static{U}_{{\mathrm{1}}}  \ottsym{)}  \static{Gamma} }%
\ottpremise{ \ottsym{(}  \mathit{x}  \ottsym{:}  \static{U}_{{\mathrm{1}}}  \ottsym{)}  \static{Gamma}  \vdash  \static{U}_{{\mathrm{2}}}  : \TypeType }%
}{
 \static{Gamma}  \vdash  \ottsym{(}  \mathit{x}  \ottsym{:}  \static{U}_{{\mathrm{1}}}  \ottsym{)}  \rightarrow  \static{U}_{{\mathrm{2}}}  : \TypeType }{%
{\ottdrulename{StaticTypePi}}{}%
}}

\newcommand{\ottdefnStaticType}[1]{\begin{ottdefnblock}[#1]{$ \static{Gamma}  \vdash  \static{U}  : \TypeType $}{\ottcom{Well-formed Types with Unknown Level}}
\ottusedrule{\ottdruleStaticTypeType{}}
\ottusedrule{\ottdruleStaticTypePi{}}
\end{ottdefnblock}}

\newcommand{\ottdruleSCSynthType}[1]{\ottdrule[#1]{%
\ottpremise{\ottnt{i}  \ottsym{>}  \ottsym{0}}%
}{
\static{Gamma}  \vdash   \TypeType_{ \ottnt{i} }   \Rightarrow   \TypeType_{ \ottnt{i}  + 1} }{%
{\ottdrulename{SCSynthType}}{}%
}}

\newcommand{\ottdruleSCSynthVar}[1]{\ottdrule[#1]{%
\ottpremise{ \vdash  \static{Gamma} }%
\ottpremise{\ottsym{(}  \mathit{x}  \ottsym{:}  \static{U}  \ottsym{)} \, \in \, \static{Gamma}}%
}{
\static{Gamma}  \vdash   \mathit{x}   \Rightarrow  \static{U}}{%
{\ottdrulename{SCSynthVar}}{}%
}}

\newcommand{\ottdruleSCSynthApp}[1]{\ottdrule[#1]{%
\ottpremise{\static{Gamma}  \vdash   \mathit{x} \static{e}   \Rightarrow  \ottsym{(}  \mathit{y}  \ottsym{:}  \static{U}_{{\mathrm{1}}}  \ottsym{)}  \rightarrow  \static{U}_{{\mathrm{2}}}}%
\ottpremise{ \static{Gamma}  \vdash  \static{u}  \Leftarrow  \static{U}_{{\mathrm{1}}}  \qquad   {[  \static{u}  / { \mathit{y} } ]}^{ \static{U}_{{\mathrm{1}}}  }  \static{U}_{{\mathrm{2}}} =  \static{U}_{{\mathrm{3}}}  }%
}{
\static{Gamma}  \vdash   \mathit{x}  \static{e} \  \static{u}    \Rightarrow  \static{U}_{{\mathrm{3}}}}{%
{\ottdrulename{SCSynthApp}}{}%
}}

\newcommand{\ottdefnSCSynth}[1]{\begin{ottdefnblock}[#1]{$\static{Gamma}  \vdash  \static{rr}  \Rightarrow  \static{U}$}{\ottcom{Static Well-Formed Atomics}}
\ottusedrule{\ottdruleSCSynthType{}}
\ottusedrule{\ottdruleSCSynthVar{}}
\ottusedrule{\ottdruleSCSynthApp{}}
\end{ottdefnblock}}

\newcommand{\ottdruleSCCheckSynth}[1]{\ottdrule[#1]{%
\ottpremise{\static{Gamma}  \vdash  \static{rr}  \Rightarrow  \static{RR}}%
}{
\static{Gamma}  \vdash  \static{rr}  \Leftarrow  \static{RR}}{%
{\ottdrulename{SCCheckSynth}}{}%
}}

\newcommand{\ottdruleSCCheckLevel}[1]{\ottdrule[#1]{%
\ottpremise{\static{Gamma}  \vdash  \static{RR}  \Rightarrow   \TypeType_{ \ottnt{i} } }%
\ottpremise{ 0 <  \ottnt{i}  <  \ottnt{j} }%
}{
\static{Gamma}  \vdash  \static{RR}  \Leftarrow   \TypeType_{ \ottnt{j} } }{%
{\ottdrulename{SCCheckLevel}}{}%
}}

\newcommand{\ottdruleSCCheckLam}[1]{\ottdrule[#1]{%
\ottpremise{  \vdash  \ottsym{(}  \mathit{x}  \ottsym{:}  \static{U}_{{\mathrm{1}}}  \ottsym{)}  \static{Gamma}   \qquad  \ottsym{(}  \mathit{x}  \ottsym{:}  \static{U}_{{\mathrm{1}}}  \ottsym{)}  \static{Gamma}  \vdash  \static{u}  \Leftarrow  \static{U}_{{\mathrm{2}}} }%
}{
\static{Gamma}  \vdash  \ottsym{(}  \lambda  \mathit{x}  \ldotp  \static{u}  \ottsym{)}  \Leftarrow  \ottsym{(}  \mathit{x}  \ottsym{:}  \static{U}_{{\mathrm{1}}}  \ottsym{)}  \rightarrow  \static{U}_{{\mathrm{2}}}}{%
{\ottdrulename{SCCheckLam}}{}%
}}

\newcommand{\ottdruleSCCheckPi}[1]{\ottdrule[#1]{%
\ottpremise{\static{Gamma}  \vdash  \static{U}_{{\mathrm{1}}}  \Leftarrow   \TypeType_{ \ottnt{i} } }%
\ottpremise{  \vdash  \ottsym{(}  \mathit{x}  \ottsym{:}  \static{U}_{{\mathrm{1}}}  \ottsym{)}  \static{Gamma}   \qquad  \ottsym{(}  \mathit{x}  \ottsym{:}  \static{U}_{{\mathrm{1}}}  \ottsym{)}  \static{Gamma}  \vdash  \static{U}_{{\mathrm{2}}}  \Leftarrow   \TypeType_{ \ottnt{i} }  }%
}{
\static{Gamma}  \vdash  \ottsym{(}  \mathit{x}  \ottsym{:}  \static{U}_{{\mathrm{1}}}  \ottsym{)}  \rightarrow  \static{U}_{{\mathrm{2}}}  \Leftarrow   \TypeType_{ \ottnt{i} } }{%
{\ottdrulename{SCCheckPi}}{}%
}}

\newcommand{\ottdefnSCCheck}[1]{\begin{ottdefnblock}[#1]{$\static{Gamma}  \vdash  \static{u}  \Leftarrow  \static{U}$}{\ottcom{Static Well-Formed Canonical Forms}}
\ottusedrule{\ottdruleSCCheckSynth{}}
\ottusedrule{\ottdruleSCCheckLevel{}}
\ottusedrule{\ottdruleSCCheckLam{}}
\ottusedrule{\ottdruleSCCheckPi{}}
\end{ottdefnblock}}

\newcommand{\ottdruleSSynthAnn}[1]{\ottdrule[#1]{%
\ottpremise{ \static{Gamma}  \vdash  \static{U}  \leadsfrom  \static{T}  : \TypeType }%
\ottpremise{\static{Gamma}  \vdash  \static{t}  \Leftarrow  \static{U}}%
}{
\static{Gamma}  \vdash  \ottsym{(}   \static{t} \dblcolon \static{T}   \ottsym{)}  \Rightarrow  \static{U}}{%
{\ottdrulename{SSynthAnn}}{}%
}}

\newcommand{\ottdruleSSynthType}[1]{\ottdrule[#1]{%
\ottpremise{\ottnt{i}  \ottsym{>}  \ottsym{0}}%
}{
\static{Gamma}  \vdash   \TypeType_{ \ottnt{i} }   \Rightarrow   \TypeType_{ \ottnt{i}  + 1} }{%
{\ottdrulename{SSynthType}}{}%
}}

\newcommand{\ottdruleSSynthVar}[1]{\ottdrule[#1]{%
\ottpremise{ \vdash  \static{Gamma} }%
\ottpremise{\ottsym{(}  \mathit{x}  \ottsym{:}  \static{U}  \ottsym{)} \, \in \, \static{Gamma}}%
}{
\static{Gamma}  \vdash  \mathit{x}  \Rightarrow  \static{U}}{%
{\ottdrulename{SSynthVar}}{}%
}}

\newcommand{\ottdruleSSynthApp}[1]{\ottdrule[#1]{%
\ottpremise{ \static{Gamma}  \vdash  \static{t}_{{\mathrm{1}}}  \Rightarrow  \ottsym{(}  \mathit{x}  \ottsym{:}  \static{U}_{{\mathrm{1}}}  \ottsym{)}  \rightarrow  \static{U}_{{\mathrm{2}}}  \qquad  \static{Gamma}  \vdash  \static{u}  \leadsfrom\  \static{t}_{{\mathrm{2}}}  \Leftarrow  \static{U}_{{\mathrm{1}}} }%
\ottpremise{ {[  \static{u}  / { \mathit{x} } ]}^{ \static{U}_{{\mathrm{1}}}  }  \static{U}_{{\mathrm{2}}} =  \static{U}_{{\mathrm{3}}} }%
}{
\static{Gamma}  \vdash   \static{t}_{{\mathrm{1}}} \  \static{t}_{{\mathrm{2}}}   \Rightarrow  \static{U}_{{\mathrm{3}}}}{%
{\ottdrulename{SSynthApp}}{}%
}}

\newcommand{\ottdefnSSynth}[1]{\begin{ottdefnblock}[#1]{$\static{Gamma}  \vdash  \static{t}  \Rightarrow  \static{U}$}{\ottcom{Static Synthesis}}
\ottusedrule{\ottdruleSSynthAnn{}}
\ottusedrule{\ottdruleSSynthType{}}
\ottusedrule{\ottdruleSSynthVar{}}
\ottusedrule{\ottdruleSSynthApp{}}
\end{ottdefnblock}}

\newcommand{\ottdruleSCheckSynth}[1]{\ottdrule[#1]{%
\ottpremise{\static{Gamma}  \vdash  \static{t}  \Rightarrow  \static{U}}%
}{
\static{Gamma}  \vdash  \static{t}  \Leftarrow  \static{U}}{%
{\ottdrulename{SCheckSynth}}{}%
}}

\newcommand{\ottdruleSCheckLevel}[1]{\ottdrule[#1]{%
\ottpremise{\static{Gamma}  \vdash  \static{T}  \Rightarrow   \TypeType_{ \ottnt{i} } }%
\ottpremise{ 0 <  \ottnt{i}  <  \ottnt{j} }%
}{
\static{Gamma}  \vdash  \static{T}  \Leftarrow   \TypeType_{ \ottnt{j} } }{%
{\ottdrulename{SCheckLevel}}{}%
}}

\newcommand{\ottdruleSCheckPi}[1]{\ottdrule[#1]{%
\ottpremise{ \static{Gamma}  \vdash  \static{U}  \leadsfrom\  \static{T}_{{\mathrm{1}}}  \Leftarrow   \TypeType_{ \ottnt{i} }   \qquad   \vdash  \ottsym{(}  \mathit{x}  \ottsym{:}  \static{U}  \ottsym{)}  \static{Gamma}  }%
\ottpremise{\ottsym{(}  \mathit{x}  \ottsym{:}  \static{U}  \ottsym{)}  \static{Gamma}  \vdash  \static{T}_{{\mathrm{2}}}  \Leftarrow   \TypeType_{ \ottnt{i} } }%
}{
\static{Gamma}  \vdash  \ottsym{(}  \mathit{x}  \ottsym{:}  \static{T}_{{\mathrm{1}}}  \ottsym{)}  \rightarrow  \static{T}_{{\mathrm{2}}}  \Leftarrow   \TypeType_{ \ottnt{i} } }{%
{\ottdrulename{SCheckPi}}{}%
}}

\newcommand{\ottdruleSCheckLam}[1]{\ottdrule[#1]{%
\ottpremise{ \vdash  \ottsym{(}  \mathit{x}  \ottsym{:}  \static{U}_{{\mathrm{1}}}  \ottsym{)}  \static{Gamma} }%
\ottpremise{\ottsym{(}  \mathit{x}  \ottsym{:}  \static{U}_{{\mathrm{1}}}  \ottsym{)}  \static{Gamma}  \vdash  \static{t}  \Leftarrow  \static{U}_{{\mathrm{2}}}}%
}{
\static{Gamma}  \vdash  \ottsym{(}  \lambda  \mathit{x}  \ldotp  \static{t}  \ottsym{)}  \Leftarrow  \ottsym{(}  \mathit{x}  \ottsym{:}  \static{U}_{{\mathrm{1}}}  \ottsym{)}  \rightarrow  \static{U}_{{\mathrm{2}}}}{%
{\ottdrulename{SCheckLam}}{}%
}}

\newcommand{\ottdefnSCheck}[1]{\begin{ottdefnblock}[#1]{$\static{Gamma}  \vdash  \static{t}  \Leftarrow  \static{U}$}{\ottcom{Static Checking}}
\ottusedrule{\ottdruleSCheckSynth{}}
\ottusedrule{\ottdruleSCheckLevel{}}
\ottusedrule{\ottdruleSCheckPi{}}
\ottusedrule{\ottdruleSCheckLam{}}
\end{ottdefnblock}}

\newcommand{\ottdruleStaticTypeNormSynth}[1]{\ottdrule[#1]{%
\ottpremise{\static{Gamma}  \vdash  \static{T}  \leadsto  \static{U}  \Rightarrow   \TypeType_{ \ottnt{i} } }%
}{
 \static{Gamma}  \vdash  \static{U}  \leadsfrom  \static{T}  : \TypeType }{%
{\ottdrulename{StaticTypeNormSynth}}{}%
}}

\newcommand{\ottdruleStaticTypeNormPi}[1]{\ottdrule[#1]{%
\ottpremise{ \static{Gamma}  \vdash  \static{U}_{{\mathrm{1}}}  \leadsfrom  \static{T}_{{\mathrm{1}}}  : \TypeType }%
\ottpremise{ \ottsym{(}  \mathit{x}  \ottsym{:}  \static{U}_{{\mathrm{1}}}  \ottsym{)}  \static{Gamma}  \vdash  \static{U}_{{\mathrm{2}}}  \leadsfrom  \static{T}_{{\mathrm{2}}}  : \TypeType }%
}{
 \static{Gamma}  \vdash  \ottsym{(}  \mathit{x}  \ottsym{:}  \static{U}_{{\mathrm{1}}}  \ottsym{)}  \rightarrow  \static{U}_{{\mathrm{2}}}  \leadsfrom  \ottsym{(}  \mathit{x}  \ottsym{:}  \static{T}_{{\mathrm{1}}}  \ottsym{)}  \rightarrow  \static{T}_{{\mathrm{2}}}  : \TypeType }{%
{\ottdrulename{StaticTypeNormPi}}{}%
}}

\newcommand{\ottdefnStaticTypeNorm}[1]{\begin{ottdefnblock}[#1]{$ \static{Gamma}  \vdash  \static{U}  \leadsfrom  \static{T}  : \TypeType $}{\ottcom{Normalization for Types with Unknown Level}}
\ottusedrule{\ottdruleStaticTypeNormSynth{}}
\ottusedrule{\ottdruleStaticTypeNormPi{}}
\end{ottdefnblock}}

\newcommand{\ottdruleSNormSynthAnn}[1]{\ottdrule[#1]{%
\ottpremise{ \static{Gamma}  \vdash  \static{U}  \leadsfrom  \static{T}  : \TypeType }%
\ottpremise{\static{Gamma}  \vdash  \static{u}  \leadsfrom\  \static{t}  \Leftarrow  \static{U}}%
}{
\static{Gamma}  \vdash  \ottsym{(}   \static{t} \dblcolon \static{T}   \ottsym{)}  \leadsto  \static{u}  \Rightarrow  \static{U}}{%
{\ottdrulename{SNormSynthAnn}}{}%
}}

\newcommand{\ottdruleSNormSynthType}[1]{\ottdrule[#1]{%
\ottpremise{\ottnt{i}  \ottsym{>}  \ottsym{0}}%
}{
\static{Gamma}  \vdash   \TypeType_{ \ottnt{i} }   \leadsto   \TypeType_{ \ottnt{i} }   \Rightarrow   \TypeType_{ \ottnt{i}  + 1} }{%
{\ottdrulename{SNormSynthType}}{}%
}}

\newcommand{\ottdruleSNormSynthVar}[1]{\ottdrule[#1]{%
\ottpremise{ \vdash  \static{Gamma} }%
\ottpremise{\ottsym{(}  \mathit{x}  \ottsym{:}  \static{U}  \ottsym{)} \, \in \, \static{Gamma}}%
\ottpremise{  \mathit{x}  \leadsto_\eta  \static{u} :  \static{U} }%
}{
\static{Gamma}  \vdash  \mathit{x}  \leadsto  \static{u}  \Rightarrow  \static{U}}{%
{\ottdrulename{SNormSynthVar}}{}%
}}

\newcommand{\ottdruleSNormSynthApp}[1]{\ottdrule[#1]{%
\ottpremise{\static{Gamma}  \vdash  \static{t}_{{\mathrm{1}}}  \leadsto  \ottsym{(}  \lambda  \mathit{x}  \ldotp  \static{u}_{{\mathrm{1}}}  \ottsym{)}  \Rightarrow  \ottsym{(}  \mathit{x}  \ottsym{:}  \static{U}_{{\mathrm{1}}}  \ottsym{)}  \rightarrow  \static{U}_{{\mathrm{2}}}}%
\ottpremise{\static{Gamma}  \vdash  \static{u}_{{\mathrm{2}}}  \leadsfrom\  \static{t}_{{\mathrm{2}}}  \Leftarrow  \static{U}_{{\mathrm{1}}}}%
\ottpremise{  {[  \static{u}_{{\mathrm{2}}}  / { \mathit{x} } ]}^{ \static{U}_{{\mathrm{1}}}  }  \static{u}_{{\mathrm{1}}} =  \static{u}_{{\mathrm{3}}}   \qquad   {[  \static{u}_{{\mathrm{2}}}  / { \mathit{x} } ]}^{ \static{U}_{{\mathrm{1}}}  }  \static{U}_{{\mathrm{2}}} =  \static{U}_{{\mathrm{3}}}  }%
}{
\static{Gamma}  \vdash   \static{t}_{{\mathrm{1}}} \  \static{t}_{{\mathrm{2}}}   \leadsto  \static{u}_{{\mathrm{3}}}  \Rightarrow  \static{U}_{{\mathrm{3}}}}{%
{\ottdrulename{SNormSynthApp}}{}%
}}

\newcommand{\ottdefnSNormSynth}[1]{\begin{ottdefnblock}[#1]{$\static{Gamma}  \vdash  \static{t}  \leadsto  \static{u}  \Rightarrow  \static{U}$}{\ottcom{Static Normalization Synthesis}}
\ottusedrule{\ottdruleSNormSynthAnn{}}
\ottusedrule{\ottdruleSNormSynthType{}}
\ottusedrule{\ottdruleSNormSynthVar{}}
\ottusedrule{\ottdruleSNormSynthApp{}}
\end{ottdefnblock}}

\newcommand{\ottdruleSNormCheckSynth}[1]{\ottdrule[#1]{%
\ottpremise{\static{Gamma}  \vdash  \static{t}  \leadsto  \static{u}  \Rightarrow  \static{U}}%
}{
\static{Gamma}  \vdash  \static{u}  \leadsfrom\  \static{t}  \Leftarrow  \static{U}}{%
{\ottdrulename{SNormCheckSynth}}{}%
}}

\newcommand{\ottdruleSNormCheckLevel}[1]{\ottdrule[#1]{%
\ottpremise{\static{Gamma}  \vdash  \static{T}  \leadsto  \static{U}  \Rightarrow   \TypeType_{ \ottnt{i} } }%
\ottpremise{ 0 <  \ottnt{i}  <  \ottnt{j} }%
}{
\static{Gamma}  \vdash  \static{U}  \leadsfrom\  \static{T}  \Leftarrow   \TypeType_{ \ottnt{j} } }{%
{\ottdrulename{SNormCheckLevel}}{}%
}}

\newcommand{\ottdruleSNormCheckPi}[1]{\ottdrule[#1]{%
\ottpremise{\static{Gamma}  \vdash  \static{U}_{{\mathrm{1}}}  \leadsfrom\  \static{T}_{{\mathrm{1}}}  \Leftarrow   \TypeType_{ \ottnt{i} } }%
\ottpremise{ \vdash  \ottsym{(}  \mathit{x}  \ottsym{:}  \static{U}_{{\mathrm{1}}}  \ottsym{)}  \static{Gamma} }%
\ottpremise{\ottsym{(}  \mathit{x}  \ottsym{:}  \static{U}_{{\mathrm{1}}}  \ottsym{)}  \static{Gamma}  \vdash  \static{U}_{{\mathrm{2}}}  \leadsfrom\  \static{T}_{{\mathrm{2}}}  \Leftarrow   \TypeType_{ \ottnt{i} } }%
}{
\static{Gamma}  \vdash  \ottsym{(}  \mathit{x}  \ottsym{:}  \static{U}_{{\mathrm{1}}}  \ottsym{)}  \rightarrow  \static{U}_{{\mathrm{2}}}  \leadsfrom\  \ottsym{(}  \mathit{x}  \ottsym{:}  \static{T}_{{\mathrm{1}}}  \ottsym{)}  \rightarrow  \static{T}_{{\mathrm{2}}}  \Leftarrow   \TypeType_{ \ottnt{i} } }{%
{\ottdrulename{SNormCheckPi}}{}%
}}

\newcommand{\ottdruleSNormCheckLam}[1]{\ottdrule[#1]{%
\ottpremise{ \vdash  \ottsym{(}  \mathit{x}  \ottsym{:}  \static{U}_{{\mathrm{1}}}  \ottsym{)}  \static{Gamma} }%
\ottpremise{\ottsym{(}  \mathit{x}  \ottsym{:}  \static{U}_{{\mathrm{1}}}  \ottsym{)}  \static{Gamma}  \vdash  \static{u}  \leadsfrom\  \static{t}  \Leftarrow  \static{U}_{{\mathrm{2}}}}%
}{
\static{Gamma}  \vdash  \ottsym{(}  \lambda  \mathit{x}  \ldotp  \static{u}  \ottsym{)}  \leadsfrom\  \ottsym{(}  \lambda  \mathit{x}  \ldotp  \static{t}  \ottsym{)}  \Leftarrow  \ottsym{(}  \mathit{x}  \ottsym{:}  \static{U}_{{\mathrm{1}}}  \ottsym{)}  \rightarrow  \static{U}_{{\mathrm{2}}}}{%
{\ottdrulename{SNormCheckLam}}{}%
}}

\newcommand{\ottdefnSNormCheck}[1]{\begin{ottdefnblock}[#1]{$\static{Gamma}  \vdash  \static{u}  \leadsfrom\  \static{t}  \Leftarrow  \static{U}$}{\ottcom{Static Normalization Checking}}
\ottusedrule{\ottdruleSNormCheckSynth{}}
\ottusedrule{\ottdruleSNormCheckLevel{}}
\ottusedrule{\ottdruleSNormCheckPi{}}
\ottusedrule{\ottdruleSNormCheckLam{}}
\end{ottdefnblock}}

\newcommand{\ottdruleEtaExpandAtomic}[1]{\ottdrule[#1]{%
}{
 \static{rr} \leadsto_\eta  \static{rr} :  \static{RR} }{%
{\ottdrulename{EtaExpandAtomic}}{}%
}}

\newcommand{\ottdruleEtaExpandPi}[1]{\ottdrule[#1]{%
\ottpremise{   \mathit{y}  \leadsto_\eta  \static{u} :  \static{U}_{{\mathrm{1}}}   \qquad    \mathit{x}  \static{e} \  \static{u}   \leadsto_\eta  \static{u} :  \static{U}_{{\mathrm{2}}}  }%
}{
  \mathit{x} \static{e}  \leadsto_\eta  \ottsym{(}  \lambda  \mathit{y}  \ldotp  \static{u}  \ottsym{)} :  \ottsym{(}  \mathit{y}  \ottsym{:}  \static{U}_{{\mathrm{1}}}  \ottsym{)}  \rightarrow  \static{U}_{{\mathrm{2}}} }{%
{\ottdrulename{EtaExpandPi}}{}%
}}

\newcommand{\ottdefnEtaExpand}[1]{\begin{ottdefnblock}[#1]{$ \static{rr} \leadsto_\eta  \static{u} :  \static{U} $}{\ottcom{Eta Expansion}}
\ottusedrule{\ottdruleEtaExpandAtomic{}}
\ottusedrule{\ottdruleEtaExpandPi{}}
\end{ottdefnblock}}

\newcommand{\ottdruleGradualEnvSubEmpty}[1]{\ottdrule[#1]{%
}{
 {[  \gradual{u}  / { \mathit{x} } ]}^{ \gradual{U}  }  \cdot  =  \cdot }{%
{\ottdrulename{GradualEnvSubEmpty}}{}%
}}

\newcommand{\ottdruleGradualEnvSubVarDiff}[1]{\ottdrule[#1]{%
\ottpremise{ \mathit{x}   \neq   \mathit{y} }%
\ottpremise{ {[  \gradual{u}  / { \mathit{x} } ]}^{ \gradual{U}  }  \gradual{U}_{{\mathrm{1}}}  =  \gradual{U}_{{\mathrm{2}}} }%
\ottpremise{ {[  \gradual{u}  / { \mathit{x} } ]}^{ \gradual{U}  }  \gradual{Gamma}  =  \gradual{Gamma}' }%
}{
 {[  \gradual{u}  / { \mathit{x} } ]}^{ \gradual{U}  }  \ottsym{(}  \mathit{y}  \ottsym{:}  \gradual{U}_{{\mathrm{1}}}  \ottsym{)}  \gradual{Gamma}  =  \ottsym{(}  \mathit{y}  \ottsym{:}  \gradual{U}_{{\mathrm{2}}}  \ottsym{)}  \gradual{Gamma}' }{%
{\ottdrulename{GradualEnvSubVarDiff}}{}%
}}

\newcommand{\ottdruleGradualEnvSubVarSame}[1]{\ottdrule[#1]{%
\ottpremise{ {[  \gradual{u}  / { \mathit{x} } ]}^{ \gradual{U}  }  \gradual{Gamma}  =  \gradual{Gamma}' }%
}{
 {[  \gradual{u}  / { \mathit{x} } ]}^{ \gradual{U}  }  \ottsym{(}  \mathit{x}  \ottsym{:}  \gradual{U}'  \ottsym{)}  \gradual{Gamma}  =  \gradual{Gamma}' }{%
{\ottdrulename{GradualEnvSubVarSame}}{}%
}}

\newcommand{\ottdefnGradualEnvSub}[1]{\begin{ottdefnblock}[#1]{$ {[  \gradual{u}  / { \mathit{x} } ]}^{ \gradual{U}  }  \gradual{Gamma}  =  \gradual{Gamma}' $}{\ottcom{Approximate Substitution on Gradual Environments}}
\ottusedrule{\ottdruleGradualEnvSubEmpty{}}
\ottusedrule{\ottdruleGradualEnvSubVarDiff{}}
\ottusedrule{\ottdruleGradualEnvSubVarSame{}}
\end{ottdefnblock}}

\newcommand{\ottdruleDomainPi}[1]{\ottdrule[#1]{%
}{
 \ottkw{dom}\  \ottsym{(}  \mathit{x}  \ottsym{:}  \gradual{U}_{{\mathrm{1}}}  \ottsym{)}  \rightarrow  \gradual{U}_{{\mathrm{2}}} =  \gradual{U}_{{\mathrm{1}}} }{%
{\ottdrulename{DomainPi}}{}%
}}

\newcommand{\ottdruleDomainDyn}[1]{\ottdrule[#1]{%
}{
 \ottkw{dom}\  {\qm } =  {\qm } }{%
{\ottdrulename{DomainDyn}}{}%
}}

\newcommand{\ottdefnDomain}[1]{\begin{ottdefnblock}[#1]{$ \ottkw{dom}\  \gradual{U}_{{\mathrm{1}}} =  \gradual{U}_{{\mathrm{2}}} $}{\ottcom{Gradual Domain}}
\ottusedrule{\ottdruleDomainPi{}}
\ottusedrule{\ottdruleDomainDyn{}}
\end{ottdefnblock}}

\newcommand{\ottdruleGHsubType}[1]{\ottdrule[#1]{%
}{
 {[  \gradual{u}  / { \mathit{x} } ]}^{ \gradual{U}  }   \TypeType_{ \ottnt{i} }   =   \TypeType_{ \ottnt{i} }  }{%
{\ottdrulename{GHsubType}}{}%
}}

\newcommand{\ottdruleGHsubDyn}[1]{\ottdrule[#1]{%
}{
 {[  \gradual{u}  / { \mathit{x} } ]}^{ \gradual{U}  }  {\qm }  =  {\qm } }{%
{\ottdrulename{GHsubDyn}}{}%
}}

\newcommand{\ottdruleGHsubPi}[1]{\ottdrule[#1]{%
\ottpremise{ \mathit{x}   \neq   \mathit{y} }%
\ottpremise{ {[  \gradual{u}  / { \mathit{x} } ]}^{ \gradual{U}  }  \gradual{U}_{{\mathrm{1}}}  =  \gradual{U}'_{{\mathrm{1}}} }%
\ottpremise{ {[  \gradual{u}  / { \mathit{x} } ]}^{ \gradual{U}  }  \gradual{U}_{{\mathrm{2}}}  =  \gradual{U}'_{{\mathrm{2}}} }%
}{
 {[  \gradual{u}  / { \mathit{x} } ]}^{ \gradual{U}  }  \ottsym{(}  \mathit{y}  \ottsym{:}  \gradual{U}_{{\mathrm{1}}}  \ottsym{)}  \rightarrow  \gradual{U}_{{\mathrm{2}}}  =  \ottsym{(}  \mathit{y}  \ottsym{:}  \gradual{U}'_{{\mathrm{1}}}  \ottsym{)}  \rightarrow  \gradual{U}'_{{\mathrm{2}}} }{%
{\ottdrulename{GHsubPi}}{}%
}}

\newcommand{\ottdruleGHsubPiRdxAlpha}[1]{}

\newcommand{\ottdruleGHsubLam}[1]{\ottdrule[#1]{%
\ottpremise{ \mathit{x}   \neq   \mathit{y} }%
\ottpremise{ {[  \gradual{u}  / { \mathit{x} } ]}^{ \gradual{U}  }  \gradual{u}_{{\mathrm{2}}}  =  \gradual{u}_{{\mathrm{3}}} }%
}{
 {[  \gradual{u}  / { \mathit{x} } ]}^{ \gradual{U}  }  \ottsym{(}  \lambda  \mathit{y}  \ldotp  \gradual{u}_{{\mathrm{2}}}  \ottsym{)}  =  \ottsym{(}  \lambda  \mathit{y}  \ldotp  \gradual{u}_{{\mathrm{3}}}  \ottsym{)} }{%
{\ottdrulename{GHsubLam}}{}%
}}

\newcommand{\ottdruleGHsubLamRdxAlpha}[1]{}

\newcommand{\ottdruleGHsubDiffNil}[1]{\ottdrule[#1]{%
\ottpremise{ \mathit{x}   \neq   \mathit{y} }%
}{
 {[  \gradual{u}  / { \mathit{x} } ]}^{ \gradual{U}  }   \mathit{y}   =   \mathit{y}  }{%
{\ottdrulename{GHsubDiffNil}}{}%
}}

\newcommand{\ottdruleGHsubDiffCons}[1]{\ottdrule[#1]{%
\ottpremise{ \mathit{x}   \neq   \mathit{y} }%
\ottpremise{ {[  \gradual{u}  / { \mathit{x} } ]}^{ \gradual{U}  }   \mathit{y} \gradual{e}   =   \mathit{y} \gradual{e}'  }%
\ottpremise{ {[  \gradual{u}  / { \mathit{x} } ]}^{ \gradual{U}  }  \gradual{u}_{{\mathrm{2}}}  =  \gradual{u}_{{\mathrm{3}}} }%
}{
 {[  \gradual{u}  / { \mathit{x} } ]}^{ \gradual{U}  }   \mathit{y}  \gradual{e} \  \gradual{u}_{{\mathrm{2}}}    =   \mathit{y}  \gradual{e}' \  \gradual{u}_{{\mathrm{3}}}   }{%
{\ottdrulename{GHsubDiffCons}}{}%
}}

\newcommand{\ottdruleGHsubSpine}[1]{\ottdrule[#1]{%
\ottpremise{ {[  \gradual{u}  / { \mathit{x} } ]}^{ \gradual{U}  }  \mathit{x} \gradual{e}  = { \gradual{u}_{{\mathrm{2}}} } : { \gradual{U}_{{\mathrm{2}}} } }%
}{
 {[  \gradual{u}  / { \mathit{x} } ]}^{ \gradual{U}  }   \mathit{x} \gradual{e}   =  \gradual{u}_{{\mathrm{2}}} }{%
{\ottdrulename{GHsubSpine}}{}%
}}

\newcommand{\ottdefnGHsub}[1]{\begin{ottdefnblock}[#1]{$ {[  \gradual{u}_{{\mathrm{1}}}  / { \mathit{x} } ]}^{ \gradual{U}  }  \gradual{u}_{{\mathrm{2}}}  =  \gradual{u}_{{\mathrm{3}}} $}{\ottcom{Approximate Hereditary Substitution}}
\ottusedrule{\ottdruleGHsubType{}}
\ottusedrule{\ottdruleGHsubDyn{}}
\ottusedrule{\ottdruleGHsubPi{}}
\ottusedrule{\ottdruleGHsubPiRdxAlpha{}}
\ottusedrule{\ottdruleGHsubLam{}}
\ottusedrule{\ottdruleGHsubLamRdxAlpha{}}
\ottusedrule{\ottdruleGHsubDiffNil{}}
\ottusedrule{\ottdruleGHsubDiffCons{}}
\ottusedrule{\ottdruleGHsubSpine{}}
\end{ottdefnblock}}

\newcommand{\ottdruleGHsubRHead}[1]{\ottdrule[#1]{%
}{
 {[  \gradual{u}  / { \mathit{x} } ]}^{ \gradual{U}  }  \mathit{x}  \emptyspine   = { \gradual{u} } : { \gradual{U} } }{%
{\ottdrulename{GHsubRHead}}{}%
}}

\newcommand{\ottdruleGHsubRDynSpine}[1]{\ottdrule[#1]{%
\ottpremise{  {[  \gradual{u}  / { \mathit{x} } ]}^{ \gradual{U}  }  \mathit{x} \gradual{e}  = { {\qm } } : { \ottsym{(}  \mathit{y}  \ottsym{:}  \gradual{U}_{{\mathrm{1}}}  \ottsym{)}  \rightarrow  \gradual{U}_{{\mathrm{2}}} }   \qquad   {[  \gradual{u}_{{\mathrm{2}}}  / { \mathit{y} } ]}^{ \gradual{U}_{{\mathrm{1}}}  }  \gradual{U}_{{\mathrm{2}}}  =  \gradual{U}_{{\mathrm{3}}}  }%
}{
 {[  \gradual{u}  / { \mathit{x} } ]}^{ \gradual{U}  }  \mathit{x}  \gradual{e} \  \gradual{u}_{{\mathrm{2}}}   = { {\qm } } : { \gradual{U}_{{\mathrm{3}}} } }{%
{\ottdrulename{GHsubRDynSpine}}{}%
}}

\newcommand{\ottdruleGHsubRLamSpine}[1]{\ottdrule[#1]{%
\ottpremise{   {[  \gradual{u}  / { \mathit{x} } ]}^{ \gradual{U}  }  \mathit{x} \gradual{e}  = { \ottsym{(}  \lambda  \mathit{y}  \ldotp  \gradual{u}_{{\mathrm{2}}}  \ottsym{)} } : { \ottsym{(}  \mathit{y}  \ottsym{:}  \gradual{U}_{{\mathrm{1}}}  \ottsym{)}  \rightarrow  \gradual{U}_{{\mathrm{2}}} }   \qquad   {[  \gradual{u}  / { \mathit{x} } ]}^{ \gradual{U}  }  \gradual{u}_{{\mathrm{1}}}  =  \gradual{u}_{{\mathrm{3}}}    \qquad   \gradual{U}_{{\mathrm{1}}}  \prec  \gradual{U}  }%
\ottpremise{   {[  \gradual{u}_{{\mathrm{3}}}  / { \mathit{y} } ]}^{ \gradual{U}_{{\mathrm{1}}}  }  \gradual{u}_{{\mathrm{2}}}  =  \gradual{u}_{{\mathrm{4}}}   \qquad   {[  \gradual{u}_{{\mathrm{3}}}  / { \mathit{y} } ]}^{ \gradual{U}_{{\mathrm{1}}}  }  \gradual{U}_{{\mathrm{2}}}  =  \gradual{U}_{{\mathrm{3}}}    \qquad   \gradual{u}_{{\mathrm{4}}} \leadsto_\eta  \gradual{u}_{{\mathrm{5}}} :  \gradual{U}_{{\mathrm{3}}}  }%
}{
 {[  \gradual{u}  / { \mathit{x} } ]}^{ \gradual{U}  }  \mathit{x}  \gradual{e} \  \gradual{u}_{{\mathrm{1}}}   = { \gradual{u}_{{\mathrm{5}}} } : { \gradual{U}_{{\mathrm{3}}} } }{%
{\ottdrulename{GHsubRLamSpine}}{}%
}}

\newcommand{\ottdruleGHsubRLamSpineOrd}[1]{\ottdrule[#1]{%
\ottpremise{  {[  \gradual{u}  / { \mathit{x} } ]}^{ \gradual{U}  }  \mathit{x} \gradual{e}  = { \ottsym{(}  \lambda  \mathit{y}  \ldotp  \gradual{u}_{{\mathrm{2}}}  \ottsym{)} } : { \ottsym{(}  \mathit{y}  \ottsym{:}  \gradual{U}_{{\mathrm{1}}}  \ottsym{)}  \rightarrow  \gradual{U}_{{\mathrm{2}}} }   \qquad   \gradual{U}_{{\mathrm{1}}}  \not\prec  \gradual{U}  }%
}{
 {[  \gradual{u}  / { \mathit{x} } ]}^{ \gradual{U}  }  \mathit{x}  \gradual{e} \  \gradual{u}_{{\mathrm{1}}}   = { {\qm } } : { {\qm } } }{%
{\ottdrulename{GHsubRLamSpineOrd}}{}%
}}

\newcommand{\ottdruleGHsubRDynType}[1]{\ottdrule[#1]{%
\ottpremise{ {[  \gradual{u}_{{\mathrm{1}}}  / { \mathit{x} } ]}^{ \gradual{U}  }  \mathit{x} \gradual{e}  = { \gradual{u}_{{\mathrm{2}}} } : { {\qm } } }%
}{
 {[  \gradual{u}_{{\mathrm{1}}}  / { \mathit{x} } ]}^{ \gradual{U}  }  \mathit{x}  \gradual{e} \  \gradual{u}_{{\mathrm{2}}}   = { {\qm } } : { {\qm } } }{%
{\ottdrulename{GHsubRDynType}}{}%
}}

\newcommand{\ottdruleGHsubRLamSpineRdxAlpha}[1]{}

\newcommand{\ottdefnGHsubR}[1]{\begin{ottdefnblock}[#1]{$ {[  \gradual{u}  / { \mathit{x} } ]}^{ \gradual{U}  }  \mathit{x} \gradual{e}  = { \gradual{u}_{{\mathrm{2}}} } : { \gradual{U}_{{\mathrm{2}}} } $}{\ottcom{Approximate Atomic Hereditary Substitution}}
\ottusedrule{\ottdruleGHsubRHead{}}
\ottusedrule{\ottdruleGHsubRDynSpine{}}
\ottusedrule{\ottdruleGHsubRLamSpine{}}
\ottusedrule{\ottdruleGHsubRLamSpineOrd{}}
\ottusedrule{\ottdruleGHsubRDynType{}}
\ottusedrule{\ottdruleGHsubRLamSpineRdxAlpha{}}
\end{ottdefnblock}}

\newcommand{\ottdruleCodSubPi}[1]{\ottdrule[#1]{%
\ottpremise{ {[  \gradual{u}  / { \mathit{x} } ]}^{ \gradual{U}_{{\mathrm{1}}}  }  \gradual{U}_{{\mathrm{2}}}  =  \gradual{U}'_{{\mathrm{2}}} }%
}{
 [  \gradual{u}  / {\_} ]  \ottkw{cod}\  \ottsym{(}  \mathit{x}  \ottsym{:}  \gradual{U}_{{\mathrm{1}}}  \ottsym{)}  \rightarrow  \gradual{U}_{{\mathrm{2}}}  =  \gradual{U}'_{{\mathrm{2}}} }{%
{\ottdrulename{CodSubPi}}{}%
}}

\newcommand{\ottdruleCodSubDyn}[1]{\ottdrule[#1]{%
}{
 [  \gradual{u}  / {\_} ]  \ottkw{cod}\  {\qm }  =  {\qm } }{%
{\ottdrulename{CodSubDyn}}{}%
}}

\newcommand{\ottdefnCodSub}[1]{\begin{ottdefnblock}[#1]{$ [  \gradual{u}  / {\_} ]  \ottkw{cod}\  \gradual{U}  =  \gradual{U}' $}{\ottcom{Gradual Codomain Substitution}}
\ottusedrule{\ottdruleCodSubPi{}}
\ottusedrule{\ottdruleCodSubDyn{}}
\end{ottdefnblock}}

\newcommand{\ottdruleBodySubPi}[1]{\ottdrule[#1]{%
\ottpremise{ {[  \gradual{u}  / { \mathit{x} } ]}^{ \gradual{U}  }  \gradual{u}_{{\mathrm{2}}}  =  \gradual{u}'_{{\mathrm{2}}} }%
}{
 {[  \gradual{u}  / {\_} ]}^{ \gradual{U}  } \ottkw{body}\  \ottsym{(}  \lambda  \mathit{x}  \ldotp  \gradual{u}_{{\mathrm{2}}}  \ottsym{)} =  \gradual{u}'_{{\mathrm{2}}} }{%
{\ottdrulename{BodySubPi}}{}%
}}

\newcommand{\ottdruleBodySubDyn}[1]{\ottdrule[#1]{%
}{
 {[  \gradual{u}  / {\_} ]}^{ \gradual{U}  } \ottkw{body}\  {\qm } =  {\qm } }{%
{\ottdrulename{BodySubDyn}}{}%
}}

\newcommand{\ottdefnBodySub}[1]{\begin{ottdefnblock}[#1]{$ {[  \gradual{u}  / {\_} ]}^{ \gradual{U}  } \ottkw{body}\  \gradual{u}_{{\mathrm{2}}} =  \gradual{u}_{{\mathrm{3}}} $}{\ottcom{Gradual Function Body Substitution}}
\ottusedrule{\ottdruleBodySubPi{}}
\ottusedrule{\ottdruleBodySubDyn{}}
\end{ottdefnblock}}

\newcommand{\ottdruleGradualTypeDynTy}[1]{\ottdrule[#1]{%
\ottpremise{\gradual{Gamma}  \vdash  \gradual{rr}  \Rightarrow  {\qm }}%
}{
 \gradual{Gamma}  \vdash  \gradual{rr}  : \TypeType }{%
{\ottdrulename{GradualTypeDynTy}}{}%
}}

\newcommand{\ottdruleGradualTypeType}[1]{\ottdrule[#1]{%
\ottpremise{\gradual{Gamma}  \vdash  \gradual{rr}  \Rightarrow   \TypeType_{ \ottnt{i} } }%
}{
 \gradual{Gamma}  \vdash  \gradual{rr}  : \TypeType }{%
{\ottdrulename{GradualTypeType}}{}%
}}

\newcommand{\ottdruleGradualTypePi}[1]{\ottdrule[#1]{%
\ottpremise{ \gradual{Gamma}  \vdash  \gradual{U}_{{\mathrm{1}}}  : \TypeType }%
\ottpremise{\vdash  \ottsym{(}  \mathit{x}  \ottsym{:}  \gradual{U}_{{\mathrm{1}}}  \ottsym{)}  \gradual{Gamma}}%
\ottpremise{ \ottsym{(}  \mathit{x}  \ottsym{:}  \gradual{U}_{{\mathrm{1}}}  \ottsym{)}  \gradual{Gamma}  \vdash  \gradual{U}_{{\mathrm{2}}}  : \TypeType }%
}{
 \gradual{Gamma}  \vdash  \ottsym{(}  \mathit{x}  \ottsym{:}  \gradual{U}_{{\mathrm{1}}}  \ottsym{)}  \rightarrow  \gradual{U}_{{\mathrm{2}}}  : \TypeType }{%
{\ottdrulename{GradualTypePi}}{}%
}}

\newcommand{\ottdruleGradualTypeDynVal}[1]{\ottdrule[#1]{%
}{
 \gradual{Gamma}  \vdash  {\qm }  : \TypeType }{%
{\ottdrulename{GradualTypeDynVal}}{}%
}}

\newcommand{\ottdefnGradualType}[1]{\begin{ottdefnblock}[#1]{$ \gradual{Gamma}  \vdash  \gradual{U}  : \TypeType $}{\ottcom{Well-Formed Gradual Types with Unkown Level}}
\ottusedrule{\ottdruleGradualTypeDynTy{}}
\ottusedrule{\ottdruleGradualTypeType{}}
\ottusedrule{\ottdruleGradualTypePi{}}
\ottusedrule{\ottdruleGradualTypeDynVal{}}
\end{ottdefnblock}}

\newcommand{\ottdruleWFEmpty}[1]{\ottdrule[#1]{%
}{
 \vdash  \cdot }{%
{\ottdrulename{WFEmpty}}{}%
}}

\newcommand{\ottdruleWFExt}[1]{\ottdrule[#1]{%
\ottpremise{  \vdash  \gradual{Gamma}   \qquad   \gradual{Gamma}  \vdash  \gradual{U}  : \TypeType  }%
\ottpremise{\mathit{x} \, \not\in \, \gradual{Gamma}}%
}{
 \vdash  \ottsym{(}  \mathit{x}  \ottsym{:}  \gradual{U}  \ottsym{)}  \gradual{Gamma} }{%
{\ottdrulename{WFExt}}{}%
}}

\newcommand{\ottdefnWF}[1]{\begin{ottdefnblock}[#1]{$ \vdash  \gradual{Gamma} $}{\ottcom{Well Formed Gradual Environments}}
\ottusedrule{\ottdruleWFEmpty{}}
\ottusedrule{\ottdruleWFExt{}}
\end{ottdefnblock}}

\newcommand{\ottdruleGCSynthType}[1]{\ottdrule[#1]{%
\ottpremise{\ottnt{i}  \ottsym{>}  \ottsym{0}}%
}{
\gradual{Gamma}  \vdash   \TypeType_{ \ottnt{i} }   \Rightarrow   \TypeType_{ \ottnt{i}  + 1} }{%
{\ottdrulename{GCSynthType}}{}%
}}

\newcommand{\ottdruleGCSynthVar}[1]{\ottdrule[#1]{%
\ottpremise{\vdash  \gradual{Gamma}}%
\ottpremise{\ottsym{(}  \mathit{x}  \ottsym{:}  \gradual{U}  \ottsym{)} \, \in \, \gradual{Gamma}}%
}{
\gradual{Gamma}  \vdash   \mathit{x}   \Rightarrow  \gradual{U}}{%
{\ottdrulename{GCSynthVar}}{}%
}}

\newcommand{\ottdruleGCSynthApp}[1]{\ottdrule[#1]{%
\ottpremise{\gradual{Gamma}  \vdash   \mathit{x} \gradual{e}   \Rightarrow  \gradual{U}}%
\ottpremise{ \ottkw{dom}\  \gradual{U} =  \gradual{U}_{{\mathrm{2}}} }%
\ottpremise{\gradual{Gamma}  \vdash  \gradual{u}  \Leftarrow  \gradual{U}_{{\mathrm{2}}}}%
\ottpremise{ [  \gradual{u}  / {\_} ]  \ottkw{cod}\  \gradual{U}  =  \gradual{U}_{{\mathrm{3}}} }%
}{
\gradual{Gamma}  \vdash   \mathit{x}  \gradual{e} \  \gradual{u}    \Rightarrow  \gradual{U}_{{\mathrm{3}}}}{%
{\ottdrulename{GCSynthApp}}{}%
}}

\newcommand{\ottdefnGCSynth}[1]{\begin{ottdefnblock}[#1]{$\gradual{Gamma}  \vdash  \gradual{rr}  \Rightarrow  \gradual{U}$}{\ottcom{Well-formed Gradual Atomic Forms}}
\ottusedrule{\ottdruleGCSynthType{}}
\ottusedrule{\ottdruleGCSynthVar{}}
\ottusedrule{\ottdruleGCSynthApp{}}
\end{ottdefnblock}}

\newcommand{\ottdruleGCCheckSynth}[1]{\ottdrule[#1]{%
\ottpremise{\gradual{Gamma}  \vdash  \gradual{rr}  \Rightarrow  \gradual{U}}%
\ottpremise{ \gradual{U}  \cong  \gradual{U}' }%
}{
\gradual{Gamma}  \vdash  \gradual{rr}  \Leftarrow  \gradual{U}'}{%
{\ottdrulename{GCCheckSynth}}{}%
}}

\newcommand{\ottdruleGCCheckLevel}[1]{\ottdrule[#1]{%
\ottpremise{\gradual{Gamma}  \vdash  \gradual{RR}  \Rightarrow   \TypeType_{ \ottnt{i} } }%
\ottpremise{ 0 <  \ottnt{i}  <  \ottnt{j} }%
}{
\gradual{Gamma}  \vdash  \gradual{RR}  \Leftarrow   \TypeType_{ \ottnt{j} } }{%
{\ottdrulename{GCCheckLevel}}{}%
}}

\newcommand{\ottdruleGCCheckLamPi}[1]{\ottdrule[#1]{%
\ottpremise{\vdash  \ottsym{(}  \mathit{x}  \ottsym{:}  \gradual{U}_{{\mathrm{1}}}  \ottsym{)}  \gradual{Gamma}}%
\ottpremise{\ottsym{(}  \mathit{x}  \ottsym{:}  \gradual{U}_{{\mathrm{1}}}  \ottsym{)}  \gradual{Gamma}  \vdash  \gradual{u}  \Leftarrow  \gradual{U}_{{\mathrm{2}}}}%
}{
\gradual{Gamma}  \vdash  \ottsym{(}  \lambda  \mathit{x}  \ldotp  \gradual{u}  \ottsym{)}  \Leftarrow  \ottsym{(}  \mathit{x}  \ottsym{:}  \gradual{U}_{{\mathrm{1}}}  \ottsym{)}  \rightarrow  \gradual{U}_{{\mathrm{2}}}}{%
{\ottdrulename{GCCheckLamPi}}{}%
}}

\newcommand{\ottdruleGCCheckLamPiRdxAlpha}[1]{}

\newcommand{\ottdruleGCCheckLamDyn}[1]{\ottdrule[#1]{%
\ottpremise{\vdash  \ottsym{(}  \mathit{x}  \ottsym{:}  {\qm }  \ottsym{)}  \gradual{Gamma}}%
\ottpremise{\ottsym{(}  \mathit{x}  \ottsym{:}  {\qm }  \ottsym{)}  \gradual{Gamma}  \vdash  \gradual{u}  \Leftarrow  {\qm }}%
}{
\gradual{Gamma}  \vdash  \ottsym{(}  \lambda  \mathit{x}  \ldotp  \gradual{u}  \ottsym{)}  \Leftarrow  {\qm }}{%
{\ottdrulename{GCCheckLamDyn}}{}%
}}

\newcommand{\ottdruleGCCheckPi}[1]{\ottdrule[#1]{%
\ottpremise{ \gradual{U}_{{\mathrm{3}}}  \cong \TypeType }%
\ottpremise{\gradual{Gamma}  \vdash  \gradual{U}_{{\mathrm{1}}}  \Leftarrow  \gradual{U}_{{\mathrm{3}}}}%
\ottpremise{\vdash  \ottsym{(}  \mathit{x}  \ottsym{:}  \gradual{U}_{{\mathrm{1}}}  \ottsym{)}  \gradual{Gamma}}%
\ottpremise{\ottsym{(}  \mathit{x}  \ottsym{:}  \gradual{U}_{{\mathrm{1}}}  \ottsym{)}  \gradual{Gamma}  \vdash  \gradual{U}_{{\mathrm{2}}}  \Leftarrow  \gradual{U}_{{\mathrm{3}}}}%
}{
\gradual{Gamma}  \vdash  \ottsym{(}  \mathit{x}  \ottsym{:}  \gradual{U}_{{\mathrm{1}}}  \ottsym{)}  \rightarrow  \gradual{U}_{{\mathrm{2}}}  \Leftarrow  \gradual{U}_{{\mathrm{3}}}}{%
{\ottdrulename{GCCheckPi}}{}%
}}

\newcommand{\ottdruleGCCheckDyn}[1]{\ottdrule[#1]{%
\ottpremise{ \gradual{Gamma}  \vdash  \gradual{U}  : \TypeType }%
}{
\gradual{Gamma}  \vdash  {\qm }  \Leftarrow  \gradual{U}}{%
{\ottdrulename{GCCheckDyn}}{}%
}}

\newcommand{\ottdefnGCCheck}[1]{\begin{ottdefnblock}[#1]{$\gradual{Gamma}  \vdash  \gradual{u}  \Leftarrow  \gradual{U}$}{\ottcom{Well-formed Gradual Canonical Forms}}
\ottusedrule{\ottdruleGCCheckSynth{}}
\ottusedrule{\ottdruleGCCheckLevel{}}
\ottusedrule{\ottdruleGCCheckLamPi{}}
\ottusedrule{\ottdruleGCCheckLamPiRdxAlpha{}}
\ottusedrule{\ottdruleGCCheckLamDyn{}}
\ottusedrule{\ottdruleGCCheckPi{}}
\ottusedrule{\ottdruleGCCheckDyn{}}
\end{ottdefnblock}}

\newcommand{\ottdruleConsistentTypeType}[1]{\ottdrule[#1]{%
}{
  \TypeType_{ \ottnt{i} }   \cong \TypeType }{%
{\ottdrulename{ConsistentTypeType}}{}%
}}

\newcommand{\ottdruleConsistentTypeDyn}[1]{\ottdrule[#1]{%
}{
 {\qm }  \cong \TypeType }{%
{\ottdrulename{ConsistentTypeDyn}}{}%
}}

\newcommand{\ottdefnConsistentType}[1]{\begin{ottdefnblock}[#1]{$ \gradual{U}  \cong \TypeType $}{\ottcom{Type Consistency with Unknown Level}}
\ottusedrule{\ottdruleConsistentTypeType{}}
\ottusedrule{\ottdruleConsistentTypeDyn{}}
\end{ottdefnblock}}

\newcommand{\ottdruleConsistentEq}[1]{\ottdrule[#1]{%
}{
 \gradual{u}  \cong  \gradual{u} }{%
{\ottdrulename{ConsistentEq}}{}%
}}

\newcommand{\ottdruleConsistentPi}[1]{\ottdrule[#1]{%
\ottpremise{ \gradual{U}_{{\mathrm{1}}}  \cong  \gradual{U}'_{{\mathrm{1}}} }%
\ottpremise{ \gradual{U}_{{\mathrm{2}}}  \cong  \gradual{U}'_{{\mathrm{2}}} }%
}{
 \ottsym{(}  \mathit{x}  \ottsym{:}  \gradual{U}_{{\mathrm{1}}}  \ottsym{)}  \rightarrow  \gradual{U}_{{\mathrm{2}}}  \cong  \ottsym{(}  \mathit{x}  \ottsym{:}  \gradual{U}'_{{\mathrm{1}}}  \ottsym{)}  \rightarrow  \gradual{U}'_{{\mathrm{2}}} }{%
{\ottdrulename{ConsistentPi}}{}%
}}

\newcommand{\ottdruleConsistentPiRdxAlpha}[1]{}

\newcommand{\ottdruleConsistentLam}[1]{\ottdrule[#1]{%
\ottpremise{ \gradual{u}  \cong  \gradual{u}' }%
}{
 \ottsym{(}  \lambda  \mathit{x}  \ldotp  \gradual{u}  \ottsym{)}  \cong  \ottsym{(}  \lambda  \mathit{x}  \ldotp  \gradual{u}'  \ottsym{)} }{%
{\ottdrulename{ConsistentLam}}{}%
}}

\newcommand{\ottdruleConsistentLamRdxAlpha}[1]{}

\newcommand{\ottdruleConsistentApp}[1]{\ottdrule[#1]{%
\ottpremise{  \mathit{x} \gradual{e}   \cong   \mathit{x} \gradual{e}'  }%
\ottpremise{ \gradual{u}  \cong  \gradual{u}' }%
}{
  \mathit{x}  \gradual{e} \  \gradual{u}    \cong   \mathit{x}  \gradual{e}' \  \gradual{u}'   }{%
{\ottdrulename{ConsistentApp}}{}%
}}

\newcommand{\ottdruleConsistentDynL}[1]{\ottdrule[#1]{%
}{
 {\qm }  \cong  \gradual{u} }{%
{\ottdrulename{ConsistentDynL}}{}%
}}

\newcommand{\ottdruleConsistentDynR}[1]{\ottdrule[#1]{%
}{
 \gradual{u}  \cong  {\qm } }{%
{\ottdrulename{ConsistentDynR}}{}%
}}

\newcommand{\ottdefnConsistent}[1]{\begin{ottdefnblock}[#1]{$ \gradual{U}  \cong  \gradual{U}' $}{\ottcom{Consistency of Gradual Canonical Terms}}
\ottusedrule{\ottdruleConsistentEq{}}
\ottusedrule{\ottdruleConsistentPi{}}
\ottusedrule{\ottdruleConsistentPiRdxAlpha{}}
\ottusedrule{\ottdruleConsistentLam{}}
\ottusedrule{\ottdruleConsistentLamRdxAlpha{}}
\ottusedrule{\ottdruleConsistentApp{}}
\ottusedrule{\ottdruleConsistentDynL{}}
\ottusedrule{\ottdruleConsistentDynR{}}
\end{ottdefnblock}}

\newcommand{\ottdruleMeetPi}[1]{\ottdrule[#1]{%
\ottpremise{ \gradual{U}_{{\mathrm{1}}} \sqcap  \gradual{U}'_{{\mathrm{1}}}  =  \gradual{U}''_{{\mathrm{1}}} }%
\ottpremise{ \gradual{U}_{{\mathrm{2}}} \sqcap  \gradual{U}'_{{\mathrm{2}}}  =  \gradual{U}''_{{\mathrm{2}}} }%
}{
 \ottsym{(}  \mathit{x}  \ottsym{:}  \gradual{U}_{{\mathrm{1}}}  \ottsym{)}  \rightarrow  \gradual{U}_{{\mathrm{2}}} \sqcap  \ottsym{(}  \mathit{x}  \ottsym{:}  \gradual{U}'_{{\mathrm{1}}}  \ottsym{)}  \rightarrow  \gradual{U}'_{{\mathrm{2}}}  =  \ottsym{(}  \mathit{x}  \ottsym{:}  \gradual{U}''_{{\mathrm{1}}}  \ottsym{)}  \rightarrow  \gradual{U}''_{{\mathrm{2}}} }{%
{\ottdrulename{MeetPi}}{}%
}}

\newcommand{\ottdruleMeetPiRdxAlpha}[1]{}

\newcommand{\ottdruleMeetLam}[1]{\ottdrule[#1]{%
\ottpremise{ \gradual{u} \sqcap  \gradual{u}'  =  \gradual{u}'' }%
}{
 \lambda  \mathit{x}  \ldotp  \gradual{u} \sqcap  \lambda  \mathit{x}  \ldotp  \gradual{u}'  =  \lambda  \mathit{x}  \ldotp  \gradual{u}'' }{%
{\ottdrulename{MeetLam}}{}%
}}

\newcommand{\ottdruleMeetLamRdxAlpha}[1]{}

\newcommand{\ottdruleMeetApp}[1]{\ottdrule[#1]{%
\ottpremise{ \gradual{u} \sqcap  \gradual{u}'  =  \gradual{u}'' }%
\ottpremise{  \mathit{x} \gradual{e}  \sqcap   \mathit{x} \gradual{e}'   =   \mathit{x} \gradual{e}''  }%
}{
  \mathit{x}  \gradual{e} \  \gradual{u}   \sqcap   \mathit{x}  \gradual{e}' \  \gradual{u}'    =   \mathit{x}  \gradual{e}'' \  \gradual{u}''   }{%
{\ottdrulename{MeetApp}}{}%
}}

\newcommand{\ottdruleMeetDynL}[1]{\ottdrule[#1]{%
}{
 {\qm } \sqcap  \gradual{U}  =  \gradual{U} }{%
{\ottdrulename{MeetDynL}}{}%
}}

\newcommand{\ottdruleMeetDynR}[1]{\ottdrule[#1]{%
\ottpremise{\gradual{U}  \neq  {\qm }}%
}{
 \gradual{U} \sqcap  {\qm }  =  \gradual{U} }{%
{\ottdrulename{MeetDynR}}{}%
}}

\newcommand{\ottdruleMeetRefl}[1]{\ottdrule[#1]{%
\ottpremise{\gradual{u}  \neq  {\qm }}%
}{
 \gradual{u} \sqcap  \gradual{u}  =  \gradual{u} }{%
{\ottdrulename{MeetRefl}}{}%
}}

\newcommand{\ottdefnMeet}[1]{\begin{ottdefnblock}[#1]{$ \gradual{U}_{{\mathrm{1}}} \sqcap  \gradual{U}_{{\mathrm{2}}}  =  \gradual{U}_{{\mathrm{3}}} $}{\ottcom{Precision Meet of Gradual Canonical Forms}}
\ottusedrule{\ottdruleMeetPi{}}
\ottusedrule{\ottdruleMeetPiRdxAlpha{}}
\ottusedrule{\ottdruleMeetLam{}}
\ottusedrule{\ottdruleMeetLamRdxAlpha{}}
\ottusedrule{\ottdruleMeetApp{}}
\ottusedrule{\ottdruleMeetDynL{}}
\ottusedrule{\ottdruleMeetDynR{}}
\ottusedrule{\ottdruleMeetRefl{}}
\end{ottdefnblock}}

\newcommand{\ottdruleMorePrecisePi}[1]{\ottdrule[#1]{%
\ottpremise{ \gradual{U}_{{\mathrm{1}}} \sqsubseteq  \gradual{U}'_{{\mathrm{1}}} }%
\ottpremise{ \gradual{U}_{{\mathrm{2}}} \sqsubseteq  \gradual{U}'_{{\mathrm{2}}} }%
}{
 \ottsym{(}  \mathit{x}  \ottsym{:}  \gradual{U}_{{\mathrm{1}}}  \ottsym{)}  \rightarrow  \gradual{U}_{{\mathrm{2}}} \sqsubseteq  \ottsym{(}  \mathit{x}  \ottsym{:}  \gradual{U}'_{{\mathrm{1}}}  \ottsym{)}  \rightarrow  \gradual{U}'_{{\mathrm{2}}} }{%
{\ottdrulename{MorePrecisePi}}{}%
}}

\newcommand{\ottdruleMorePrecisePiRdxAlpha}[1]{}

\newcommand{\ottdruleMorePreciseLam}[1]{\ottdrule[#1]{%
\ottpremise{ \gradual{u} \sqsubseteq  \gradual{u}' }%
}{
 \lambda  \mathit{x}  \ldotp  \gradual{u} \sqsubseteq  \lambda  \mathit{x}  \ldotp  \gradual{u}' }{%
{\ottdrulename{MorePreciseLam}}{}%
}}

\newcommand{\ottdruleMorePreciseLamRdxAlpha}[1]{}

\newcommand{\ottdruleMorePreciseApp}[1]{\ottdrule[#1]{%
\ottpremise{ \gradual{u} \sqsubseteq  \gradual{u}' }%
\ottpremise{  \mathit{x} \gradual{e}  \sqsubseteq   \mathit{x} \gradual{e}'  }%
}{
  \mathit{x}  \gradual{e} \  \gradual{u}   \sqsubseteq   \mathit{x}  \gradual{e}' \  \gradual{u}'   }{%
{\ottdrulename{MorePreciseApp}}{}%
}}

\newcommand{\ottdruleMorePreciseDyn}[1]{\ottdrule[#1]{%
}{
 \gradual{U} \sqsubseteq  {\qm } }{%
{\ottdrulename{MorePreciseDyn}}{}%
}}

\newcommand{\ottdruleMorePreciseRefl}[1]{\ottdrule[#1]{%
\ottpremise{\gradual{U}  \neq  {\qm }}%
}{
 \gradual{U} \sqsubseteq  \gradual{U} }{%
{\ottdrulename{MorePreciseRefl}}{}%
}}

\newcommand{\ottdefnMorePrecise}[1]{\begin{ottdefnblock}[#1]{$ \gradual{U} \sqsubseteq  \gradual{U}' $}{\ottcom{Precision of Canonical Forms}}
\ottusedrule{\ottdruleMorePrecisePi{}}
\ottusedrule{\ottdruleMorePrecisePiRdxAlpha{}}
\ottusedrule{\ottdruleMorePreciseLam{}}
\ottusedrule{\ottdruleMorePreciseLamRdxAlpha{}}
\ottusedrule{\ottdruleMorePreciseApp{}}
\ottusedrule{\ottdruleMorePreciseDyn{}}
\ottusedrule{\ottdruleMorePreciseRefl{}}
\end{ottdefnblock}}

\newcommand{\ottdruleEvConsistentDef}[1]{\ottdrule[#1]{%
\ottpremise{ \gradual{U}_{{\mathrm{1}}} \sqcap  \gradual{U}_{{\mathrm{2}}}  =  \gradual{U}_{{\mathrm{3}}} }%
\ottpremise{ \gradual{U} \sqsubseteq  \gradual{U}_{{\mathrm{3}}} }%
}{
 \langle  \gradual{U}  \rangle \vdash  \gradual{U}_{{\mathrm{1}}}  \cong  \gradual{U}_{{\mathrm{2}}} }{%
{\ottdrulename{EvConsistentDef}}{}%
}}

\newcommand{\ottdefnEvConsistent}[1]{\begin{ottdefnblock}[#1]{$ \myepsilon \vdash  \gradual{U}  \cong  \gradual{U}' $}{\ottcom{Consistency Supported by Evidence}}
\ottusedrule{\ottdruleEvConsistentDef{}}
\end{ottdefnblock}}

\newcommand{\ottdruleGSynthAnn}[1]{\ottdrule[#1]{%
\ottpremise{ \gradual{Gamma}  \vdash  \gradual{U}  \leadsfrom  \gradual{T}  : \TypeType_{\Rightarrow  \ottnt{i} } }%
\ottpremise{\gradual{Gamma}  \vdash  \gradual{t}  \Leftarrow  \gradual{U}}%
}{
\gradual{Gamma}  \vdash  \ottsym{(}   \gradual{t} \dblcolon \gradual{T}   \ottsym{)}  \Rightarrow  \gradual{U}}{%
{\ottdrulename{GSynthAnn}}{}%
}}

\newcommand{\ottdruleGSynthType}[1]{\ottdrule[#1]{%
\ottpremise{\ottnt{i}  \ottsym{>}  \ottsym{0}}%
}{
\gradual{Gamma}  \vdash   \TypeType_{ \ottnt{i} }   \Rightarrow   \TypeType_{ \ottnt{i}  + 1} }{%
{\ottdrulename{GSynthType}}{}%
}}

\newcommand{\ottdruleGSynthVar}[1]{\ottdrule[#1]{%
\ottpremise{\ottsym{(}  \mathit{x}  \ottsym{:}  \gradual{U}  \ottsym{)} \, \in \, \gradual{Gamma}}%
\ottpremise{\vdash  \gradual{Gamma}}%
}{
\gradual{Gamma}  \vdash  \mathit{x}  \Rightarrow  \gradual{U}}{%
{\ottdrulename{GSynthVar}}{}%
}}

\newcommand{\ottdruleGSynthApp}[1]{\ottdrule[#1]{%
\ottpremise{ \gradual{Gamma}  \vdash  \gradual{t}_{{\mathrm{1}}}  \Rightarrow  \gradual{U}  \qquad   \gradual{Gamma}  \vdash  \gradual{u}  \leadsfrom\  \gradual{t}_{{\mathrm{2}}}  \Leftarrow   \ottkw{dom}\  \gradual{U}   }%
\ottpremise{  }%
\ottpremise{ [  \gradual{u}  / {\_} ]  \ottkw{cod}\  \gradual{U}  =  \gradual{U}_{{\mathrm{2}}} }%
}{
\gradual{Gamma}  \vdash   \gradual{t}_{{\mathrm{1}}} \  \gradual{t}_{{\mathrm{2}}}   \Rightarrow  \gradual{U}_{{\mathrm{2}}}}{%
{\ottdrulename{GSynthApp}}{}%
}}

\newcommand{\ottdruleGSynthDyn}[1]{\ottdrule[#1]{%
}{
\gradual{Gamma}  \vdash  {\qm }  \Rightarrow  {\qm }}{%
{\ottdrulename{GSynthDyn}}{}%
}}

\newcommand{\ottdefnGSynth}[1]{\begin{ottdefnblock}[#1]{$\gradual{Gamma}  \vdash  \gradual{t}  \Rightarrow  \gradual{U}$}{\ottcom{Gradual Synthesis}}
\ottusedrule{\ottdruleGSynthAnn{}}
\ottusedrule{\ottdruleGSynthType{}}
\ottusedrule{\ottdruleGSynthVar{}}
\ottusedrule{\ottdruleGSynthApp{}}
\ottusedrule{\ottdruleGSynthDyn{}}
\end{ottdefnblock}}

\newcommand{\ottdruleGCheckSynth}[1]{\ottdrule[#1]{%
\ottpremise{ \gradual{Gamma}  \vdash  \gradual{t}  \Rightarrow  \gradual{U}'  \\\\ }%
\ottpremise{ \gradual{U}'  \cong  \gradual{U} }%
}{
\gradual{Gamma}  \vdash  \gradual{t}  \Leftarrow  \gradual{U}}{%
{\ottdrulename{GCheckSynth}}{}%
}}

\newcommand{\ottdruleGCheckLevel}[1]{\ottdrule[#1]{%
\ottpremise{ \gradual{Gamma}  \vdash  \gradual{T}  \Rightarrow   \TypeType_{ \ottnt{i} }   \\\\ }%
\ottpremise{ 0 <  \ottnt{i}  <  \ottnt{j} }%
}{
\gradual{Gamma}  \vdash  \gradual{T}  \Leftarrow   \TypeType_{ \ottnt{j} } }{%
{\ottdrulename{GCheckLevel}}{}%
}}

\newcommand{\ottdruleGCheckPi}[1]{\ottdrule[#1]{%
\ottpremise{ \gradual{Gamma}  \vdash  \gradual{U}'  \leadsfrom\  \gradual{T}_{{\mathrm{1}}}  \Leftarrow  \gradual{U}  \qquad   \gradual{U}  \cong \TypeType  }%
\ottpremise{ \vdash  \ottsym{(}  \mathit{x}  \ottsym{:}  \gradual{U}'  \ottsym{)}  \gradual{Gamma}  \qquad  \ottsym{(}  \mathit{x}  \ottsym{:}  \gradual{U}'  \ottsym{)}  \gradual{Gamma}  \vdash  \gradual{T}_{{\mathrm{2}}}  \Leftarrow  \gradual{U} }%
}{
\gradual{Gamma}  \vdash  \ottsym{(}  \mathit{x}  \ottsym{:}  \gradual{T}_{{\mathrm{1}}}  \ottsym{)}  \rightarrow  \gradual{T}_{{\mathrm{2}}}  \Leftarrow  \gradual{U}}{%
{\ottdrulename{GCheckPi}}{}%
}}

\newcommand{\ottdruleGCheckLamPi}[1]{\ottdrule[#1]{%
\ottpremise{ \vdash  \ottsym{(}  \mathit{x}  \ottsym{:}  \gradual{U}_{{\mathrm{1}}}  \ottsym{)}  \gradual{Gamma}  \\\\ }%
\ottpremise{\ottsym{(}  \mathit{x}  \ottsym{:}  \gradual{U}_{{\mathrm{1}}}  \ottsym{)}  \gradual{Gamma}  \vdash  \gradual{t}  \Leftarrow  \gradual{U}_{{\mathrm{2}}}}%
}{
\gradual{Gamma}  \vdash  \ottsym{(}  \lambda  \mathit{x}  \ldotp  \gradual{t}  \ottsym{)}  \Leftarrow  \ottsym{(}  \mathit{x}  \ottsym{:}  \gradual{U}_{{\mathrm{1}}}  \ottsym{)}  \rightarrow  \gradual{U}_{{\mathrm{2}}}}{%
{\ottdrulename{GCheckLamPi}}{}%
}}

\newcommand{\ottdruleGCheckLamPiRdxAlpha}[1]{}

\newcommand{\ottdruleGCheckLamDyn}[1]{\ottdrule[#1]{%
\ottpremise{ \vdash  \ottsym{(}  \mathit{x}  \ottsym{:}  {\qm }  \ottsym{)}  \gradual{Gamma}  \\\\ }%
\ottpremise{\ottsym{(}  \mathit{x}  \ottsym{:}  {\qm }  \ottsym{)}  \gradual{Gamma}  \vdash  \gradual{t}  \Leftarrow  {\qm }}%
}{
\gradual{Gamma}  \vdash  \ottsym{(}  \lambda  \mathit{x}  \ldotp  \gradual{t}  \ottsym{)}  \Leftarrow  {\qm }}{%
{\ottdrulename{GCheckLamDyn}}{}%
}}

\newcommand{\ottdefnGCheck}[1]{\begin{ottdefnblock}[#1]{$\gradual{Gamma}  \vdash  \gradual{t}  \Leftarrow  \gradual{U}$}{\ottcom{Gradual Checking}}
\ottusedrule{\ottdruleGCheckSynth{}}
\ottusedrule{\ottdruleGCheckLevel{}}
\ottusedrule{\ottdruleGCheckPi{}}
\ottusedrule{\ottdruleGCheckLamPi{}}
\ottusedrule{\ottdruleGCheckLamPiRdxAlpha{}}
\ottusedrule{\ottdruleGCheckLamDyn{}}
\end{ottdefnblock}}

\newcommand{\ottdruleGradualTypeNormSynth}[1]{\ottdrule[#1]{%
\ottpremise{\gradual{Gamma}  \vdash  \gradual{T}  \leadsto  \gradual{U}  \Rightarrow   \TypeType_{ \ottnt{i} } }%
}{
 \gradual{Gamma}  \vdash  \gradual{U}  \leadsfrom  \gradual{T}  : \TypeType_{\Rightarrow  \ottnt{i} } }{%
{\ottdrulename{GradualTypeNormSynth}}{}%
}}

\newcommand{\ottdruleGradualTypeNormPi}[1]{\ottdrule[#1]{%
\ottpremise{ \gradual{Gamma}  \vdash  \gradual{U}_{{\mathrm{1}}}  \leadsfrom  \gradual{T}_{{\mathrm{1}}}  : \TypeType_{\Rightarrow  \ottnt{i} } }%
\ottpremise{ \ottsym{(}  \mathit{x}  \ottsym{:}  \gradual{U}_{{\mathrm{1}}}  \ottsym{)}  \gradual{Gamma}  \vdash  \gradual{U}_{{\mathrm{2}}}  \leadsfrom  \gradual{T}_{{\mathrm{2}}}  : \TypeType_{\Rightarrow  \ottnt{j} } }%
}{
 \gradual{Gamma}  \vdash   ( \mathit{x}  :  \gradual{U}_{{\mathrm{1}}} ) \ifboolexpr{bool{ShowArrowIndex} }{\xrightarrow{  \ottkw{max}( \ottnt{i} , \ottnt{j} )  } }{\to}  \gradual{U}_{{\mathrm{2}}}   \leadsfrom  \ottsym{(}  \mathit{x}  \ottsym{:}  \gradual{T}_{{\mathrm{1}}}  \ottsym{)}  \rightarrow  \gradual{T}_{{\mathrm{2}}}  : \TypeType_{\Rightarrow   \ottkw{max}( \ottnt{i} , \ottnt{j} )  } }{%
{\ottdrulename{GradualTypeNormPi}}{}%
}}

\newcommand{\ottdruleGradualTypeNormDyn}[1]{\ottdrule[#1]{%
}{
 \gradual{Gamma}  \vdash  {\qm }  \leadsfrom  {\qm }  : \TypeType_{\Rightarrow   \omega  } }{%
{\ottdrulename{GradualTypeNormDyn}}{}%
}}

\newcommand{\ottdefnGradualTypeNorm}[1]{\begin{ottdefnblock}[#1]{$ \gradual{Gamma}  \vdash  \gradual{U}  \leadsfrom  \gradual{T}  : \TypeType_{\Rightarrow  \ottnt{i} } $}{\ottcom{Approximate Normalization and Level Inference}}
\ottusedrule{\ottdruleGradualTypeNormSynth{}}
\ottusedrule{\ottdruleGradualTypeNormPi{}}
\ottusedrule{\ottdruleGradualTypeNormDyn{}}
\end{ottdefnblock}}

\newcommand{\ottdruleGEtaLongAtomic}[1]{\ottdrule[#1]{%
}{
 \gradual{rr}  :_\eta  \gradual{RR} }{%
{\ottdrulename{GEtaLongAtomic}}{}%
}}

\newcommand{\ottdruleGEtaLongAtomicDyn}[1]{\ottdrule[#1]{%
}{
 \gradual{rr}  :_\eta  {\qm } }{%
{\ottdrulename{GEtaLongAtomicDyn}}{}%
}}

\newcommand{\ottdruleGEtaLongLam}[1]{\ottdrule[#1]{%
\ottpremise{ \gradual{u}  :_\eta  \gradual{U} }%
}{
 \ottsym{(}  \lambda  \mathit{x}  \ldotp  \gradual{u}  \ottsym{)}  :_\eta  \ottsym{(}  \mathit{x}  \ottsym{:}  \gradual{U}'  \ottsym{)}  \rightarrow  \gradual{U} }{%
{\ottdrulename{GEtaLongLam}}{}%
}}

\newcommand{\ottdruleGEtaLongLamRdxAlpha}[1]{}

\newcommand{\ottdruleGEtaLongDyn}[1]{\ottdrule[#1]{%
}{
 {\qm }  :_\eta  \gradual{U} }{%
{\ottdrulename{GEtaLongDyn}}{}%
}}

\newcommand{\ottdruleGEtaLongPi}[1]{\ottdrule[#1]{%
}{
 \ottsym{(}  \mathit{x}  \ottsym{:}  \gradual{U}  \ottsym{)}  \rightarrow  \gradual{U}'  :_\eta  \gradual{U}'' }{%
{\ottdrulename{GEtaLongPi}}{}%
}}

\newcommand{\ottdefnGEtaLong}[1]{\begin{ottdefnblock}[#1]{$ \gradual{u}  :_\eta  \gradual{U} $}{\ottcom{Eta-long Canonical Forms}}
\ottusedrule{\ottdruleGEtaLongAtomic{}}
\ottusedrule{\ottdruleGEtaLongAtomicDyn{}}
\ottusedrule{\ottdruleGEtaLongLam{}}
\ottusedrule{\ottdruleGEtaLongLamRdxAlpha{}}
\ottusedrule{\ottdruleGEtaLongDyn{}}
\ottusedrule{\ottdruleGEtaLongPi{}}
\end{ottdefnblock}}

\newcommand{\ottdruleGEtaExpandAtomic}[1]{\ottdrule[#1]{%
}{
 \gradual{rr} \leadsto_\eta  \gradual{rr} :  \gradual{RR} }{%
{\ottdrulename{GEtaExpandAtomic}}{}%
}}

\newcommand{\ottdruleGEtaExpandDyn}[1]{\ottdrule[#1]{%
}{
 \gradual{rr} \leadsto_\eta  \gradual{rr} :  {\qm } }{%
{\ottdrulename{GEtaExpandDyn}}{}%
}}

\newcommand{\ottdruleGEtaExpandPi}[1]{\ottdrule[#1]{%
\ottpremise{  \mathit{y}  \leadsto_\eta  \gradual{u} :  \gradual{U}_{{\mathrm{1}}} }%
\ottpremise{  \mathit{x}  \gradual{e} \  \gradual{u}   \leadsto_\eta  \gradual{u}_{{\mathrm{2}}} :  \gradual{U}_{{\mathrm{2}}} }%
}{
  \mathit{x} \gradual{e}  \leadsto_\eta  \ottsym{(}  \lambda  \mathit{y}  \ldotp  \gradual{u}_{{\mathrm{2}}}  \ottsym{)} :  \ottsym{(}  \mathit{y}  \ottsym{:}  \gradual{U}_{{\mathrm{1}}}  \ottsym{)}  \rightarrow  \gradual{U}_{{\mathrm{2}}} }{%
{\ottdrulename{GEtaExpandPi}}{}%
}}

\newcommand{\ottdefnGEtaExpand}[1]{\begin{ottdefnblock}[#1]{$ \gradual{rr} \leadsto_\eta  \gradual{u} :  \gradual{U} $}{\ottcom{Eta-expansion of Gradual Atomic Forms}}
\ottusedrule{\ottdruleGEtaExpandAtomic{}}
\ottusedrule{\ottdruleGEtaExpandDyn{}}
\ottusedrule{\ottdruleGEtaExpandPi{}}
\end{ottdefnblock}}

\newcommand{\ottdruleGEtaExpandCAtomic}[1]{\ottdrule[#1]{%
\ottpremise{ \gradual{rr} \leadsto_\eta  \gradual{u} :  \gradual{U} }%
}{
 \gradual{rr} \leadsto_\eta  \gradual{u} :  \gradual{U} }{%
{\ottdrulename{GEtaExpandCAtomic}}{}%
}}

\newcommand{\ottdruleGEtaExpandCLam}[1]{\ottdrule[#1]{%
\ottpremise{ \gradual{u} \leadsto_\eta  \gradual{u}' :  \gradual{U}_{{\mathrm{2}}} }%
}{
 \ottsym{(}  \lambda  \mathit{x}  \ldotp  \gradual{u}  \ottsym{)} \leadsto_\eta  \ottsym{(}  \lambda  \mathit{x}  \ldotp  \gradual{u}'  \ottsym{)} :  \ottsym{(}  \mathit{x}  \ottsym{:}  \gradual{U}_{{\mathrm{1}}}  \ottsym{)}  \rightarrow  \gradual{U}_{{\mathrm{2}}} }{%
{\ottdrulename{GEtaExpandCLam}}{}%
}}

\newcommand{\ottdruleGEtaExpandCLamRdxAlpha}[1]{}

\newcommand{\ottdruleGEtaExpandCPi}[1]{\ottdrule[#1]{%
\ottpremise{ \gradual{U}_{{\mathrm{1}}} \leadsto_\eta  \gradual{U}'_{{\mathrm{1}}} :  \gradual{U}'' }%
\ottpremise{ \gradual{U}_{{\mathrm{2}}} \leadsto_\eta  \gradual{U}'_{{\mathrm{2}}} :  \gradual{U}'' }%
}{
 \ottsym{(}  \mathit{x}  \ottsym{:}  \gradual{U}_{{\mathrm{1}}}  \ottsym{)}  \rightarrow  \gradual{U}_{{\mathrm{2}}} \leadsto_\eta  \ottsym{(}  \mathit{x}  \ottsym{:}  \gradual{U}'_{{\mathrm{1}}}  \ottsym{)}  \rightarrow  \gradual{U}'_{{\mathrm{2}}} :  \gradual{U}'' }{%
{\ottdrulename{GEtaExpandCPi}}{}%
}}

\newcommand{\ottdruleGEtaExpandCDyn}[1]{\ottdrule[#1]{%
}{
 {\qm } \leadsto_\eta  {\qm } :  \gradual{U} }{%
{\ottdrulename{GEtaExpandCDyn}}{}%
}}

\newcommand{\ottdruleGEtaExpandCDynType}[1]{\ottdrule[#1]{%
}{
 \gradual{u} \leadsto_\eta  \gradual{u} :  {\qm } }{%
{\ottdrulename{GEtaExpandCDynType}}{}%
}}

\newcommand{\ottdefnGEtaExpandC}[1]{\begin{ottdefnblock}[#1]{$ \gradual{u} \leadsto_\eta  \gradual{u}' :  \gradual{U} $}{\ottcom{Eta-expansion of Gradual Canonical Forms}}
\ottusedrule{\ottdruleGEtaExpandCAtomic{}}
\ottusedrule{\ottdruleGEtaExpandCLam{}}
\ottusedrule{\ottdruleGEtaExpandCLamRdxAlpha{}}
\ottusedrule{\ottdruleGEtaExpandCPi{}}
\ottusedrule{\ottdruleGEtaExpandCDyn{}}
\ottusedrule{\ottdruleGEtaExpandCDynType{}}
\end{ottdefnblock}}

\newcommand{\ottdruleGNSynthAnn}[1]{\ottdrule[#1]{%
\ottpremise{ \gradual{Gamma}  \vdash  \gradual{U}  \leadsfrom  \gradual{T}  : \TypeType_{\Rightarrow  \ottnt{i} } }%
\ottpremise{\gradual{Gamma}  \vdash  \gradual{u}  \leadsfrom\  \gradual{t}  \Leftarrow  \gradual{U}}%
}{
\gradual{Gamma}  \vdash  \ottsym{(}   \gradual{t} \dblcolon \gradual{T}   \ottsym{)}  \leadsto  \gradual{u}  \Rightarrow  \gradual{U}}{%
{\ottdrulename{GNSynthAnn}}{}%
}}

\newcommand{\ottdruleGNSynthType}[1]{\ottdrule[#1]{%
\ottpremise{\ottnt{i}  \ottsym{>}  \ottsym{0}}%
}{
\gradual{Gamma}  \vdash   \TypeType_{ \ottnt{i} }   \leadsto   \TypeType_{ \ottnt{i} }   \Rightarrow   \TypeType_{ \ottnt{i}  + 1} }{%
{\ottdrulename{GNSynthType}}{}%
}}

\newcommand{\ottdruleGNSynthVar}[1]{\ottdrule[#1]{%
\ottpremise{  }%
\ottpremise{\vdash  \gradual{Gamma}}%
\ottpremise{\ottsym{(}  \mathit{x}  \ottsym{:}  \gradual{U}  \ottsym{)} \, \in \, \gradual{Gamma}}%
\ottpremise{  \mathit{x}  \leadsto_\eta  \gradual{u} :  \gradual{U} }%
}{
\gradual{Gamma}  \vdash  \mathit{x}  \leadsto  \gradual{u}  \Rightarrow  \gradual{U}}{%
{\ottdrulename{GNSynthVar}}{}%
}}

\newcommand{\ottdruleGNSynthVarLook}[1]{}

\newcommand{\ottdruleGNSynthApp}[1]{\ottdrule[#1]{%
\ottpremise{  }%
\ottpremise{  \gradual{Gamma}  \vdash  \gradual{t}_{{\mathrm{1}}}  \leadsto  \gradual{u}_{{\mathrm{1}}}  \Rightarrow  \gradual{U}  \qquad   \gradual{u}_{{\mathrm{1}}} \leadsto_\eta  \gradual{u}'_{{\mathrm{1}}} :   {\qm } \to {\qm }    }%
\ottpremise{  \ottkw{dom}\  \gradual{U} =  \gradual{U}_{{\mathrm{1}}}   \qquad  \gradual{Gamma}  \vdash  \gradual{u}_{{\mathrm{2}}}  \leadsfrom\  \gradual{t}_{{\mathrm{2}}}  \Leftarrow  \gradual{U}_{{\mathrm{1}}} }%
\ottpremise{ {[  \gradual{u}_{{\mathrm{2}}}  / {\_} ]}^{ \gradual{U}_{{\mathrm{1}}}  } \ottkw{body}\  \gradual{u}'_{{\mathrm{1}}} =  \gradual{u}_{{\mathrm{3}}} }%
\ottpremise{  [  \gradual{u}_{{\mathrm{2}}}  / {\_} ]  \ottkw{cod}\  \gradual{U}  =  \gradual{U}_{{\mathrm{2}}}   \qquad   \gradual{u}_{{\mathrm{3}}} \leadsto_\eta  \gradual{u}'_{{\mathrm{3}}} :  \gradual{U}_{{\mathrm{2}}}  }%
}{
\gradual{Gamma}  \vdash   \gradual{t}_{{\mathrm{1}}} \  \gradual{t}_{{\mathrm{2}}}   \leadsto  \gradual{u}'_{{\mathrm{3}}}  \Rightarrow  \gradual{U}_{{\mathrm{2}}}}{%
{\ottdrulename{GNSynthApp}}{}%
}}

\newcommand{\ottdruleGNSynthDyn}[1]{\ottdrule[#1]{%
}{
\gradual{Gamma}  \vdash  {\qm }  \leadsto  {\qm }  \Rightarrow  {\qm }}{%
{\ottdrulename{GNSynthDyn}}{}%
}}

\newcommand{\ottdruleGNSynthEv}[1]{\ottdrule[#1]{%
\ottpremise{\gradual{Gamma}  \vdash  \gradual{u}  \leadsfrom\  \evterm{t}  \Leftarrow  \gradual{U}}%
}{
\gradual{Gamma}  \vdash  \langle  \gradual{U}  \rangle \, \evterm{t}  \leadsto  \gradual{u}  \Rightarrow  \gradual{U}}{%
{\ottdrulename{GNSynthEv}}{}%
}}

\newcommand{\ottdefnGNSynth}[1]{\begin{ottdefnblock}[#1]{$\gradual{Gamma}  \vdash  \gradual{t}  \leadsto  \gradual{u}  \Rightarrow  \gradual{U}$}{\ottcom{Approximate Normalization Synthesis}}
\ottusedrule{\ottdruleGNSynthAnn{}}
\ottusedrule{\ottdruleGNSynthType{}}
\ottusedrule{\ottdruleGNSynthVar{}}
\ottusedrule{\ottdruleGNSynthVarLook{}}
\ottusedrule{\ottdruleGNSynthApp{}}
\ottusedrule{\ottdruleGNSynthDyn{}}
\ottusedrule{\ottdruleGNSynthEv{}}
\end{ottdefnblock}}

\newcommand{\ottdruleGNCheckSynthRdx}[1]{}

\newcommand{\ottdruleGNCheckLamPiRdxAlpha}[1]{}

\newcommand{\ottdruleGElabSynthAnn}[1]{\ottdrule[#1]{%
\ottpremise{ \gradual{Gamma}  \vdash  \gradual{U}  \leadsfrom  \gradual{T}  : \TypeType_{\Rightarrow  \ottnt{i} } }%
\ottpremise{ \gradual{Gamma}  \vdash  \evterm{t}  \leftarrowtriangle  \gradual{t}  \Leftarrow  \gradual{U} }%
}{
 \gradual{Gamma}  \vdash  \ottsym{(}   \gradual{t} \dblcolon \gradual{T}   \ottsym{)}  \rightarrowtriangle  \evterm{t}  \Rightarrow  \gradual{U} }{%
{\ottdrulename{GElabSynthAnn}}{}%
}}

\newcommand{\ottdruleGElabSynthType}[1]{\ottdrule[#1]{%
\ottpremise{\ottnt{i}  \ottsym{>}  \ottsym{0}}%
}{
 \gradual{Gamma}  \vdash   \TypeType_{ \ottnt{i} }   \rightarrowtriangle   \TypeType_{ \ottnt{i} }   \Rightarrow   \TypeType_{ \ottnt{i}  + 1}  }{%
{\ottdrulename{GElabSynthType}}{}%
}}

\newcommand{\ottdruleGElabSynthVar}[1]{\ottdrule[#1]{%
\ottpremise{  }%
\ottpremise{\vdash  \gradual{Gamma}}%
\ottpremise{\ottsym{(}  \mathit{x}  \ottsym{:}  \gradual{U}  \ottsym{)} \, \in \, \gradual{Gamma}}%
}{
 \gradual{Gamma}  \vdash  \mathit{x}  \rightarrowtriangle  \mathit{x}  \Rightarrow  \gradual{U} }{%
{\ottdrulename{GElabSynthVar}}{}%
}}

\newcommand{\ottdruleGElabSynthVarLook}[1]{}

\newcommand{\ottdruleGElabSynthAppPi}[1]{\ottdrule[#1]{%
\ottpremise{ \gradual{Gamma}  \vdash  \gradual{t}_{{\mathrm{1}}}  \rightarrowtriangle  \evterm{t}_{{\mathrm{1}}}  \Rightarrow  \ottsym{(}  \mathit{x}  \ottsym{:}  \gradual{U}  \ottsym{)}  \rightarrow  \gradual{U}_{{\mathrm{2}}} }%
\ottpremise{\gradual{Gamma}  \vdash  \gradual{u}  \leadsfrom\  \evterm{t}_{{\mathrm{2}}}  \ottsym{<~~}  \gradual{t}_{{\mathrm{2}}}  \ottsym{<~~=}  \gradual{U}_{{\mathrm{1}}}}%
\ottpremise{ {[  \gradual{u}  / { \mathit{x} } ]}^{ \gradual{U}_{{\mathrm{1}}}  }  \gradual{U}_{{\mathrm{2}}}  =  \gradual{U}'_{{\mathrm{2}}} }%
}{
 \gradual{Gamma}  \vdash   \gradual{t}_{{\mathrm{1}}} \  \gradual{t}_{{\mathrm{2}}}   \rightarrowtriangle   \evterm{t}_{{\mathrm{1}}} \  \evterm{t}_{{\mathrm{2}}}   \Rightarrow  \gradual{U}'_{{\mathrm{2}}} }{%
{\ottdrulename{GElabSynthAppPi}}{}%
}}

\newcommand{\ottdruleGElabSynthAppDyn}[1]{\ottdrule[#1]{%
\ottpremise{ \gradual{Gamma}  \vdash  \gradual{t}_{{\mathrm{1}}}  \rightarrowtriangle  \evterm{t}_{{\mathrm{1}}}  \Rightarrow  {\qm } }%
\ottpremise{\gradual{Gamma}  \vdash  \gradual{u}  \leadsfrom\  \evterm{t}_{{\mathrm{2}}}  \ottsym{<~~}  \gradual{t}_{{\mathrm{2}}}  \ottsym{<~~=}  {\qm }}%
}{
 \gradual{Gamma}  \vdash   \gradual{t}_{{\mathrm{1}}} \  \gradual{t}_{{\mathrm{2}}}   \rightarrowtriangle   \ottsym{(}  \langle   {\qm } \to {\qm }   \rangle \, \evterm{t}_{{\mathrm{1}}}  \ottsym{)} \  \evterm{t}_{{\mathrm{2}}}   \Rightarrow  {\qm } }{%
{\ottdrulename{GElabSynthAppDyn}}{}%
}}

\newcommand{\ottdruleGElabSynthDyn}[1]{\ottdrule[#1]{%
}{
 \gradual{Gamma}  \vdash  {\qm }  \rightarrowtriangle  \langle  {\qm }  \rangle \, {\qm }  \Rightarrow  {\qm } }{%
{\ottdrulename{GElabSynthDyn}}{}%
}}

\newcommand{\ottdefnGElabSynth}[1]{\begin{ottdefnblock}[#1]{$ \gradual{Gamma}  \vdash  \gradual{t}  \rightarrowtriangle  \evterm{t}  \Rightarrow  \gradual{U} $}{\ottcom{Gradual Synthesis Elaboration}}
\ottusedrule{\ottdruleGElabSynthAnn{}}
\ottusedrule{\ottdruleGElabSynthType{}}
\ottusedrule{\ottdruleGElabSynthVar{}}
\ottusedrule{\ottdruleGElabSynthVarLook{}}
\ottusedrule{\ottdruleGElabSynthAppPi{}}
\ottusedrule{\ottdruleGElabSynthAppDyn{}}
\ottusedrule{\ottdruleGElabSynthDyn{}}
\end{ottdefnblock}}

\newcommand{\ottdruleGElabCheckSynth}[1]{\ottdrule[#1]{%
\ottpremise{ \gradual{Gamma}  \vdash  \gradual{t}  \rightarrowtriangle  \evterm{t}  \Rightarrow  \gradual{U}' }%
\ottpremise{ \gradual{U} \sqcap  \gradual{U}'  =  \gradual{U}'' }%
}{
 \gradual{Gamma}  \vdash  \langle  \gradual{U}''  \rangle \, \evterm{t}  \leftarrowtriangle  \gradual{t}  \Leftarrow  \gradual{U} }{%
{\ottdrulename{GElabCheckSynth}}{}%
}}

\newcommand{\ottdruleGElabCheckLevel}[1]{\ottdrule[#1]{%
\ottpremise{ \gradual{Gamma}  \vdash  \gradual{T}  \rightarrowtriangle  \evterm{T}  \Rightarrow   \TypeType_{ \ottnt{i} }  }%
\ottpremise{ 0 <  \ottnt{i}  <  \ottnt{j} }%
}{
 \gradual{Gamma}  \vdash  \evterm{T}  \leftarrowtriangle  \gradual{T}  \Leftarrow   \TypeType_{ \ottnt{j} }  }{%
{\ottdrulename{GElabCheckLevel}}{}%
}}

\newcommand{\ottdruleGElabCheckPi}[1]{\ottdrule[#1]{%
\ottpremise{ \gradual{U}  \cong \TypeType }%
\ottpremise{\gradual{Gamma}  \vdash  \gradual{U}_{{\mathrm{1}}}  \leadsfrom\  \evterm{T}_{{\mathrm{1}}}  \ottsym{<~~}  \gradual{T}_{{\mathrm{1}}}  \ottsym{<~~=}  \gradual{U}}%
\ottpremise{\vdash  \ottsym{(}  \mathit{x}  \ottsym{:}  \gradual{U}_{{\mathrm{1}}}  \ottsym{)}  \gradual{Gamma}}%
\ottpremise{ \ottsym{(}  \mathit{x}  \ottsym{:}  \gradual{U}_{{\mathrm{1}}}  \ottsym{)}  \gradual{Gamma}  \vdash  \evterm{T}_{{\mathrm{2}}}  \leftarrowtriangle  \gradual{T}_{{\mathrm{2}}}  \Leftarrow  \gradual{U} }%
}{
 \gradual{Gamma}  \vdash  \ottsym{(}  \mathit{x}  \ottsym{:}  \evterm{T}_{{\mathrm{1}}}  \ottsym{)}  \rightarrow  \evterm{T}_{{\mathrm{2}}}  \leftarrowtriangle  \ottsym{(}  \mathit{x}  \ottsym{:}  \gradual{T}_{{\mathrm{1}}}  \ottsym{)}  \rightarrow  \gradual{T}_{{\mathrm{2}}}  \Leftarrow  \gradual{U} }{%
{\ottdrulename{GElabCheckPi}}{}%
}}

\newcommand{\ottdruleGElabCheckLamPi}[1]{\ottdrule[#1]{%
\ottpremise{\vdash  \ottsym{(}  \mathit{x}  \ottsym{:}  \gradual{U}_{{\mathrm{1}}}  \ottsym{)}  \gradual{Gamma}}%
\ottpremise{ \ottsym{(}  \mathit{x}  \ottsym{:}  \gradual{U}_{{\mathrm{1}}}  \ottsym{)}  \gradual{Gamma}  \vdash  \evterm{t}  \leftarrowtriangle  \gradual{t}  \Leftarrow  \gradual{U}_{{\mathrm{2}}} }%
}{
 \gradual{Gamma}  \vdash  \langle  \ottsym{(}  \mathit{x}  \ottsym{:}  \gradual{U}_{{\mathrm{1}}}  \ottsym{)}  \rightarrow  \gradual{U}_{{\mathrm{2}}}  \rangle \, \ottsym{(}  \lambda  \mathit{x}  \ldotp  \evterm{t}  \ottsym{)}  \leftarrowtriangle  \ottsym{(}  \lambda  \mathit{x}  \ldotp  \gradual{t}  \ottsym{)}  \Leftarrow  \ottsym{(}  \mathit{x}  \ottsym{:}  \gradual{U}_{{\mathrm{1}}}  \ottsym{)}  \rightarrow  \gradual{U}_{{\mathrm{2}}} }{%
{\ottdrulename{GElabCheckLamPi}}{}%
}}

\newcommand{\ottdruleGElabCheckLamPiRdxAlpha}[1]{}

\newcommand{\ottdruleGElabCheckLamDyn}[1]{\ottdrule[#1]{%
\ottpremise{\vdash  \ottsym{(}  \mathit{x}  \ottsym{:}  {\qm }  \ottsym{)}  \gradual{Gamma}}%
\ottpremise{ \ottsym{(}  \mathit{x}  \ottsym{:}  {\qm }  \ottsym{)}  \gradual{Gamma}  \vdash  \evterm{t}  \leftarrowtriangle  \gradual{t}  \Leftarrow  {\qm } }%
}{
 \gradual{Gamma}  \vdash  \langle  {\qm }  \rangle \, \ottsym{(}  \lambda  \mathit{x}  \ldotp  \evterm{t}  \ottsym{)}  \leftarrowtriangle  \ottsym{(}  \lambda  \mathit{x}  \ldotp  \gradual{t}  \ottsym{)}  \Leftarrow  {\qm } }{%
{\ottdrulename{GElabCheckLamDyn}}{}%
}}

\newcommand{\ottdefnGElabCheck}[1]{\begin{ottdefnblock}[#1]{$ \gradual{Gamma}  \vdash  \evterm{t}  \leftarrowtriangle  \gradual{t}  \Leftarrow  \gradual{U} $}{\ottcom{Gradual Checking Elaboration}}
\ottusedrule{\ottdruleGElabCheckSynth{}}
\ottusedrule{\ottdruleGElabCheckLevel{}}
\ottusedrule{\ottdruleGElabCheckPi{}}
\ottusedrule{\ottdruleGElabCheckLamPi{}}
\ottusedrule{\ottdruleGElabCheckLamPiRdxAlpha{}}
\ottusedrule{\ottdruleGElabCheckLamDyn{}}
\end{ottdefnblock}}

\newcommand{\ottdruleGradualNESynthVarLook}[1]{}

\newcommand{\ottdruleGradualNESynthSynthRdxError}[1]{}

\newcommand{\ottdruleGradualNECheckSynthRdx}[1]{}

\newcommand{\ottdruleGradualNECheckLamPiRdxAlpha}[1]{}

\newcommand{\ottdruleEvTypeVar}[1]{\ottdrule[#1]{%
\ottpremise{\vdash  \gradual{Gamma}}%
\ottpremise{\ottsym{(}  \mathit{x}  \ottsym{:}  \gradual{U}  \ottsym{)} \, \in \, \gradual{Gamma}}%
}{
\gradual{Gamma}  \vdash  \mathit{x}  \ottsym{:}  \gradual{U}}{%
{\ottdrulename{EvTypeVar}}{}%
}}

\newcommand{\ottdruleEvTypeApp}[1]{\ottdrule[#1]{%
\ottpremise{\gradual{Gamma}  \vdash  \evterm{t}_{{\mathrm{1}}}  \ottsym{:}  \ottsym{(}  \mathit{x}  \ottsym{:}  \gradual{U}_{{\mathrm{1}}}  \ottsym{)}  \rightarrow  \gradual{U}_{{\mathrm{2}}}}%
\ottpremise{\gradual{Gamma}  \vdash  \evterm{t}_{{\mathrm{2}}}  \ottsym{:}  \gradual{U}_{{\mathrm{1}}}}%
\ottpremise{\gradual{Gamma}  \vdash  \gradual{u}_{{\mathrm{2}}}  \leadsfrom\  \evterm{t}_{{\mathrm{2}}}  \Leftarrow  \gradual{U}_{{\mathrm{1}}}}%
\ottpremise{ {[  \gradual{u}_{{\mathrm{2}}}  / { \mathit{x} } ]}^{ \gradual{U}_{{\mathrm{1}}}  }  \gradual{U}_{{\mathrm{2}}}  =  \gradual{U}_{{\mathrm{3}}} }%
}{
\gradual{Gamma}  \vdash   \evterm{t}_{{\mathrm{1}}} \  \evterm{t}_{{\mathrm{2}}}   \ottsym{:}  \gradual{U}_{{\mathrm{3}}}}{%
{\ottdrulename{EvTypeApp}}{}%
}}

\newcommand{\ottdruleEvTypePi}[1]{\ottdrule[#1]{%
\ottpremise{ \gradual{U}  \cong \TypeType }%
\ottpremise{\gradual{Gamma}  \vdash  \gradual{U}'  \leadsfrom\  \evterm{T}_{{\mathrm{1}}}  \Leftarrow  \gradual{U}}%
\ottpremise{\vdash  \ottsym{(}  \mathit{x}  \ottsym{:}  \gradual{U}'  \ottsym{)}  \gradual{Gamma}}%
\ottpremise{\ottsym{(}  \mathit{x}  \ottsym{:}  \gradual{U}'  \ottsym{)}  \gradual{Gamma}  \vdash  \evterm{T}_{{\mathrm{2}}}  \ottsym{:}  \gradual{U}}%
}{
\gradual{Gamma}  \vdash  \ottsym{(}  \mathit{x}  \ottsym{:}  \evterm{T}_{{\mathrm{1}}}  \ottsym{)}  \rightarrow  \evterm{T}_{{\mathrm{2}}}  \ottsym{:}  \gradual{U}}{%
{\ottdrulename{EvTypePi}}{}%
}}

\newcommand{\ottdruleEvTypeType}[1]{\ottdrule[#1]{%
\ottpremise{\ottnt{i}  \ottsym{>}  \ottsym{0}}%
}{
\gradual{Gamma}  \vdash   \TypeType_{ \ottnt{i} }   \ottsym{:}   \TypeType_{ \ottnt{i}  + 1} }{%
{\ottdrulename{EvTypeType}}{}%
}}

\newcommand{\ottdruleEvTypeEv}[1]{\ottdrule[#1]{%
\ottpremise{\gradual{Gamma}  \vdash  \evterm{t}  \ottsym{:}  \gradual{U}'}%
\ottpremise{ \myepsilon \vdash  \gradual{U}'  \cong  \gradual{U} }%
}{
\gradual{Gamma}  \vdash  \myepsilon \, \evterm{t}  \ottsym{:}  \gradual{U}}{%
{\ottdrulename{EvTypeEv}}{}%
}}

\newcommand{\ottdruleEvTypeLevel}[1]{\ottdrule[#1]{%
\ottpremise{\gradual{Gamma}  \vdash  \evterm{T}  \ottsym{:}   \TypeType_{ \ottnt{i} } }%
\ottpremise{ 0 <  \ottnt{i}  <  \ottnt{j} }%
}{
\gradual{Gamma}  \vdash  \evterm{T}  \ottsym{:}   \TypeType_{ \ottnt{j} } }{%
{\ottdrulename{EvTypeLevel}}{}%
}}

\newcommand{\ottdruleEvTypeLam}[1]{\ottdrule[#1]{%
\ottpremise{\vdash  \ottsym{(}  \mathit{x}  \ottsym{:}  \gradual{U}_{{\mathrm{1}}}  \ottsym{)}  \gradual{Gamma}}%
\ottpremise{\ottsym{(}  \mathit{x}  \ottsym{:}  \gradual{U}_{{\mathrm{1}}}  \ottsym{)}  \gradual{Gamma}  \vdash  \evterm{t}  \ottsym{:}  \gradual{U}_{{\mathrm{2}}}}%
}{
\gradual{Gamma}  \vdash  \langle  \ottsym{(}  \mathit{x}  \ottsym{:}  \gradual{U}_{{\mathrm{1}}}  \ottsym{)}  \rightarrow  \gradual{U}_{{\mathrm{2}}}  \rangle \, \ottsym{(}  \lambda  \mathit{x}  \ldotp  \evterm{t}  \ottsym{)}  \ottsym{:}  \ottsym{(}  \mathit{x}  \ottsym{:}  \gradual{U}_{{\mathrm{1}}}  \ottsym{)}  \rightarrow  \gradual{U}_{{\mathrm{2}}}}{%
{\ottdrulename{EvTypeLam}}{}%
}}

\newcommand{\ottdruleEvTypeDyn}[1]{\ottdrule[#1]{%
\ottpremise{ \gradual{Gamma}  \vdash  \gradual{U}  : \TypeType }%
\ottpremise{ \myepsilon \vdash  \gradual{U}  \cong  \gradual{U} }%
}{
\gradual{Gamma}  \vdash  \myepsilon \, {\qm }  \ottsym{:}  \gradual{U}}{%
{\ottdrulename{EvTypeDyn}}{}%
}}

\newcommand{\ottdefnEvType}[1]{\begin{ottdefnblock}[#1]{$\gradual{Gamma}  \vdash  \evterm{t}  \ottsym{:}  \gradual{U}$}{\ottcom{Evidence Term Typing}}
\ottusedrule{\ottdruleEvTypeVar{}}
\ottusedrule{\ottdruleEvTypeApp{}}
\ottusedrule{\ottdruleEvTypePi{}}
\ottusedrule{\ottdruleEvTypeType{}}
\ottusedrule{\ottdruleEvTypeEv{}}
\ottusedrule{\ottdruleEvTypeLevel{}}
\ottusedrule{\ottdruleEvTypeLam{}}
\ottusedrule{\ottdruleEvTypeDyn{}}
\end{ottdefnblock}}

\newcommand{\ottdruleSimpleStepAnn}[1]{\ottdrule[#1]{%
}{
\ottsym{(}   \static{v} \dblcolon \static{T}   \ottsym{)}  \longrightarrow  \static{v}}{%
{\ottdrulename{SimpleStepAnn}}{}%
}}

\newcommand{\ottdruleSimpleStepApp}[1]{\ottdrule[#1]{%
}{
 \ottsym{(}  \lambda  \mathit{x}  \ldotp  \static{t}  \ottsym{)} \  \static{v}   \longrightarrow   [  \mathit{x}  \Mapsto  \static{v}  ]  \static{t} }{%
{\ottdrulename{SimpleStepApp}}{}%
}}

\newcommand{\ottdruleSimpleStepContext}[1]{\ottdrule[#1]{%
\ottpremise{\static{t}_{{\mathrm{1}}}  \longrightarrow  \static{t}_{{\mathrm{2}}}}%
}{
 \static{C}[  \static{t}_{{\mathrm{1}}}  ]   \longrightarrow   \static{C}[  \static{t}_{{\mathrm{2}}}  ] }{%
{\ottdrulename{SimpleStepContext}}{}%
}}

\newcommand{\ottdefnSimpleStep}[1]{\begin{ottdefnblock}[#1]{$\static{t}_{{\mathrm{1}}}  \longrightarrow  \static{t}_{{\mathrm{2}}}$}{\ottcom{Simple Small-Step Semantics}}
\ottusedrule{\ottdruleSimpleStepAnn{}}
\ottusedrule{\ottdruleSimpleStepApp{}}
\ottusedrule{\ottdruleSimpleStepContext{}}
\end{ottdefnblock}}

\newcommand{\ottdruleStepAscr}[1]{\ottdrule[#1]{%
\ottpremise{ \myepsilon_{{\mathrm{1}}} \sqcap  \myepsilon_{{\mathrm{2}}}  =  \myepsilon_{{\mathrm{3}}} }%
}{
\myepsilon_{{\mathrm{1}}} \, \ottsym{(}  \myepsilon_{{\mathrm{2}}} \, \evterm{rv}  \ottsym{)}  \longrightarrow  \myepsilon_{{\mathrm{3}}} \, \evterm{rv}}{%
{\ottdrulename{StepAscr}}{}%
}}

\newcommand{\ottdruleStepAscrFail}[1]{\ottdrule[#1]{%
\ottpremise{\myepsilon_{{\mathrm{1}}}  \sqcap  \myepsilon_{{\mathrm{2}}} \, \ottkw{undefined}}%
\ottpremise{  }%
}{
\myepsilon_{{\mathrm{1}}} \, \ottsym{(}  \myepsilon_{{\mathrm{2}}} \, \evterm{rv}  \ottsym{)}  \longrightarrow  \mathsf{err}}{%
{\ottdrulename{StepAscrFail}}{}%
}}

\newcommand{\ottdruleStepAppDyn}[1]{\ottdrule[#1]{%
\ottpremise{ \cdot  \vdash  \gradual{u}  \leadsfrom\  \evterm{v}  \Leftarrow   \ottkw{dom}\  \gradual{U}   \qquad   [  \gradual{u}  / {\_} ]  \ottkw{cod}\  \gradual{U}  =  \gradual{U}_{{\mathrm{2}}}  }%
}{
 \ottsym{(}  \langle  \gradual{U}  \rangle \, {\qm }  \ottsym{)} \  \evterm{v}   \longrightarrow  \langle  \gradual{U}_{{\mathrm{2}}}  \rangle \, {\qm }}{%
{\ottdrulename{StepAppDyn}}{}%
}}

\newcommand{\ottdruleStepAppEv}[1]{\ottdrule[#1]{%
\ottpremise{   \gradual{U}' \sqcap   \ottkw{dom}\  \gradual{U}   =  \gradual{U}_{{\mathrm{1}}}   \qquad  \cdot  \vdash  \gradual{u}  \leadsfrom\  \evterm{rv}  \Leftarrow  \gradual{U}_{{\mathrm{1}}}   \qquad   [  \gradual{u}  / {\_} ]  \ottkw{cod}\  \gradual{U}  =  \gradual{U}_{{\mathrm{2}}}  }%
}{
 \ottsym{(}  \langle  \gradual{U}  \rangle \, \ottsym{(}  \lambda  \mathit{x}  \ldotp  \evterm{t}  \ottsym{)}  \ottsym{)} \  \ottsym{(}  \langle  \gradual{U}'  \rangle \, \evterm{rv}  \ottsym{)}   \longrightarrow  \langle  \gradual{U}_{{\mathrm{2}}}  \rangle \, \ottsym{(}   {[  \mathit{x}  \Mapsto  \langle  \gradual{U}_{{\mathrm{1}}}  \rangle \, \evterm{rv}  ]}^{ \gradual{u}  :  \gradual{U}_{{\mathrm{1}}} }  \evterm{t}   \ottsym{)}}{%
{\ottdrulename{StepAppEv}}{}%
}}

\newcommand{\ottdruleStepAppEvRaw}[1]{\ottdrule[#1]{%
\ottpremise{ \cdot  \vdash  \gradual{u}  \leadsfrom\  \evterm{rv}  \Leftarrow   \ottkw{dom}\  \gradual{U}   \qquad   [  \gradual{u}  / {\_} ]  \ottkw{cod}\  \gradual{U}  =  \gradual{U}_{{\mathrm{2}}}  }%
}{
 \ottsym{(}  \langle  \gradual{U}  \rangle \, \ottsym{(}  \lambda  \mathit{x}  \ldotp  \evterm{t}  \ottsym{)}  \ottsym{)} \  \evterm{rv}   \longrightarrow  \langle  \gradual{U}_{{\mathrm{2}}}  \rangle \, \ottsym{(}   {[  \mathit{x}  \Mapsto  \ottsym{(}  \langle  \gradual{U}_{{\mathrm{1}}}  \rangle \, \evterm{rv}  \ottsym{)}  ]}^{ \gradual{u}  :  \gradual{U}_{{\mathrm{1}}} }  \evterm{t}   \ottsym{)}}{%
{\ottdrulename{StepAppEvRaw}}{}%
}}

\newcommand{\ottdruleStepAppFailTrans}[1]{\ottdrule[#1]{%
\ottpremise{ \ottkw{dom}\  \gradual{U}   \sqcap  \gradual{U}' \, \ottkw{undefined}}%
}{
 \ottsym{(}  \langle  \gradual{U}  \rangle \, \ottsym{(}  \lambda  \mathit{x}  \ldotp  \evterm{t}  \ottsym{)}  \ottsym{)} \  \ottsym{(}  \langle  \gradual{U}'  \rangle \, \evterm{rv}  \ottsym{)}   \longrightarrow  \mathsf{err}}{%
{\ottdrulename{StepAppFailTrans}}{}%
}}

\newcommand{\ottdruleStepContext}[1]{\ottdrule[#1]{%
\ottpremise{\evterm{t}_{{\mathrm{1}}}  \longrightarrow  \evterm{t}_{{\mathrm{2}}}}%
\ottpremise{\evterm{t}_{{\mathrm{1}}}  \ottsym{,}  \evterm{t}_{{\mathrm{2}}}  \neq \, \mathsf{err}}%
}{
 \evterm{C}[  \evterm{t}_{{\mathrm{1}}}  ]   \longrightarrow   \evterm{C}[  \evterm{t}_{{\mathrm{2}}}  ] }{%
{\ottdrulename{StepContext}}{}%
}}

\newcommand{\ottdruleStepContextErr}[1]{\ottdrule[#1]{%
\ottpremise{\evterm{t}  \longrightarrow  \mathsf{err}}%
}{
 \evterm{C}[  \evterm{t}  ]   \longrightarrow  \mathsf{err}}{%
{\ottdrulename{StepContextErr}}{}%
}}

\newcommand{\ottdefnStep}[1]{\begin{ottdefnblock}[#1]{$\evterm{t}_{{\mathrm{1}}}  \longrightarrow  \evterm{t}_{{\mathrm{2}}}$}{\ottcom{Evidence-based Small-Step Semantics}}
\ottusedrule{\ottdruleStepAscr{}}
\ottusedrule{\ottdruleStepAscrFail{}}
\ottusedrule{\ottdruleStepAppDyn{}}
\ottusedrule{\ottdruleStepAppEv{}}
\ottusedrule{\ottdruleStepAppEvRaw{}}
\ottusedrule{\ottdruleStepAppFailTrans{}}
\ottusedrule{\ottdruleStepContext{}}
\ottusedrule{\ottdruleStepContextErr{}}
\end{ottdefnblock}}

  \renewottcommands[ott]

\renewcommand{\ottcom}[1]{\textbf{#1}}

\newcommand{\rev}[1]{#1}
\newcommand{\added}[1]{\rev{#1}}
\newcommand{\replaced}[2]{\rev{#1}}

\begin{document}

    \title{Approximate Normalization for Gradual Dependent Types}                                                                                             \renewcommand\footnotemark{}
\titlenote{
This article extends the ICFP19 article found at \url{https://doi.org/10.1145/3341692}. 
This work is partially funded by CONICYT FONDECYT Regular Project 1190058, ERC Starting Grant SECOMP (715753),
an NSERC Discovery grant,
and the NSERC Canada Graduate Scholarship.
}

    \author{Joseph Eremondi}
                                                \affiliation{
        \department{Department of Computer Science}                  \institution{University of British Columbia}                                \country{Canada}                      }
  \email{{jeremond@cs.ubc.ca}}               

        \author{\'{E}ric Tanter}
    \affiliation{
            \department{Computer Science Department (DCC)}                    \institution{University of Chile}                                          \country{Chile}                        }
    \affiliation{
     \institution{Inria Paris}
     \country{France}
    }
    \email{etanter@dcc.uchile.cl}            
      \author{Ronald Garcia}
  \affiliation{
        \department{Department of Computer Science}                  \institution{University of British Columbia}                                \country{Canada}                      }
  \email{rxg@cs.ubc.ca}

        \begin{abstract}
  
Dependent types help programmers write highly reliable code.  However, this
reliability comes at a cost: it can be challenging to write new prototypes in
(or migrate old code to) dependently-typed programming languages.  Gradual
typing {makes} static type disciplines more
flexible, so an appropriate notion of gradual dependent types could fruitfully
lower this cost. 
However, dependent types raise unique challenges for gradual typing.  Dependent 
typechecking involves the execution of program code, but gradually-typed code
can signal runtime type errors or diverge.  These runtime errors
threaten the soundness guarantees that make dependent types so attractive,
while divergence spoils the type-driven programming experience.

This paper presents GDTL, a gradual dependently-typed language
that emphasizes pragmatic dependently-typed programming.  GDTL fully embeds both an
untyped and dependently-typed language, and allows for smooth transitions
between the two.  In addition to gradual types we introduce \emph{gradual
terms}, which allow the user to be imprecise in type indices and to omit proof
terms; runtime checks ensure type {safety}.  To account for nontermination and
failure, we distinguish between compile-time normalization and run-time
execution: compile-time normalization is {\em approximate} but total, while
runtime execution is {\em exact}, but may fail or diverge.  We prove that GDTL
has decidable typechecking and satisfies all the expected properties of gradual
languages. In particular, GDTL satisfies the static and dynamic gradual
guarantees: reducing type precision preserves typedness, and altering type
precision does not change program behavior outside of dynamic type failures.
To prove these properties, we were led to establish a novel {\em
normalization gradual guarantee} that captures the monotonicity of approximate
normalization with respect to imprecision.

  \end{abstract}


         
          \maketitle

\section{Introduction}  

Dependent types support the development of extremely reliable software.  With
the full power of higher-order logic, programmers can write expressive
specifications as types, and be confident that if a program typechecks, then it
meets its specification. 
Dependent types are at the core of proof assistants like Coq~\citep{coqart} and
Agda~\citep{agdaPaper}. While these pure systems can be used for certified
programming~\citep{cpdt}, their focus is on the construction of proofs, rather
than practical or efficient code. 
\added{Dependently-typed extensions of practical programming languages maintain a clear phase distinction between compile-time typechecking and runtime execution, and have to embrace some compromise regarding impurity.} 
\rev{One possibility is to forbid potentially impure expressions from occurring in types, either by considering a separate pure sub-language of type-level computation as in Dependent ML~\citep{Xi1999}, by using an effect system and termination checker to prevent impurity to leak in type dependencies as in F$\star$~\cite{mumon} and Idris~\citep{idrisPaper}, or by explicitly separating the language into two fragments with controlled interactions between them, as in Zombie~\citep{EPTCS76.9, Casinghino:2014:CPP:2535838.2535883}. A radical alternative is to give up on decidable typechecking and logical consistency altogether and give all responsibility to the programmer, as in Dependent Haskell~\citep{eisenberg2016dependent}. The design space is wide, and practical dependently-typed programming is a fertile area of research. }

As with any static type system, the reliability brought by dependent types
comes at a cost. Dependently-typed languages impose a rigid discipline,
sometimes requiring programmers to explicitly construct proofs in their
programs. Because of this rigidity, dependent types can interfere with rapid
prototyping, and migrating code from languages with simpler type systems can be
difficult.
Several approaches have been proposed to relax dependent typing in order to
ease programming, for instance by supporting some form of interoperability
between (possibly polymorphic) non-dependently-typed and dependently-typed programs and
structures~\citep{10.1007/1-4020-8141-3_34,Osera:2012:DI:2103776.2103779,
  dagand_tabareau_tanter_2018,Tanter:2015:GCP:2816707.2816710}. These
approaches require programmers to explicitly trigger runtime checks through
casts, liftings, or block boundaries.  Such explicit interventions hamper evolution.

In contrast, gradual
typing~\citep{gradualTypeInitial} exploits {\em type imprecision} to
drive the interaction between static and dynamic checking in a smooth,
continuous manner~\citep{refinedCriteria}. A gradual language introduces an
unknown type $ {\qm } $, and admits imprecise types such as
$  \ottkw{Nat}  \to {\qm } $.  The gradual type system optimistically handles
imprecision, deferring to runtime checks where needed. Therefore, runtime
checking is an {\em implicit} consequence of type imprecision, and is
seamlessly adjusted as programmers evolve the declared types of components, be
they modules, functions, or expressions.  This paper extends gradual typing 
to provide a flexible incremental path to adopting dependent types.

Gradual typing has been adapted to many other type disciplines, including
ownership types~\citep{sergeyClarke:esop2012},
effects~\cite{banadosAl:jfp2016}, refinement
types~\citep{lehmannTanter:popl2017}, security
types~\citep{fennellThiemann:csf2013,toroAl:toplas2018}, and session
types~\citep{igarashiAl:icfp2017b}.  But it has not yet reached dependent
types.  Even as the idea holds much promise, it also poses significant
challenges.  The greatest barrier to gradual dependent types is that a
dependent type checker must evaluate some program terms as part of type
checking, and gradual types complicate this in two ways.
First, if a gradual language fully embeds
an untyped language, then some programs will diverge: indeed, self
application $(\lambda  \mathit{x} : {\qm } .x\; x)$ is typeable in such a
language.
Second, gradual languages introduce the possibility of type
errors \added{that are uncovered as a term is evaluated}: applying the function $(\lambda  \mathit{x} : {\qm }  \ldotp  \mathit{x}  + 1 )$ may fail,
depending on whether its argument \added{can actually be used as a} number.
So a gradual dependently-typed language must account for the potential of non-termination and failure \textit{during typechecking}.

\paragraph{{\bf A gradual dependently-typed language}} This work presents \lang, a gradual dependently-typed core language that supports the whole spectrum between an untyped functional language and a dependently-typed one. As such, \lang adopts a unified term and type language, meaning that the unknown type $ {\qm } $ is also a valid term. This allows programmers to specify types with imprecise indices, and to replace proof terms with $ {\qm } $ (\autoref{sec:motivation}).

\lang is a gradual version of the predicative fragment of the Calculus of Constructions with a cumulative universe hierarchy (\CCw) (\autoref{sec:staticlang}), similar to the core language of Idris~\citep{idrisPaper}. 
We gradualize this language following the Abstracting Gradual Typing (AGT) methodology~\citep{agt} (\autoref{sec:graduallang}).
Because \lang is a {\em conservative extension} of this dependently-typed calculus, it is both strongly normalizing and logically consistent for fully static code. \rev{These strong properties are however lost as soon as imprecise types and/or terms are introduced.}
On the dynamic side, \lang can fully embed the untyped lambda calculus. When writing purely untyped code, static type errors are never encountered.
In between, \lang satisfies the {\em gradual guarantees} of \citet{refinedCriteria}, meaning that typing and evaluation are monotone with respect to type imprecision. These guarantees ensure that programmers can move code between imprecise and precise types in small, incremental steps,
with the program typechecking and \added{behaving identically (modulo dynamic type errors)} at each step.
\added{If a program fails to typecheck, the programmer knows the problem is not too few type annotations, but rather incompatible types.}

\lang is a call-by-value language with a sharp two-phase distinction. The key technical insight on which \lang is built is to exploit two distinct notions of evaluation: one for {\em normalization} during typechecking, and one for {\em execution} at runtime. Specifically, we present a novel {\em approximate} normalization technique that guarantees decidable typechecking (\autoref{sec:gradualhsub}): 
applying a function of unknown type, which may trigger non-termination, normalizes to the {\em unknown value} $ {\qm} $. Consequently, some terms that would be distinct at runtime become indistinguishable as type indices.
Approximation is also used to ensure that compile-time normalization (i.e. during typechecking) always terminates and never signals a ``dynamic error''. \added{In this sense, \lang is closer to Idris than Dependent Haskell, which does admit reduction errors and non-termination during typechecking.} At runtime, \lang uses the standard, precise runtime execution strategy of gradual languages, which may fail due to dynamic type errors, and may diverge as well (\autoref{sec:gradual-semantics}). \added{In that respect, \lang is closer to Dependent Haskell and Zombie than to Idris, which features a termination checker and a static effect system.}
We prove that \lang has decidable typechecking and satisfies all the expected properties of gradual languages \citep{refinedCriteria}: type safety, conservative extension of the static language, embedding of the untyped language, and the gradual guarantees (\autoref{sec:lang-properties}).
We then show how inductive types with eliminators can be added to \lang without significant changes (\autoref{sec:inductive}).
\autoref{sec:related} discusses related work, and \autoref{sec:conclusion} discusses limitations and perspectives for future work.

\paragraph{{\bf Disclaimer}}
\rev{
This work does not aim to develop a full-fledged dependent gradual type theory, a fascinating objective that would raise many metatheoretic challenges. Rather, it proposes a novel technique, {\em approximate normalization}, applicable to \added{full-spectrum} dependently-typed  programming languages. This technique reflects specific design choices that affect the pragmatics of programming and reasoning in \lang: we review these design decisions, after the informal presentation of the language, in \autoref{sec:choices}.
Also, being a core calculus, \lang currently lacks several features expected of a practical dependently-typed language (\autoref{sec:conclusion}); nevertheless, this work provides a foundation on which practical gradual dependently-typed languages can be built.
}

{\bf \added{Implementation}.} 
\rev{ We provide a prototype implementation of \lang in Racket, based on
a Redex model. The code for the implementation is open source~\citep{GDTL-github}.
The implementation also supports our extension of natural numbers, vectors, and equality
as built-in inductive types. }

{\bf \added{Technical report}.} 
\rev{ Complete definitions and proofs can be found in \citep{eremondi2019approximate} }.

\section{Goals and Challenges}
\label{sec:motivation}

We begin by motivating our goals for \lang, and describe the challenges and
design choices that accompany them.

\subsection{The Pain and Promise of Dependent Types}
\label{subsec:vector}

To introduce dependent types, we start with a classic example: length-indexed
vectors. In a dependently-typed language, the type $   \ottkw{Vec}  \  \mathtt{A}  \  \mathtt{n} $ describes any
vector that contains $\mathtt{n}$ elements of type $\mathtt{A}$.  We say that $ \ottkw{Vec} $
is \textit{indexed} by the value $\mathtt{n}$.  This type has two constructors,
${ \ottkw{Nil} :   \ottsym{(}  \mathtt{A}  \ottsym{:}   \TypeType_{ \ottsym{1} }   \ottsym{)}  \rightarrow   \ottkw{Vec}  \  \mathtt{A}  \   0  }$ and
${\ottkw{Cons}:    \ottsym{(}  \mathtt{A}  \ottsym{:}   \TypeType_{ \ottsym{1} }   \ottsym{)}  \rightarrow  \ottsym{(}  \mathtt{n}  \ottsym{:}   \ottkw{Nat}   \ottsym{)}  \rightarrow   \mathtt{A} \to     \ottkw{Vec}  \  \mathtt{A}  \  \mathtt{n}  \to  \ottkw{Vec}    \  \mathtt{A}  \  \ottsym{(}   \mathtt{n}  + 1   \ottsym{)}  }$.
By making length part of the type, we ensure that operations that are typically
partial can only receive values for which they produce results.  
One can type a function that yields the first element of a vector as follows:
$$ \ottkw{head}\ :\ \ottsym{(}  \mathtt{A}  \ottsym{:}   \TypeType_{ \ottsym{1} }   \ottsym{)}  \rightarrow  \ottsym{(}  \mathtt{n}  \ottsym{:}   \ottkw{Nat}   \ottsym{)}  \rightarrow      \ottkw{Vec}  \  \mathtt{A}  \  \ottsym{(}   \mathtt{n}  + 1   \ottsym{)}  \to \mathtt{A} $$ 

Since $\mathtt{head}$ takes a vector of non-zero length, \rev{it can never receive} an empty vector. 
The downside of this strong guarantee is that we can only use vectors in
contexts where their length is known.  This makes it difficult to migrate code
from languages with weaker types, or for newcomers to prototype algorithms.
For example, a programmer may wish to migrate the following quicksort algorithm
into a dependently-typed language:
\begin{align*}
    & \mathtt{sort} \  \mathtt{vec}  \ =&&&&  \ottkw{if}\ \mathtt{vec} == \ottkw{Nil}\ \ottkw{then}\ \ottkw{Nil}\ \ottkw{else}  \\
    &&&&& \qquad (\text{\texttt{sort (filter ($\leq$ (head vec)) (tail vec)))}}) \ \mathtt{++}\   \mathtt{head} \  \mathtt{vec}    \\
    &&&&& \qquad \ \mathtt{++}\ (\text{\texttt{sort (filter ($>$ (head vec)) (tail vec)))}}) 
\end{align*}

\noindent Migrating this definition to a dependently-typed language poses some
difficulties. The recursive calls are not direct deconstructions of $\mathtt{vec}$,
so it takes work to convince the type system that the code will terminate, and
is thus safe to run at compile time.  Moreover, if we try to use this
definition with $ \ottkw{Vec} $, we must account for how the length of each filtered list is unknown, 
and while we can prove that the length
of the resulting list is the same as the input, this must be done
manually. Alternately, we could use simply-typed lists in the dependently-typed
language, but we do not wish to duplicate every vector function for lists.

\subsection{Gradual Types to the Rescue?}

Even at first glance, gradual typing seems like it can provide the desired flexibility.
In a gradually-typed language, a programmer can use the unknown type, 
written $ {\qm } $, to soften the static typing discipline.  Terms with type
$ {\qm } $ can appear in any context.  Runtime type checks ensure
that dynamically-typed code does not violate invariants expressed with static
types.

Since $ {\qm } $ allows us to embed untyped code in a typed language, we can write a gradually-typed fixed-point
combinator ${\mathsf{Z}  :  \ottsym{(}  \mathtt{A}  \ottsym{:}   \TypeType_{ \ottsym{1} }   \ottsym{)}  \rightarrow  \ottsym{(}  \mathtt{B}  \ottsym{:}   \TypeType_{ \ottsym{1} }   \ottsym{)}  \rightarrow   \ottsym{(}   \ottsym{(}    \mathtt{A}  \to  \mathtt{B}    \ottsym{)} \to   \mathtt{A}  \to  \mathtt{B}     \ottsym{)} \to   \mathtt{A}  \to  \mathtt{B}   }$.
\added{We define this in the same way as the usual $\mathsf{Z}$ combinator, but the input function is ascribed the type $\qm$,
allowing for self-application.}
Using this combinator, the programmer can write $\mathtt{sort}$ using general recursion. 
Furthermore, the programmer can give $\mathtt{sort}$ the type $  {\qm } \to {\qm }  $, 
causing the length of the results of $\mathtt{filter}$ to be ignored in typechecking.
Annotating the vector with $ {\qm } $
inserts runtime checks that ensure that the program will \rev{fail (rather than behave in an undefined manner)} if it is given an argument that is not a vector.

However, introducing the dynamic type $ {\qm } $ in a dependently-typed language brings challenges.

\paragraph{{\bf The unknown type is not enough}}If we assign $\mathtt{sort}$
the type $  {\qm } \to {\qm }  $, then we can pass it any argument, whether
it is a vector or not.  This seems like overkill: we want to restrict
$\mathtt{sort}$ to vectors and statically rule out nonsensical calls like
$ \mathtt{sort} \  \mathtt{false} $.  Unfortunately the usual notion of type imprecision is too
coarse-grained to support this. We want to introduce imprecision
more judiciously, as in $   \ottkw{Vec}  \  \mathtt{A}  \  {\qm } $: the type of vectors with \textit{unknown
  length}. But the length is a natural number, not a type.  How can we express
imprecision in type indices?

\paragraph{{\bf Dependent types require proofs}} To seamlessly blend
untyped and dependently-typed code, we want to let programs omit
proof terms, yet still allow code to typecheck and run. But this goes further
than imprecision in type indices, since imprecision also manifests in program
terms.  What should the dynamic semantics of imprecise programs be?

\paragraph{{\bf Gradual typing introduces effects}} 
Adding $ {\qm } $ to types introduces two effects. The ability to type
self-applications means programs may diverge, and the ability to write
imprecise types introduces the possibility of type errors uncovered while evaluating terms.  These effects are troubling because dependent typechecking must often evaluate code, sometimes under
binders, to compare dependent types.  
We must normalize terms at compile time to compute the type of dependent function applications.
This means that both effects can manifest
\textit{during typechecking}.  How should compile-time evaluation errors be
handled? Can we make typechecking decidable in the presence of possible
non-termination?

\paragraph{{\bf Dependent types rely on equality}}

The key reason for normalizing at compile time is that
we must compare types, and
since types can be indexed by terms, we need a method of comparing arbitrary terms.
If we have $A : (\ottkw{Nat} \to \ottkw{Nat}) \to \TypeType_1$,
then  $A\ (\lambda  \mathit{x}  \ldotp    1   + 1) $ and $A\ (\lambda  \mathit{x}  \ldotp   2) $
should be seen as the same type.
In intensional type theories, like that of Coq, Agda and Idris, this is done using \textit{definitional equality}, which fully evaluates terms (even under binders) to determine
whether their normal forms are \textit{syntactically} equal. Of course, this is a weaker notion than {\em propositional equality}, but typechecking with propositional equality (as found in {\em extensional} theories) is undecidable. In intensional theories, explicit rewriting must be used to exploit propositional equalities.

In a gradual language, types are also compared \textit{at runtime} to compensate for imprecise static type information. With gradual dependent types, how should types (and the terms they contain) be compared at runtime? The simplest solution is to use the same notion of definitional equality as used for static typechecking. This has some unfortunate consequences, such as  $A\ (\lambda  \mathit{x}  \ldotp    x   + x - x) $ and $A\ (\lambda  \mathit{x}  \ldotp   x) $ being deemed inconsistent, even though they are clearly propositionally equal. However, mirroring compile-time typechecking at runtime simplifies reasoning about the language behavior.

\subsection{GDTL in Action} 
\label{subsec:gdtl-action}
To propagate imprecision to type indices, and soundly allow omission of proof terms, \lang admits $ {\qm } $ {\em both as a type and a term}.
To manage effects due to gradual typing, we use separate notions of evaluation for compile-time and runtime.
Introducing {\em imprecision in the compile-time normalization of types} avoids both non-termination and failures during typechecking.

\paragraph{{\bf The unknown as a type index}}

Since full-spectrum dependently-typed languages conflate types and terms, \lang allows $ {\qm } $ to be used as either a term or a type. 
Just as any term can have type $ {\qm } $, the term $ {\qm } $ can have any type.
This lets dependent type checks be deferred to runtime. For example, we can define
vectors $\mathtt{staticNil}$, $\mathtt{dynNil}$ and $\mathtt{dynCons}$ as follows:
\begin{align*}
    &\mathtt{staticNil}  :     \ottkw{Vec}  \   \ottkw{Nat}   \   0   && \mathtt{dynNil}  :     \ottkw{Vec}  \   \ottkw{Nat}   \  {\qm }  \quad  &&\mathtt{dynCons}  :     \ottkw{Vec}  \   \ottkw{Nat}   \  {\qm }  \\
    &\mathtt{staticNil}  =    \ottkw{Nil}  \   \ottkw{Nat}   && \mathtt{dynNil}  =    \ottkw{Nil}  \   \ottkw{Nat}    && \mathtt{dynCons}  =       \ottkw{Cons}  \   \ottkw{Nat}   \   0   \   0   \  \ottsym{(}    \ottkw{Nil}  \   \ottkw{Nat}    \ottsym{)} 
\end{align*}

\noindent Then, 
$\ottsym{(}     \mathtt{head} \   \ottkw{Nat}   \   1   \  \mathtt{staticNil}   \ottsym{)}$  does not typecheck,
$\ottsym{(}     \mathtt{head} \   \ottkw{Nat}   \   1   \  \mathtt{dynNil}   \ottsym{)}$ typechecks but fails at runtime,
and $\ottsym{(}     \mathtt{head} \   \ottkw{Nat}   \   1   \  \mathtt{dynCons}   \ottsym{)}$ typechecks and succeeds at runtime.
The programmer can choose between compile-time or runtime checks, but \replaced{safety}{soundness} is maintained
either way, and in the fully-static case, the unsafe code is still rejected.

\paragraph{{\bf The unknown as a term at runtime}}

Having $ {\qm } $ as a term means that programmers can use it to optimistically omit \textit{proof terms}.
Indeed, terms can be used not only as type indices, but also as proofs of propositions.
For example, consider \added{the equality type} $ \ottkw{Eq}  : \ottsym{(}  \mathtt{A}  \ottsym{:}   \TypeType_{ \ottnt{i} }   \ottsym{)}  \rightarrow    \mathtt{A}  \to   \mathtt{A}  \to  \TypeType_{ \ottnt{i} }   $,
along with its lone constructor $ \ottkw{Refl}  :    \ottsym{(}  \mathtt{A}  \ottsym{:}   \TypeType_{ \ottnt{i} }   \ottsym{)}  \rightarrow  \ottsym{(}  \mathit{x}  \ottsym{:}  \mathtt{A}  \ottsym{)}  \rightarrow   \ottkw{Eq}  \  \mathtt{A}  \  \mathit{x}  \  \mathit{x} $.
We can use these to write a (slightly contrived) formulation of the $\mathtt{head}$ function:
$$ \mathtt{head}'\ :\ \ottsym{(}  \mathtt{A}  \ottsym{:}   \TypeType_{ \ottnt{i} }   \ottsym{)}  \rightarrow  \ottsym{(}  \mathtt{n}  \ottsym{:}   \ottkw{Nat}   \ottsym{)}  \rightarrow  \ottsym{(}  \mathtt{m}  \ottsym{:}   \ottkw{Nat}   \ottsym{)}  \rightarrow       \ottkw{Eq}  \   \ottkw{Nat}   \  \mathtt{n}  \  \ottsym{(}   \mathtt{m}  + 1   \ottsym{)}  \to     \ottkw{Vec}  \  \mathtt{A}  \  \mathtt{n}  \to \mathtt{A}  $$

This variant accepts vectors of any length,
provided the user also supplies a proof that its length $\mathtt{n}$ is not zero (by providing the predecessor $\mathtt{m}$ and the equality proof $   \ottkw{Eq}  \  \mathtt{n}  \  \ottsym{(}   \mathtt{m}  + 1   \ottsym{)} $).
\lang allows $ {\qm } $ to be used in place of a proof,
while still ensuring that a runtime error is thrown if $\mathtt{head}'$ is ever given an empty list.  
For instance, suppose we define a singleton vector
and a proof that $0=0$:
\begin{align*}
    &\mathtt{staticCons}  :    \ottkw{Vec}  \   \ottkw{Nat}   \   1   \qquad\qquad\qquad  && \mathtt{staticProof}  :      \ottkw{Eq}  \   \ottkw{Nat}   \   0   \   0   \\
    &\mathtt{staticCons}  =      \ottkw{Cons}  \   \ottkw{Nat}   \   0   \  \ottsym{(}    \ottkw{Nil}  \   \ottkw{Nat}    \ottsym{)}   && \mathtt{staticProof}  =     \ottkw{Refl}  \   \ottkw{Nat}   \   0   
\end{align*}

Then
${\ottsym{(}       \mathtt{head'} \   \ottkw{Nat}   \   0   \   0   \  \mathtt{staticProof}  \  \mathtt{staticNil}   \ottsym{)}}$ does not typecheck,
${\ottsym{(}       \mathtt{head'} \   \ottkw{Nat}   \   0   \  {\qm }  \  {\qm }  \  \mathtt{staticNil}   \ottsym{)}}$  typechecks but fails at runtime, and
${\ottsym{(}       \mathtt{head'} \   \ottkw{Nat}   \   1   \  {\qm }  \  {\qm }  \  \mathtt{staticCons}   \ottsym{)}}$ typechecks and succeeds at runtime.
To see why we get a runtime failure for $\mathtt{staticNil}$, we note that internally, $\mathtt{head'}$ uses an explicit \textit{rewriting} of the equality, i.e. if $\mathit{x}$ and $\mathit{y}$ are equal, then 
any property $\mathtt{P}$ that holds for $\mathit{x}$ must also hold for $\mathit{y}$:
$$ \mathtt{rewrite}\ :\  \ottsym{(}  \mathtt{A}  \ottsym{:}   \TypeType_{ \ottsym{1} }   \ottsym{)}  \rightarrow  \ottsym{(}  \mathit{x}  \ottsym{:}  \mathtt{A}  \ottsym{)}  \rightarrow  \ottsym{(}  \mathit{y}  \ottsym{:}  \mathtt{A}  \ottsym{)}  \rightarrow  \ottsym{(}  \mathtt{P}  \ottsym{:}   \mathtt{A} \to  \TypeType_{ \ottsym{1} }    \ottsym{)}  \rightarrow       \ottkw{Eq}  \  \mathtt{A}  \  \mathit{x}  \  \mathit{y}  \to   \mathtt{P} \  \mathit{x}  \to \mathtt{P}   \  \mathit{y} $$

\rev{
In \lang, when $ {\qm } $ is treated as an equality proof, it behaves as $   \ottkw{Refl}  \  {\qm }  \  {\qm } $.\footnote{
	This follows directly from the understanding of the unknown type $ {\qm } $ as denoting all possible static types~\citep{agt}. 
	Analogously, {\em the gradual term $ {\qm } $ 
denotes all possible static terms}. Thus applying the term $ {\qm } $ as a function represents applying all possible functions, producing all possible results, which can be abstracted as the gradual term $ {\qm } $. This means that $ {\qm } $, when applied as a function, behaves as $\lambda  \mathit{x}  \ldotp  {\qm }$. Similarly, $ {\qm } $ treated as an equality proof behaves as $   \ottkw{Refl}  \  {\qm }  \  {\qm } $.
}
Therefore, in the latter two cases of our example, $\mathtt{rewrite}$ gives a result of type $ \mathtt{P} \  {\qm } $, 
which is checked against type $ \mathtt{P} \  \mathit{y} $, \textit{at runtime}.
If $ \mathtt{P} \  \mathit{x} $ and $ \mathtt{P} \  \mathit{y} $ are not \added{definitionally equal}, then 
this check fails with a runtime error.
}

\paragraph{{\bf Managing effects from gradual typing}}
\label{subsec:factorial}

To illustrate how the effects of gradual typing can show up in typechecking, 
suppose a programmer uses the aforementioned $\mathtt{Z}$ combinator to
accidentally write a non-terminating function $\mathtt{badFact}$.
$$
    \mathtt{badFact}   =   \lambda\ m\ \ldotp Z\ (\lambda\ f\ \ldotp  \mathtt{ifzero}\ m\ (f\ 1) (m * f\ (m)) \quad \textit{ \texttt{--} never terminates} 
$$

As explained before, from a practical point of view, it is desirable for \lang 
to fully support dynamically-typed terms, because it allows the programmer to
opt out of both the type discipline and the termination discipline of a dependently-typed language.
However, this means that computing the return type of a function application may diverge, for instance:
\begin{align*}
    &\mathsf{repeat}  :   \ottsym{(}  \mathtt{A}  \ottsym{:}   \TypeType_{ \ottsym{1} }   \ottsym{)}  \rightarrow  \ottsym{(}  \mathtt{n}  \ottsym{:}   \ottkw{Nat}   \ottsym{)}  \rightarrow   \mathtt{A} \to \ottsym{(}     \ottkw{Vec}  \  \mathtt{A}  \  \mathtt{n}   \ottsym{)} \\
   &\mathsf{factList}   =     \mathtt{repeat} \   \ottkw{Nat}   \  \ottsym{(}   \mathtt{badFact} \   1    \ottsym{)}  \   0    \quad \textit{ \texttt{--} has type } \mathit{   \ottkw{Vec}  \  \ottkw{Nat}  \  \ottsym{(}   \mathtt{badFact} \   1    \ottsym{)} }
\end{align*}

\noindent To isolate the non-termination from imprecise code, we observe that any diverging code will necessarily apply a function of type $ {\qm } $.
\added{While $\mathtt{badFact}$ does not have type $\qm$, its definition uses $Z$, which contains ascriptions of type $\qm$. }
  
Similarly, a na\"ive approach to gradual dependent types will encounter failures
when normalizing some terms.
Returning to our $\mathtt{head}$ function, how should we typecheck the following term?
$$\mathtt{failList} = {   \mathtt{head} \   \ottkw{Nat}   \  \ottsym{(}   \mathtt{false} \dblcolon {\qm }   \ottsym{)}  \  \mathtt{staticCons} }$$ 
\noindent We need to check $\mathtt{staticCons}$ against ${   \ottkw{Vec}  \   \ottkw{Nat}   \  \ottsym{(}   \ottsym{(}   \mathtt{false} \dblcolon {\qm }   \ottsym{)}  + 1   \ottsym{)} }$, but what does ${ \ottsym{(}   \mathtt{false} \dblcolon {\qm }   \ottsym{)}  + 1 }$ mean as a vector length? 

The difficulty here is that if a term contains type ascriptions that 
may produce a runtime failure, then it will
\textit{always} trigger an error when normalizing, since normalization evaluates under binders.
This means that typechecking will fail
any time we apply a function to a possibly-failing term.   
This is highly undesirable, and goes against the spirit of gradual typing: writing programs in the large would be very difficult if applying a function to an argument that does match its domain type caused a type error, or caused typechecking to diverge! 
\added{Whereas Dependent Haskell places the burden on the programmer to ensure termination and freedom from failure during typechecking, doing so in a gradual language would make it difficult for programmers because of the possibly indirect interactions with untyped code.}

\lang avoids both problems by using different notions of running programs for the
compile-time and runtime phases.
We distinguish compile-time \textit{normalization}, which is {\em approximate but total}, 
from runtime  \textit{execution}, 
which is {\em exact but partial}. 
When non-termination or failures are possible, compile-time normalization uses $ {\qm } $ as an approximate but pure result.
So both $\mathtt{factList}$ and $\mathtt{failList}$ can be defined and used in runtime code, but they are assigned type $   \ottkw{Vec}  \   \ottkw{Nat}   \  {\qm } $.
To avoid non-termination and dynamic failures, we want our language to be strongly normalizing during typechecking.
Approximate normalization gives us this. 
  
\lang normalization is focused around \textit{hereditary substitution}~\citep{hereditary}, which is a total operation from canonical forms to canonical forms.
Because hereditary substitution is structurally decreasing in the \textit{type} of the value being substituted,
a static termination proof is easily adapted to \lang.
This allows us to pinpoint exactly where gradual types \rev{introduce} effects,
approximate in those cases, and easily adapt the proof of termination of a static language to the gradual language \lang.
Similarly, our use of bidirectional typing means that a single check needs to be added to
prevent failures in normalization.

\subsection{Gradual Guarantees for \lang}
\label{subsec:gg-intro}

To ensure a smooth transition between precise and imprecise typing,
\lang satisfies the \textit{gradual guarantee},
which comes in two parts~\citep{refinedCriteria}.
The \textit{static gradual guarantee} says that reducing the precision
of a program preserves its well-typedness. 
The \textit{dynamic gradual guarantee} states that reducing the precision of a program also preserves its behavior, though the resulting value may be less precise.

One novel insight of \lang's design is that the interplay between dependent typechecking and program evaluation
carries over to the gradual guarantees. Specifically, the static gradual guarantee fundamentally depends on a restricted
variant of the dynamic gradual guarantee. We show that approximate normalization maps
terms related by precision to canonical forms related by precision, thereby ensuring that reducing a term's precision always preserves well-typedness.

By satisfying the gradual typing criteria, \added{and embedding both a fully static and a fully dynamic fragment}, \lang gives programmers freedom to move within the entire spectrum of typedness,
from the safety of higher-order logic to the flexibility of dynamic languages.
Furthermore, admitting $ {\qm } $ as a term means that we can easily combine
\added{code with dependent and non-dependent types}, the midpoint between dynamic and dependent types.
For example, the simple list type could be written as $  \ottkw{List}  \  \mathtt{A}  =    \ottkw{Vec}  \  \mathtt{A}  \  {\qm } $,
so lists could be given to vector-expecting code and vice-versa.
The programmer knows that as long as vectors are used in vector-expecting code, no crashes can happen,
and \replaced{safety}{soundness} ensures that using a list in a vector operation will always fail gracefully or run successfully.
This is significantly different from work on casts to subset types~\cite{Tanter:2015:GCP:2816707.2816710} and dependent interoperability~\citep{dagand_tabareau_tanter_2018}, 
where the user must explicitly provide decidable properties or (partial) equivalences.

\subsection{Summary of Design Decisions}
\label{sec:choices}

\rev{\lang embodies several important design decisions, each with tradeoffs related to ease of reasoning and usability of the language.

By embracing full-spectrum dependent types, 
\lang allows types to be first-class citizens: arbitrary terms can appear in types and expressions can produce types as a result. Therefore the programmer does not need to learn a separate index language, and there is no need to recreate term-level operations at the type level. 

Sticking to clearly separated phases allows us to adopt different reduction strategies for typechecking and for execution. Crucially, by using \textit{approximate normalization}, we ensure that typechecking in \lang always terminates: compile-time normalization is a total (though imprecise) operation. This means that some type information is statically lost, with checks deferred to runtime. 

\lang features an unknown term $\qm$, which resembles term holes in Agda and Idris, and existential variables in Coq; the notable difference is that programs containing $\qm$ can be run without evaluation getting stuck. Every type in \lang is therefore inhabited at least by the unknown term $\qm$, which means that the language is inconsistent as a logic, except for fully-precise programs. 

In a gradual language that can embed arbitrary untyped terms, programs may not terminate at runtime. Every type in \lang contains expressions that can fail or diverge at runtime, due to imprecision. Fully-precise programs are guaranteed to terminate.

Finally, like Coq, Agda, and Idris, \lang is based on an intensional type theory, meaning that it automatically decides {\em definitional equality}---i.e. syntactic equality up to normalization---and not propositional equality; explicit rewriting is necessary to exploit propositional equalities. Consequently, runtime checks in \lang also rely on definitional equality. 
This makes equality decidable, but means that a runtime error can be triggered even though two (syntactically different) terms are propositionally equal.
}

\section{\slang: A Static Dependently-Typed Language}
\label{sec:staticlang}

We now present \slang, a static dependently-typed language 
which is essentially a bidirectional, call-by-value, cumulative variant of the predicative fragment of \CCw (i.e. the calculus of constructions with a universe hierarchy~\citep{COQUAND198895}). 
\slang is the starting point of our gradualization effort, following the Abstracting Gradual Typing (AGT) methodology~\cite{agt}, refined to accommodate dependent types.

\subsection{Syntax and Dynamic Semantics}

\begin{figure}
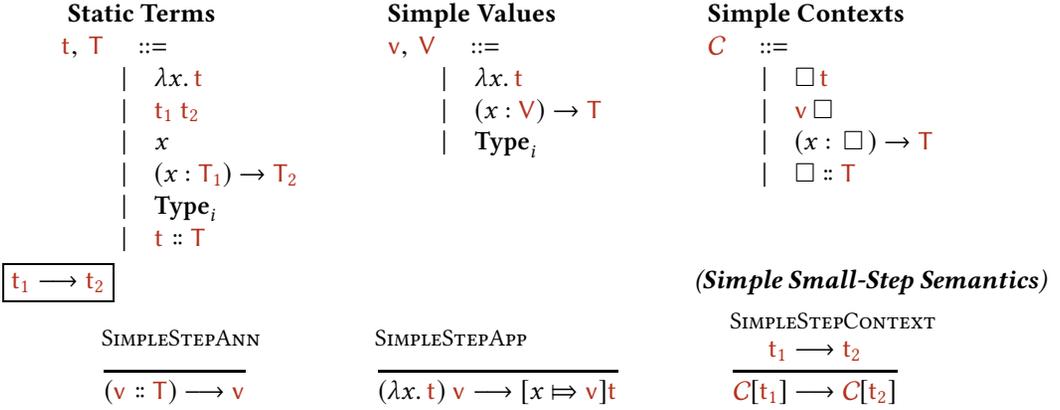

\begin{minipage}[t]{0.3\linewidth}  
    \nonterms{t}   
\end{minipage}   
\begin{minipage}[t]{0.3\linewidth} 
    \nonterms{simpleValue}  
\end{minipage}   
\begin{minipage}[t]{0.3\linewidth} 
    \nonterms{simpleContext}  
\end{minipage}

\ottdefnSimpleStep{}
\caption{\slang: Syntax and Semantics}
\label{fig:syntax-static} 
\end{figure}          

The syntax of \slang is shown in \autoref{fig:syntax-static}. Metavariables for the static variants of terms, values, etc. are written in $\staticstyle{red,\ sans\text{-}serif\ font}{}{}$.
Types and terms share a syntactic category. 
Functions and applications are in their usual form.
Function types are \textit{dependent}:
a variable name is given to the argument, and the codomain may refer to this variable.
We have a \textit{universe hierarchy}: the lowest types have the type $ \TypeType_{ \ottsym{1} } $,
and each $ \TypeType_{ \ottnt{i} } $ has type $ \TypeType_{ \ottnt{i}  \ottsym{+}  \ottsym{1} } $.
 This hierarchy is \textit{cumulative}:
any value in $ \TypeType_{ \ottnt{i} } $ is also in $ \TypeType_{ \ottnt{i}  + 1} $.
Finally, we have a form for explicit type ascriptions.

We use metavariables $\static{v},\static{V}$ 
to range over values,
which are the subset of terms consisting only of functions, function types and universes.
For evaluation, we use a call-by-value reduction semantics (\autoref{fig:syntax-static}).
Ascriptions are dropped when evaluating, and function applications result in (syntactic) substitution.
We refer to the values and semantics as \textit{simple} rather than \textit{static},
since they apply equally well to an untyped calculus, albeit without the same soundness guarantees.

\subsection{Comparing Types: Canonical Forms}
\label{subsec:static-canonical}

Since dependent types can contain expressions,
it is possible that types may contain redexes.
Most dependent type systems have a \textit{conversion} rule 
that assigns an expression type $\static{T}_{{\mathrm{1}}}$ 
if it has type $\static{T}_{{\mathrm{2}}}$,
and $\static{T}_{{\mathrm{2}}}$ is convertible to $\static{T}_{{\mathrm{1}}}$ through some sequence of $\beta$-conversions,
$\eta$-conversions, and $\alpha$-renamings.
Instead, we treat types as
$\abe$-equivalence classes.
To compare equivalence classes, we represent them using \textit{canonical forms}~\citep{hereditary},
denoted with metavariables $\static{u}$ and $\static{U}$.
These are $\beta$-reduced, $\eta$-long canonical members of an equivalence class. 
We compare terms for $\abe$-equivalence by normalizing and syntactically comparing their canonical forms.

The syntax for canonical forms is given in \autoref{fig:static-types}.
We omit well-formedness rules for terms and environments, since the only difference from 
the typing rules is the $\eta$-longness check.

By representing function applications in \textit{spine form}~\citep{linearSpineCalc}, 
we can ensure that all heads are variables, and thus no redexes are present, even under binders.
The well-formedness of canonical terms is ensured using bidirectional typing~\citep{localTypeInference}.
 An \textit{atomic form} can be a universe $ \TypeType_{ \ottnt{i} } $,
or a variable $x$ applied to 0 or more arguments, which we refer to as its \textit{spine}.  
 Our well-formedness rules ensure the types of
atomic forms are themselves atomic. This ensures that canonical forms are $\eta$-long,
since they cannot have type $\ottsym{(}  \mathit{y}  \ottsym{:}  \static{U}_{{\mathrm{1}}}  \ottsym{)}  \rightarrow  \static{U}_{{\mathrm{2}}}$.

\setbool{ShowArrowIndex}{false}
\begin{figure}
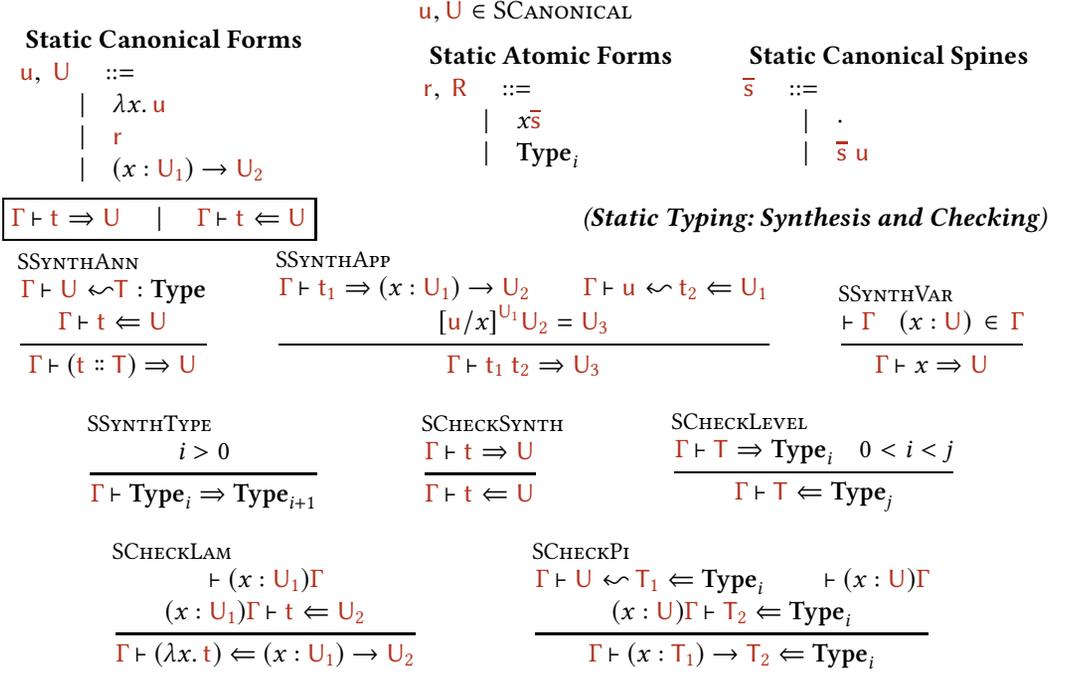

    $\static{u},\static{U} \in \AllTerms$
    
    \setbool{ShowEmptyDot}{true}
    \begin{minipage}[t]{0.38\textwidth} 
    \nonterms{u}  
    \end{minipage}  
    \begin{minipage}[t]{0.3\textwidth}
        \nonterms{rr} 
        \end{minipage}  
        \begin{minipage}[t]{0.3\textwidth}
            \nonterms{e} 
            \end{minipage}  
    \setbool{ShowEmptyDot}{false} 
    \setbool{ShowArrowIndex}{true}

\drules{$\static{Gamma}  \vdash  \static{t}  \Rightarrow  \static{U} \quad\mid\quad \static{Gamma}  \vdash  \static{t}  \Leftarrow  \static{U}$}{\ottcom{Static Typing: Synthesis and Checking}}
{SSynthAnn,SSynthApp,  SSynthVar, SSynthType,  
SCheckSynth, SCheckLevel,  SCheckLam, SCheckPi}    
   
\caption{\slang: Canonical Forms and Typing Rules}   
\label{fig:static-types}    
\end{figure} 
 
\subsection{Typechecking and Normalization}

Using the concept of canonical forms, we can now express the type rules for \slang
in \autoref{fig:static-types}. To ensure syntax-directedness, we again use bidirectional typing.

The \textit{type synthesis} judgement $\static{Gamma}  \vdash  \static{t}  \Rightarrow  \static{U}$ says that $\static{t}$ has type $\static{U}$
under context $\static{Gamma}$, where the type is treated as an output of the judgement. That is,
from examining the term, we can determine its type.
Conversely, the \textit{checking} judgment $\static{Gamma}  \vdash  \static{t}  \Leftarrow  \static{U}$ says that, given a type $\static{U}$, we can confirm that $\static{t}$ has that type.
These rules allow us to propagate the information from ascriptions inwards, so that only top-level terms and redexes
need ascriptions.

Most rules in the system are standard. To support dependent types,
\rrule{SSynthApp} computes the result of
applying a particular value. 
We switch between checking and synthesis using \rrule{SSynthAnn} and \rrule{SCheckSynth}.
The predicativity of our system is distilled in the \rrule{SSynthType} rule: $ \TypeType_{ \ottnt{i} } $ always has type $ \TypeType_{ \ottnt{i}  + 1} $.
The rule \rrule{SCheckLevel} encodes \textit{cumulativity}:
we can always treat types as if they were \added{at} a higher \added{level}, though the converse does not hold.
This allows us to check function types against any $ \TypeType_{ \ottnt{i} } $ in \rrule{SCheckPi}, provided the domain and codomain
 check against that $ \TypeType_{ \ottnt{i} } $.  

We distinguish \textit{hereditary substitution} on canonical forms $ {[  \static{u}_{{\mathrm{1}}}  / { \mathit{x} } ]}^{ \static{U}  }  \static{u}_{{\mathrm{2}}} =  \static{u}_{{\mathrm{3}}} $,
from \textit{syntactic substitution} $ [  \mathit{x}  \Mapsto  \static{t}_{{\mathrm{1}}}  ]  \static{t}_{{\mathrm{2}}}  = \static{t}_{{\mathrm{3}}}$ on terms.
Notably, the former takes the type of its variable as input, and has canonical forms
as both inputs and as output. 
In \rrule{SSynthApp} and \rrule{SCheckPi}, we use the \textit{normalization} judgement $\static{Gamma}  \vdash  \static{u}  \leadsfrom\  \static{t}  \Leftarrow  \static{U}$,
which computes the canonical form of $\static{t}$ while checking it against $\static{U}$.
Similarly, \rrule{SSynthAnn} uses the judgement  $ \static{Gamma}  \vdash  \static{U}  \leadsfrom  \static{T}  : \TypeType $,
 which uses hereditary substitution to compute the canonical form of $\static{T}$
while ensuring it checks against some $ \TypeType_{ \ottnt{i} } $.    

The rules for normalization (\autoref{fig:static-hsub}) directly mirror those for well-typed terms,
building up the canonical forms from sub-derivations.
In particular, the rule \rrule{SNormSynthVar}
$\eta$-expands any variables with function types, which allows us to assume that the function
in an application will always normalize to a $\lambda$-term.
(The rules for the eta expansion function $  \mathit{x}  \leadsto_\eta  \static{u} :  \static{U} $ are standard, so we omit them).
We utilize this assumption in \rrule{SNormSynthApp}, where the canonical form of an application is
computed using hereditary substitution.

\subsection{Hereditary Substitution}
\label{subsec:static-hsub} 
 
\begin{figure}
    \drules{$ \static{Gamma}  \vdash  \static{U}  \leadsfrom  \static{T}  : \TypeType $}{\ottcom{Type Normalization with Unknown Level (rules omitted)}}{} 
    \drules{$\static{Gamma}  \vdash  \static{t}  \leadsto  \static{u}  \Rightarrow  \static{U} \quad\mid\quad \static{Gamma}  \vdash  \static{u}  \leadsfrom\  \mathit{y}  \Leftarrow  \static{U}$}{\ottcom{Static Normalization}}{SNormSynthAnn, SNormSynthVar, SNormSynthApp}

    \ottdefnSHsub  
    \\ 

    \ottdefnSHsubR  

    \caption{\slang: Normalization (select rules) and Hereditary Substitution}
    \label{fig:static-hsub}  
\end{figure} 

Hereditary substitution is defined in \autoref{fig:static-hsub}.
At first glance, many of the rules look like a traditional substitution definition.
They traverse the expression looking for variables, and replace them with the corresponding term.

However, there are some key differences. 
Hereditary substitution has canonical forms as both inputs and outputs.
The key work takes place in the rule \rrule{SHsubRSpine}.
When replacing $\mathit{x}$ with $\static{u}_{{\mathrm{1}}}$ in $ \mathit{x}  \static{e} \  \static{u}_{{\mathrm{2}}}  $, find the substituted forms of $\static{u}_{{\mathrm{2}}}$ and $ \mathit{x} \static{e} $, which we call $\static{u}_{{\mathrm{3}}}$ and $\lambda  \mathit{y}  \ldotp  \static{u}'_{{\mathrm{1}}}$
respectively.
If the inputs are well-typed and $\eta$-long, the substitution of the spine will always return a $\lambda$-term, meaning that its application to $\static{u}_{{\mathrm{3}}}$
is not a canonical form.
To produce a canonical form in such a case, we continue substituting, recursively replacing $\mathit{y}$ with $\static{u}_{{\mathrm{3}}}$ in $\static{u}'_{{\mathrm{1}}}$\added{.}
A similar substitution in  the codomain of $\static{U}$ gives our result type.
Thus, if this process terminates, it will always produce a canonical form.

To ensure that the process does, in fact, terminate for well-typed inputs, we define hereditary substitution
in terms of the type of the variable being replaced. 
Since we are replacing a different variable in the premise \rrule{SHsubRSpine},
we must keep track of the type of the resultant expression when substituting in spines,
which is why substitution on atomic forms is a separate relation.
We order types by the
multiset of universes of all arrow types that are subterms of the type, similar to techniques used for Predicative System F~\citep{predSysF, eades2010hereditary}.
We can use the well-founded multiset ordering given by \citet{multisetOrder}:
if a type $\static{U}$ has maximum arrow type universe $\ottnt{i}$,
we say that it is greater than all other types containing fewer arrows at universe $\ottnt{i}$
whose maximum is not greater than $\ottnt{i}$. 
Predicativity ensures that, relative to this ordering, the return type of a function application
is always less than the type of the function itself. 
In all premises but the last two of \rrule{SHsubRSpine},
we recursively invoke substitution on strict subterms, while keeping the type of the variable the same.
In the remaining cases,  we perform substitution at a type
that is smaller by our multiset order.

\subsection{Properties of \slang}

Since \slang is mostly standard, it enjoys the standard properties of dependently-typed languages.
Hereditary substitution can be used to show that the language is strongly normalizing,
and thus consistent as a logic. Since the type rules, hereditary substitution, and normalization are syntax directed and terminating,
typechecking is decidable. Finally, because all well-typed terms have canonical forms, \slang is type \replaced{safe}{sound}.

\section{\lang: Abstracting the Static Language} 
\label{sec:graduallang}

We now present \lang, a gradual counterpart to \slang derived following the Abstracting Gradual Typing (AGT) methodology~\cite{agt}, extended to the setting of dependent types. 
\added{The key idea behind AGT is that gradual type systems can be designed by first specifying the \textit{meaning} of gradual types
in terms of sets of static types.
This meaning is given as a concretization function $\gamma$ that maps a gradual type to the set of static types that it represents,
and an abstraction function $\alpha$ that recovers the {\em most precise} gradual type that represents a given set of static types. In other words, $\gamma$ and $\alpha$ form a {\em Galois connection}.

Once the meaning of gradual types is clear, the typing rules and dynamic semantics for the gradual language can be derived systematically.}
\rev{ 
First, $\gamma$ and $\alpha$ allow us to lift the type predicates and type functions used in the static type system (such as equality, subtyping, join, etc.) to obtain their gradual counterparts. From these definitions, algorithmic characterizations can then be validated and implemented.
Second, the gradual type system is obtained from the static type system by using these lifted type predicates and functions.
Finally, the runtime semantics follow by proof reduction of the typing derivation, mirroring the type safety argument at runtime. In particular, typing derivations are augmented with pieces of {\em evidence} for consistent judgments, whose combination during reduction may be undefined, hence resulting in a runtime type error.
}

\added{In this work we follow the AGT methodology, specifying $\gamma$ and $\alpha$, then describing how the typing rules
are lifted to gradual types. In doing so, we uncover several points for which the standard AGT approach lacks the flexibility to accommodate full-spectrum dependent types with $\qm$ as a term.
We describe our extensions to (and deviations from) the AGT methodology, and how they allow us to fully support gradual dependent types.
}

Throughout this section, we assume that we have gradual versions of hereditary substitution and normalization. We leave  the detailed development of these notions to \autoref{sec:gradualhsub}, as they are non-trivial if one wants to preserve both decidable typechecking and the gradual guarantee (\autoref{subsec:gg-intro}). The dynamic semantics of \lang are presented in \autoref{sec:gradual-semantics}, and its metatheory in \autoref{sec:lang-properties}.

\subsection{Terms and Canonical Forms}
 
\subsubsection*{Syntax} 
\setbool{ShowArrowIndex}{true} 

\begin{figure}
    $\gradual{u},\gradual{U} \in \AllGTerms$\\
    \setbool{ShowEmptyDot}{true}
    \begin{minipage}[t]{0.33\textwidth} 
    \nonterms{gt} 
    \end{minipage}
    \rev{
    \begin{minipage}[t]{0.33\textwidth}
    \nonterms{canonical}
    \end{minipage} }
    \begin{minipage}[t]{0.3\textwidth}
    \nonterms{atomic,spine} 
    \end{minipage} 
    \setbool{ShowEmptyDot}{false}

    \caption{\lang: Terms and Canonical Forms} 
    \label{fig:syntax-gradual} 
    \end{figure}

The syntax of \lang (\autoref{fig:syntax-gradual}) is a simple extension
of \slang's syntax. We use $\gradualstyle{\text{blue, serif font}}{}{}$ to for metavariables denoting gradual terms, contexts, etc. 
In addition to constructs from \slang, \lang's syntax includes $ {\qm } $,
the unknown term/type.
This represents a type or term that is unknown to the programmer: by annotating a type with $ {\qm } $
or leaving a term as $ {\qm } $, they can allow their program to typecheck
with only partial typing information or an incomplete proof. Similarly, $ {\qm } $ is added to the syntax of canonical values.
 
\added{Additionally, arrow-types are annotated with a level $i$.} 
\rev{ The type $(x : \gradual{U}_1) \xrightarrow{i} \gradual{U}_2$ 
is well formed at type $\TypeType_i$,
and we have a special top level $\omega$ where $(x : \gradual{U}_1) \xrightarrow{\omega} \gradual{U}_2$ is well formed at type $?$.
These annotations are necessary for ensuring the
termination of hereditary substitution, but are inferred during normalization, and are never present in source programs.
We often omit these annotations,
as they clutter the presentation. }

Canonical forms do not contain ascriptions.
While statically-typed languages use ascriptions only for guiding typechecking, the potential for dynamic type failure
means that ascriptions have computational content in gradual typing.
Notably, only variables \added{or neutral applications} can synthesize $ {\qm } $ as a type, though
any typed expression can be checked against $ {\qm } $.
This allows us to reason about canonical forms at a given type:
while we can layer ascriptions on terms, such as $ \ottsym{(}   \mathtt{true} \dblcolon {\qm }   \ottsym{)} \dblcolon  {\qm } \to {\qm }  $,
the only canonical forms with function types are lambdas and $ {\qm } $.

\subsection{Concretization and Predicates}
\label{subsec:concretization}

\begin{figure}
    $\gamma :  \AllGTerms \to \Pow(\AllTerms) \setminus \emptyset$
\\
    \begin{minipage}{0.33\textwidth}
    \begin{align*} 
        &\gamma( {\qm } ) & = & \AllTerms\\ 
        &\gamma( \TypeType_{ \ottnt{i} } ) & = & \set{ \TypeType_{ \ottnt{i} } } \\
        &\gamma(\mathit{x}) & = & \set{x} 
    \end{align*}    
\end{minipage}
\begin{minipage}{0.6\textwidth}
    \begin{align*}
        &\gamma(\lambda  \mathit{x}  \ldotp  \gradual{u}) & = & \set{\lambda  \mathit{x}  \ldotp  \static{u}' \mid \static{u}' \in \gamma(\gradual{u})}\\
        &\gamma( \mathit{x}  \gradual{e} \  \gradual{u}  ) & = &  \set{ \mathit{x}  \static{e}' \  \static{u}'   \mid  \mathit{x} \static{e}'  \in \gamma( \mathit{x} \gradual{e} ), \static{u}' \in \gamma(\gradual{u})}\\
        &\gamma(\ottsym{(}  \mathit{x}  \ottsym{:}  \gradual{U}_{{\mathrm{1}}}  \ottsym{)}  \rightarrow  \gradual{U}_{{\mathrm{2}}}) &=& \set{\ottsym{(}  \mathit{x}  \ottsym{:}  \static{U}'_{{\mathrm{1}}}  \ottsym{)}  \rightarrow  \static{U}'_{{\mathrm{2}}} \mid \static{U}'_{{\mathrm{1}}} \in \gamma(\gradual{U}_{{\mathrm{1}}}), \static{U}'_{{\mathrm{2}}} \in \gamma(\gradual{U}_{{\mathrm{2}}})}
    \end{align*}    
\end{minipage}

    \ottdefnConsistent

    \caption{\lang: Concretization and Consistency}
    \label{fig:concretization} 
    \end{figure}

The main idea of AGT is that gradual types abstract sets of static types,
and that each gradual type can be made \textit{concrete} as a set of static types.
For our system, we simply extend this to say that gradual terms represent sets of static terms.
In a simply typed language, a static type embedded in the gradual language concretizes to the singleton
set containing itself. However, for terms, we wish to consider the entire $\abe$-equivalence
class of the static term. As with typechecking, this process is facilitated by considering only canonical forms.
The concretization function $\gamma : \AllGTerms \to \Pow(\AllTerms)$, defined in \autoref{fig:concretization},
recurs over sub-terms, with $ {\qm } $ mapping to the set of all terms.

Given the concretization, we can \textit{lift} a predicate from the static system to the gradual system.
A predicate holds for gradual types if it holds for some types in their concretizations.
For equality, this means that $ \gradual{U}  \cong  \gradual{U}' $ if and only if $\gamma(\gradual{U}) \cap \gamma(\gradual{U}') \neq \emptyset$.
We present a syntactic version of this in \autoref{fig:concretization}.
Concretization also gives us a notion of \textit{precision} on gradual types.
We say that $ \gradual{U} \sqsubseteq  \gradual{U}' $ if $\gamma(\gradual{U}) \subseteq \gamma(\gradual{U}')$: that is, $\gradual{U}$ is more precise
because there are fewer terms it could plausibly represent.
We can similarly define $ \gradual{U}  \sqcap  \gradual{U}'  $ as
the most general term that is as precise as both $\gradual{U}$ and $\gradual{U}'$.
Note that $ \gradual{U}  \sqcap  \gradual{U}' $ is defined if and only if $ \gradual{U}  \cong  \gradual{U}' $ i.e. if $\gamma(\gradual{U})\cap\gamma(\gradual{U}')\neq\emptyset$, \added{and like the consistency relation, it can be computed syntactically.}

\subsection{Functions and Abstraction}  
\label{subsec:abstraction}

\begin{figure}
        $\alpha : \Pow(\AllTerms) \setminus \emptyset \to \AllGTerms$

\begin{minipage}{0.95\textwidth}
    \begin{align*}
        \alpha(\set{ \mathit{x}  \static{e} \  \static{u}   \mid  \mathit{x} \static{e}  \in A, \static{u} \in B}) & = &&  \mathit{x}  \gradual{e}' \  \gradual{u}'   & \text{ where } \alpha(A) =  \mathit{x} \gradual{e}' , \alpha(B) = \gradual{u}'&\\
        \alpha(\set{\ottsym{(}  \mathit{x}  \ottsym{:}  \static{U}_{{\mathrm{1}}}  \ottsym{)}  \rightarrow  \static{U}_{{\mathrm{2}}} \mid \static{U}_{{\mathrm{1}}} \in A, \static{U}_{{\mathrm{2}}} \in B }) & = && \ottsym{(}  \mathit{x}  \ottsym{:}  \gradual{U}'_{{\mathrm{1}}}  \ottsym{)}  \rightarrow  \gradual{U}'_{{\mathrm{2}}}  & \text{ where } \alpha(A) = \gradual{U}'_{{\mathrm{1}}}, \alpha(B) = \gradual{U}'_{{\mathrm{2}}}&\\
        \alpha(\set{\lambda  \mathit{x}  \ldotp  \static{u} \mid \static{u} \in A} ) &= && \lambda  \mathit{x}  \ldotp  \gradual{u}' & \text{ where } \alpha(A) = \gradual{u}'&
    \end{align*} 
\end{minipage}

$\alpha(\set{ \TypeType_{ \ottnt{i} } })\ =\  \TypeType_{ \ottnt{i} }  \qquad\qquad\qquad \alpha(\set{\mathit{x}})\ =\ \mathit{x} \qquad\qquad\qquad \alpha(A)\  =\   {\qm }   \text{ otherwise } $

    \begin{minipage}{0.95\textwidth}
    \fbox{$\ottkw{dom} : \AllGTerms \pto \AllGTerms$}\qquad
    \fbox{$[\square / \_ ]^\square\ \ottkw{cod} \square : \AllGTerms^3 \pto \AllGTerms \quad$}
    \end{minipage}
    \drules{
        $[\square / \_ ]^\square\ \ottkw{body} \square : \AllGTerms^3 \pto \AllGTerms$ 
        }{\ottcom{Partial Functions}}{DomainPi,  CodSubPi,  BodySubPi, DomainDyn, CodSubDyn,  BodySubDyn}
    \caption{\lang: Abstraction and Lifted Functions}
    \label{fig:abstraction} 
    \end{figure} 

When typechecking a function application, we must handle the case where the function has type $ {\qm } $.
Since $ {\qm } $ is not an arrow type, the static version of the rule would fail in all such cases.
Instead, we extract the domain and codomain from the type using partial functions.
Statically, $ \ottkw{dom}\  \ottsym{(}  \mathit{x}  \ottsym{:}  \static{U}  \ottsym{)}  \rightarrow  \static{U}' =  \static{U} $, and is undefined otherwise.
But what should the domain of $ {\qm } $ be?

AGT gives a recipe for lifting such partial functions. To do so, we need the
counterpart to concretization: abstraction. The abstraction function $\alpha$
is defined in \autoref{fig:abstraction}.
 It takes a set of static terms,
and finds the most precise gradual term that is consistent with the entire set.
Now, we are able to take gradual terms to sets of static terms, then back to
gradual terms. It is easy to see that $\alpha(\gamma(\gradual{U})) = \gradual{U}$:
\added{they are normal forms describing the same equivalence classes}.
This lets us define our partial functions in terms of their static counterparts:
we concretize the inputs, apply the static function element-wise on all values of the concretization
for which the function is defined,
then abstract the output to obtain a gradual term as a result.

For example, the domain of a gradual term $\gradual{U}$ is
$\alpha(\set{\ottkw{dom}\ \static{U}' \mid \static{U}'\in\gamma(\gradual{U})})$,
which can be expressed algorithmically using the rules in \autoref{fig:abstraction}.
We define function-type codomains and lambda-term bodies similarly, though we pair
these operations with substitution to avoid creating a ``dummy'' bound  variable name for $ {\qm } $.

Taken together, $\alpha$ and $\gamma$ form a Galois connection, which ensures that
our derived type system is a conservative extension of the static system.

\subsection{Typing Rules} 

\begin{figure}
    \rev{
    \drules{$\gradual{Gamma}  \vdash  \gradual{t}  \Rightarrow  \gradual{U} \quad\mid\quad \gradual{Gamma}  \vdash  \gradual{t}  \Leftarrow  \gradual{U}$}{Well-Typed Gradual Terms}{GSynthAnn, GSynthType, GSynthVar, GSynthDyn, GSynthApp,  GCheckPi,   GCheckLamPi,  GCheckLamDyn, GCheckSynth, GCheckLevel}  
     \added{}   
    }
    \caption{\lang: Type Checking and Synthesis}
    \label{fig:typerules-gradual} 
    \end{figure}  

Given concretization and abstraction, AGT gives a recipe for converting a static type system into a gradual one, and we follow it closely. 
\Autoref{fig:typerules-gradual} gives the rules for typing.
Equalities implied by repeated metavariables have been replaced by consistency checks, such as in \rrule{GCheckSynth}.
Similarly, in \rrule{GCheckPi} we use the judgment $ \gradual{U}  \cong \TypeType $ to ensure that the given type is consistent to,
rather than equal to, some $ \TypeType_{ \ottnt{i} } $.
Rules that matched on the form of a synthesized type instead use partial functions, as we can see in
\rrule{GSynthApp}. We split the checking of functions into \rrule{GCheckLamPi} and \rrule{GCCheckLamDyn} for clarity,
but the rules are equivalent to a single rule using partial functions.
\rev{In \rrule{GSynthAnn}, the  judgment $\gradual{Gamma}  \vdash  \gradual{U}  \leadsfrom  \gradual{T}  : \TypeType_{\Rightarrow  \ottnt{i} }$ \added{denotes \textit{level synthesis}}, where we normalize $\gradual{U}$ while inferring at what universe level it resides.}

\added{We note that while $\gamma$ and $\alpha$ are crucial for deriving the definitions of gradual operations,
the operations can be implemented algorithmically as syntactic checks; an implementation does not need to compute $\gamma$ or $\alpha$.}
\rev{
Also, because $\gamma(x)=\{x\}$ for any variable $x$, consistency, precision and meet are all well-defined on open terms.
Consistency corresponds to the gradual lifting of \textit{definitional equality}: $\gradual{u}_1 \cong \gradual{u}_2$
if and only if there is some $\static{u}'_1 \in \gamma(\gradual{u}_1)$ and $\static{u}'_2 \in \gamma(\gradual{u}_2)$
where $\static{u}'_1 =_\abe \static{u}'_2 $.
This reflects our intensional approach: 
functions are consistent if their bodies are 
consistent.
}

We wish to allow the unknown term $ {\qm } $ to replace any term in a program.
But what should its type be? By the AGT philosophy, $ {\qm } $ represents all terms, so it should synthesize the abstraction of all inhabited types,
which is $ {\qm } $. We encode this in the rule \rrule{GSynthDyn}.
This means that we can use the unknown term in any context.

 As with the static system, we represent types in canonical form,
 which makes consistency checking easy.
 Well-formedness rules (omitted) are derived from the static system
 in the same way as the gradual type rules.
 Additionally, the gradual type rules rely on the \textit{gradual normalization} judgments,
 $\gradual{Gamma}  \vdash  \gradual{t}  \leadsto  \gradual{u}  \Rightarrow  \gradual{U}$ and $\gradual{Gamma}  \vdash  \gradual{u}  \leadsfrom\  \gradual{t}  \Leftarrow  \gradual{U}$,
 which we explain in \autoref{subsec:gradualnorm}.

\subsection{Example: Typechecking $\mathtt{head}$ of $\mathtt{nil}$}
\label{subsec:type-headnil}

To illustrate how the \lang type system works,
we explain the typechecking of one example from the introduction.
Suppose we have types for natural numbers and vectors,
and a derivation for \\
$\gradual{Gamma} \vdash \mathtt{head} \ :\  \ottsym{(}  \mathtt{A}  \ottsym{:}   \TypeType_{ \ottsym{1} }   \ottsym{)}  \rightarrow  \ottsym{(}  \mathtt{n}  \ottsym{:}   \ottkw{Nat}   \ottsym{)}  \rightarrow      \ottkw{Vec}  \  \mathtt{A}  \  \ottsym{(}   \mathtt{n}  + 1   \ottsym{)}  \to \mathtt{A} $.
In \autoref{fig:vec-deriv}, we show the (partial) derivation of $\gradual{Gamma} \vdash     \mathtt{head} \   \ottkw{Nat}   \   0   \  \ottsym{(}     \ottsym{(}    \ottkw{Nil}  \   \ottkw{Nat}    \ottsym{)} \dblcolon  \ottkw{Vec}   \   \ottkw{Nat}   \  {\qm }   \ottsym{)}    \Rightarrow   \ottkw{Nat} $.

The key detail here is that the compile-time consistency check lets us compare $ 0 $ to $ {\qm } $,
and then $ {\qm } $ to $ 1 $, which allows the example to typecheck.
Notice how we only check consistency when we switch from checking to synthesis.
While this code typechecks, it fails at runtime. We step through its execution in \autoref{subsec:nilhead-run}.

\section{Approximate Normalization} 
\label{sec:gradualhsub}

\begin{figure}[t!]
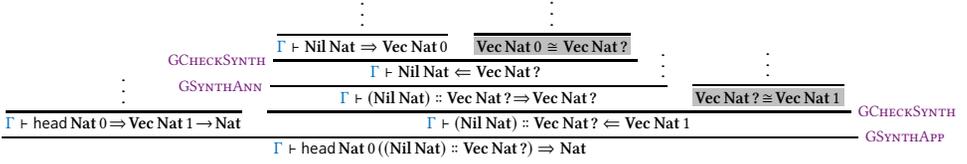

    \tiny
\begin{mathpar}
\inferrule*[Right = \tiny\rrule{GSynthApp}]{
    \inferrule*{
        \vdots
    }
    {
        \gradual{Gamma} \vdash    \mathtt{head} \,   \ottkw{Nat}   \,   0   \!  \Rightarrow  \!    \ottkw{Vec}  \,   \ottkw{Nat}   \,   1  \! \to \! \ottkw{Nat}  
    } \\
    \inferrule*[Right = \tiny\rrule{GCheckSynth}]{
        \inferrule*[Left = \tiny\rrule{GSynthAnn}]{
            \inferrule*[Left = \tiny\rrule{GCheckSynth}]{
                \inferrule*{
                    \vdots
                }
                {
                    \gradual{Gamma} \vdash    \ottkw{Nil}  \,   \ottkw{Nat}     \Rightarrow     \ottkw{Vec}  \,   \ottkw{Nat}   \,   0  
                }
                \\
                \inferrule*{\vdots}{
            \gbox{   \ottkw{Vec}  \,   \ottkw{Nat}   \,   0    \cong     \ottkw{Vec}  \,   \ottkw{Nat}   \,  {\qm } }
        }
            }
            {
                \gradual{Gamma} \vdash    \ottkw{Nil}  \,   \ottkw{Nat}     \Leftarrow     \ottkw{Vec}  \,   \ottkw{Nat}   \,  {\qm } 
            }
            \\
            \vdots
        }
        {
            \gradual{Gamma} \vdash     \ottsym{(}    \ottkw{Nil}  \,   \ottkw{Nat}    \ottsym{)} \dblcolon  \ottkw{Vec}   \,   \ottkw{Nat}   \,  {\qm }    \! \Rightarrow \!    \ottkw{Vec}  \,   \ottkw{Nat}   \,  {\qm } 
        }
        \\
        \inferrule*{\vdots}{
            \gbox{   \ottkw{Vec}  \,   \ottkw{Nat}   \,  {\qm } \!   \cong  \!   \ottkw{Vec}  \,   \ottkw{Nat}   \,   1  } 
        }
    }
    {
        \gradual{Gamma} \vdash     \ottsym{(}    \ottkw{Nil}  \,   \ottkw{Nat}    \ottsym{)} \dblcolon  \ottkw{Vec}   \,   \ottkw{Nat}   \,  {\qm }    \Leftarrow     \ottkw{Vec}  \,   \ottkw{Nat}   \,   1  
    }
}
{
    \gradual{Gamma} \vdash     \mathtt{head} \,   \ottkw{Nat}   \,   0   \,  \ottsym{(}     \ottsym{(}    \ottkw{Nil}  \,   \ottkw{Nat}    \ottsym{)} \dblcolon  \ottkw{Vec}   \,   \ottkw{Nat}   \,  {\qm }   \ottsym{)}    \Rightarrow   \ottkw{Nat} 
}
\end{mathpar}
\caption{Type Derivation for head of nil}
\label{fig:vec-deriv}
\end{figure}

In the previous example, normalization was used to compute the type of $  \mathtt{head} \   \ottkw{Nat}   \   0  $, replacing $\mathtt{n}$ with $ 0 $ in the type of $\mathtt{head}$,
normalizing $  0   + 1 $ to $ 1 $. This computation is trivial,
but not all are. 
As we saw in \autoref{subsec:factorial}, the type-term overlap in \lang
means that code that is run during typechecking may fail or diverge.  

A potential solution would be to disallow imprecisely typed code
in type indices. However, this approach breaks the criteria for a gradually-typed language.
In particular, it would result in a language that
violates the static gradual guarantee (\autoref{subsec:gg-intro}).
The static guarantee implies that if a program does not typecheck, the programmer knows that the problem
is not the absence of type precision, but that the types present are fundamentally wrong.
Increasing precision in multiple places will never cause a program to typecheck if doing so in one place fails.

In this section, we present two versions of gradual substitution.
First, we provide \textit{ideal substitution}, which is well defined on all terms,
but for which equality is undecidable.
Second, we describe \textit{approximate hereditary substitution},
which regains decidability while preserving the gradual guarantee,
by producing compile-time canonical forms that are potentially less precise than their runtime counterparts.
Thus, we trade precision for a termination guarantee.
From this, we build \textit{approximate normalization}, which uses hereditary substitution
to avoid non-termination, and avoids dynamic failures by normalizing
certain imprecise terms to $ {\qm } $.

A key insight of this work is that we need separate notions of 
\textit{compile-time normalization} and \textit{run-time execution}.
That is, we use approximate hereditary substitution \textit{only}
in our types. Executing our programs at run-time will not lose information,
but it may diverge or fail.

For typechecking, the effect of this substitution is that
non-equal terms of the unknown type may be indistinguishable at compile-time.
Returning to the example from \autoref{subsec:factorial}, the user's faulty factorial-length vector
will typecheck, but at type $   \ottkw{Vec}  \   \ottkw{Nat}   \  {\qm } $. Using it
will never raise a static error due to its length, but
it may raise a runtime error.

\subsection{Ideal Substitution}

Here, we present a definition of gradual substitution for $\abe$-equivalence classes of terms.
While comparing equivalence classes is undecidable, we will use ideal substitution as the theoretical foundation,
showing that our approximate substitution produces the same results as ideal substitution, save for
some loss of precision. 

The main difficulty with lifting the definition of hereditary substitution is that
the set of terms with a canonical form is only closed under hereditary substitution when we assume a static type discipline.
The terms $ \mathit{y} \  \mathit{y} $ and $ \lambda  \mathit{x}  \ldotp  \mathit{x} \  \mathit{x} $ are both syntactically canonical, but if we substitute the second in for $\mathit{y}$,
there is no normal form. However, both of these terms can be typed in our gradual system.
How can  $[ \ottsym{(}   \lambda  \mathit{x}  \ldotp  \mathit{x} \  \mathit{x}   \ottsym{)}/\mathit{y} ]^{ {\qm } }  \mathit{y} \  \mathit{y}   $ be defined?

If we apply the AGT lifting recipe to hereditary substitution, we get a function that may not
have a defined output for all gradually well-typed canonical inputs. Even worse is that determining whether
substitution is defined for an input is undecidable.
By AGT's formulation, $[ \gradual{u}/\mathit{x} ]^{ {\qm } } \gradual{u}' $ would be 
$\alpha(\set{ [ \static{u}_{{\mathrm{1}}} / x]^{\static{U}} \static{u}_{{\mathrm{2}}} \mid \static{u}_{{\mathrm{1}}} \in \gamma(\gradual{u}), \static{u}_{{\mathrm{2}}} \in \gamma(\gradual{u}'), \static{U}\in\AllTerms })$.
To compute $\alpha$, we must know which of the concretized results are defined,
i.e. find all pairs in $\gamma(\gradual{u})\times\gamma(\gradual{u}')$
for which there exists some $\static{U}$ on which static hereditary substitution is defined.
This means determining if there is any finite number of substitutions
under which the substitution on a (possibly dynamically-typed) term is defined, which requires solving the Halting Problem.

Recall that we introduced canonical forms in \autoref{subsec:static-canonical} to uniquely represent
$\abe$-equivalence classes.
While canonical forms are not closed under substitutions, equivalence classes are.
Going back to our initial example, what we really want is for
$[ \ottsym{(}   \lambda  \mathit{x}  \ldotp  \mathit{x} \  \mathit{x}   \ottsym{)}/\mathit{y} ]^{ {\qm } }  \mathit{y} \  \mathit{y} $  to be  $(\ottsym{(}   \lambda  \mathit{x}  \ldotp  \mathit{x} \  \mathit{x}   \ottsym{)}\ottsym{(}   \lambda  \mathit{x}  \ldotp  \mathit{x} \  \mathit{x}   \ottsym{)})^\abe $,
i.e. the set of all terms $\abe$-equivalent to $\Omega$.

Thus we define ideal substitution on $\abe$-equivalence classes themselves.
For this, we do not need hereditary substitution: if $s \in s^\abe$ and $t\in t^\abe$ are terms with their
respective equivalence classes, the substitution $[x \Mapsto s^\abe]t^\abe$ is simply the equivalence class of
$[x \Mapsto s]t$. 
We now have a total operation from equivalence classes to equivalence classes.
These classes may have no canonical representative, but the function is defined regardless.
If we extend  concretization and abstraction to be defined on equivalence classes, this gives
us the definition of ideal substitution:

$$
    [ \mathit{x} \Mapsto \gradual{t}_{{\mathrm{1}}}^\abe ] \gradual{t}_{{\mathrm{2}}}^\abe = \alpha(\{  [ \mathit{x} \Mapsto \static{t}'_{{\mathrm{1}}} ] \static{t}'_{{\mathrm{2}}} \mid  
      \static{t}'_{{\mathrm{1}}}  \in \gamma(\gradual{t}_{{\mathrm{1}}}), \static{t}'_{{\mathrm{2}}}  \in \gamma(\gradual{t}_{{\mathrm{2}}}) \} )  
$$

That is, we find the concretization of the gradual equivalence classes, which are sets of static equivalence classes.
We then substitute in each combination of these by taking the substitution of
a representative element, and abstract over
this set to obtain a single gradual equivalence class.

\subsection{Approximate Substitution}  
 
\rev{
\begin{figure}
        
    \ottdefnGHsubR         
     
    \caption{\lang: \rev{Approximate Substitution (select rules)}}  
    \label{fig:gradual-hsub}  
    \end{figure}  
}

With a well-defined but undecidable substitution, we now turn to the problem
of how to recover decidable comparison for equivalence classes, without losing the gradual guarantees.
We again turn to (gradual) canonical forms
as representatives of $\abe$-equivalence classes.
What happens when we try to construct a hereditary substitution function syntactically,
as in \slang? 

The problem is in adapting \rrule{SHsubRSpine}.
Suppose we are substituting $\gradual{u}$ for $\mathit{x}$ in $ \mathit{x}  \gradual{e} \  \gradual{u}_{{\mathrm{2}}}  $, 
and the result of substituting in $ \mathit{x} \gradual{e} $ is $\ottsym{(}  \lambda  \mathit{y}  \ldotp  \gradual{u}'  \ottsym{)} :  {\qm } $.
Following the AGT approach, we can use the $\ottkw{dom}$ function
to calculate the domain of $ {\qm } $, which is the type at which we substitute $\mathit{y}$.
But this violates the well-foundedness condition we imposed in the static case!
Since the domain of $ {\qm } $ is $ {\qm } $, eliminating redexes may infinitely apply substitutions
without decreasing the size of the type.

In all other cases, we have no problem, since the term we are substituting into is structurally decreasing.
So, while equivalence classes give us our ideal, theoretical definition,
hereditary substitution provides us with the exact cases we must approximate in order to preserve decidability.
To guarantee termination, we must not perform recursive substitutions in spines with $ {\qm } $-typed heads.

There are two apparent options for how to proceed without making recursive calls: we either fail when we try to apply a $ {\qm } $-typed
function, or we return $ {\qm } $.
The former will preserve termination, but it will not preserve the static gradual guarantee.
Reducing the precision of a well-typed program's ascriptions should never yield ill-typed code.
If applying a dynamically-typed function caused failure, then changing an ascription to $ {\qm } $ could cause
a previously successful program to crash, violating the guarantee.

Our solution is to produce $ {\qm } $ when applying a function of type $ {\qm } $.
We highlight the changes to hereditary substitution in \autoref{fig:gradual-hsub}.
\rrule{GHsubRDynType} accounts for $ {\qm } $-typed functions, and \rrule{GHsubRDynSpine}
accounts for $ {\qm } $ applied as a function.

\added{We must add one more check to guarantee termination, because $\qm : \qm$ could be used to  
circumvent the universe hierarchy. For instance,}
\rev{we can assign $(x : \qm) \to (x\ \TypeType_{99})$ the type $\TypeType_1$, and
 we can even write a version Girard's Paradox~\citep{girard1972interpretation,coquand:inria-00076023} by using $\qm$ in place of $\TypeType$.
Because of this, \rrule{GHsubRLamSpine} manually checks our decreasing metric.}

\added{Concretely,  $i \prec \omega$ for every $i$, and }
\rev{
 $\gradual{U} \prec \gradual{U}'$ when the multiset of annotations on arrow types in $\gradual{U}$
is less than that of $\gradual{U}'$ by the well-founded multiset ordering given by \citet{multisetOrder}.
In the static case, the type of substitution is always decreasing for this metric.  In the presence of $\qm$,
we must check if the order is violated and return $\qm$ if it is, as seen in the rule \rrule{GHsubRLamSpineOrd}.
Unlike applying a function of type $\qm$, we believe that this case is unlikely to arise in practice
unless programmers are deliberately using $\qm$ to circumvent the universe hierarchy.}

\subsection{Approximate Normalization}  
\label{subsec:gradualnorm}  

\begin{figure}
    \rev{
    \drules{$\gradual{Gamma}  \vdash  \gradual{t}  \leadsto  \gradual{u}  \Rightarrow  \gradual{U} \quad\mid\quad \gradual{Gamma}  \vdash  \gradual{u}  \leadsfrom\  \gradual{t}  \Leftarrow  \gradual{U}$}{\ottcom{Approximate Normalization}}{
        GNSynthApp, GNCheckSynth, GNCheckApprox, GNCheckPiType, GNCheckPiDyn  }        
    }
    \caption{\lang: Approximate Normalization (select rules) }
    \label{fig:typing-norm-gradual}  
    \end{figure}

    While approximate hereditary substitution eliminates non-termination, we must still account for dynamic failures.
    We do so with \textit{approximate normalization} (\autoref{fig:typing-norm-gradual}).
    
    To see the issue, consider that we can type the term 
    $\ottsym{(}   \lambda  \mathtt{A}  \ldotp  \ottsym{(}    0  \dblcolon {\qm }   \ottsym{)} \dblcolon \mathtt{A}   \ottsym{)}$ 
    against $\ottsym{(}  \mathtt{A}  \ottsym{:}   \TypeType_{ \ottsym{1} }   \ottsym{)}  \rightarrow   \mathtt{A} $.
    However, there are no ascriptions in canonical form,
    since ascriptions can induce casts, which are a form of computation. 
    The term $\ottsym{(}  \lambda  \mathtt{A}  \ldotp   0   \ottsym{)}$ certainly does not type against 
    against $\ottsym{(}  \mathtt{A}  \ottsym{:}   \TypeType_{ \ottsym{1} }   \ottsym{)}  \rightarrow   \mathtt{A} $, since $\ottsym{0}$ will not check against the type variable $\mathtt{A}$.
    However, if we were to raise a type error, we would never be able to
    apply a function to the above term.
    In the context $\ottsym{(}  \mathtt{A}  \ottsym{:}   \TypeType_{ \ottsym{1} }   \ottsym{)}  \cdot$, the only canonical term with type $\mathtt{A}$ is $ {\qm } $.
    \added{For the term to have a} well-typed normal form, its body must be $ {\qm } $. 

    More broadly, normalization does not preserve \textit{synthesis} of typing, only checking.
    In the rule \rrule{GNCheckSynth}, if $\gradual{Gamma}  \vdash  \gradual{t} \leadsto \gradual{u}  \Rightarrow  \gradual{U}'$, then the normal form of $\gradual{t}$ might check against $\gradual{U}$,
    but it won't necessarily synthesize $\gradual{U}$ (or any type). We need to construct a canonical form $\gradual{u}$ for $\gradual{t}$ at type $\gradual{U}$,
    assuming we have some normal form $\gradual{u}'$ for $\gradual{t}$ at type $\gradual{U}'$.
    If $ \gradual{U} \sqsubseteq  \gradual{U}' $, $\gradual{u}'$ will check against $\gradual{U}$. Otherwise, the only value we can be sure will
    check against $\gradual{U}$ is $ {\qm } $, which checks against any type.
    We formalize this using a pair of rules: \rrule{GNCheckSynth} normalizes fully when we can do so in a type-safe way,
    and \rrule{GNCheckApprox} produces $ {\qm } $ as an approximate result in all other cases.
    
    Gradual typing must also treat $\eta$-expansion carefully. 
    We $\eta$-expand all variables in \rrule{GNSynthVar} (see appendix), but in \rrule{GNSynthApp},
    we may be applying a function of type $ {\qm } $. However, a variable $\mathit{x}$ of type $ {\qm } $ is $\eta$-long.
    Since we are essentially treating a value of type $ {\qm } $ as type $ {\qm } \to {\qm } $,
    we must expand $\mathit{x}$ to $\ottsym{(}   \lambda  \mathit{y}  \ldotp  \mathit{x} \  \mathit{y}   \ottsym{)}$.
    We do this in \rrule{GNSynthVar} with an extra $\eta$-expansion at type $ {\qm } \to {\qm } $,
    which expands a $ {\qm } $-typed term one level, and has no effect on a canonical form with a function type.

    \added{Normalization is also where we generate the annotations necessary for ensuring the decreasing metric of
    hereditary substitution. As we see in the rules} 
    \rev{ \rrule{GNCheckPiType} and \rrule{GNCheckPiDyn}, we annotate arrows 
    either with the level against which they are checked, or with $\omega$ when checking against $\qm$.}
    The remaining rules for normalization (omitted) directly mirror the rules from \autoref{fig:typerules-gradual}.
    $ \TypeType_{ \ottnt{i} } $, $ {\qm } $, and variables all normalize to themselves, and all other rules simply construct normal forms
    from the normal forms of their subterms.

    Some of the difficulty with normalization arises because function arguments are normalized before being substituted.
    One could imagine a language that normalizes after substituting function arguments,
    and typechecking fails if a dynamic error is encountered during normalization.
    Here, normalization could fail, but only on terms that had truly ill-formed types, since unused failing values would be discarded. 
    We leave the development of such a language to future work.

\subsection{Properties of Approximate Normalization}

\subsubsection*{Relationship to the Ideal}

If we expand our definition of concretization to apply to equivalence classes of terms, 
gradual precision gives us a formal relationship between ideal and approximate normalization:

\begin{restatable}[Normalization Approximates the Ideal]{theorem}{norm-precision}
For any $\gradual{Gamma},\gradual{t},\gradual{U}$,
if $\gradual{Gamma}  \vdash  \gradual{t}  \Leftarrow  \gradual{U}$, then 
$\gradual{Gamma}  \vdash  \gradual{u}  \leadsfrom\  \gradual{t}  \Leftarrow  \gradual{U}$ for some $\gradual{u}$,
and $ \gradual{t}^\abe \sqsubseteq \gradual{u}  $.
\end{restatable}

\noindent Intuitively, this holds because approximate normalization for a term either matches the ideal,
or produces $ {\qm } $, which is less precise than every other term.

\subsubsection*{Preservation of Typing}
 
To prove type \replaced{safety}{soundness} for \lang, a key property of
normalization is that it preserves typing.
This property relies on the fact that hereditary substitution preserves typing,
which can be shown using a technique similar to that of \citet{churchCurry}.

\begin{restatable}[Normalization preserves typing]{theorem}{typeNorm} 
    \label{lem:typeSub}
        If $\gradual{Gamma}  \vdash  \gradual{u}  \leadsfrom\  \gradual{t}  \Leftarrow  \gradual{U}$, then $\gradual{Gamma}  \vdash  \gradual{u}  \Leftarrow  \gradual{U}$.
\end{restatable}

\subsubsection*{Normalization as a Total Function}  

Since we have defined substitution and normalization using inference rules, 
they are technically relations rather than functions.
Since the rules are syntax directed in terms of their inputs,
it is easy to show that there is at most one result for every set of inputs.
As we discussed above, the approximation in \rrule{GNCheckApprox} makes normalization
total.

\begin{restatable}[Normalization is Total]{theorem}{normTotal}
    \label{lem:normTotal}
    If $\gradual{Gamma}  \vdash  \gradual{t}  \Leftarrow  \gradual{U}$, then $\gradual{Gamma}  \vdash  \gradual{u}  \leadsfrom\  \gradual{t}  \Leftarrow  \gradual{U}$ for exactly one $\gradual{u}$.
\end{restatable}

\section{\lang: Runtime Semantics} 
\label{sec:gradual-semantics}    
  
With the type system for \lang realized, we
turn to its dynamic semantics.
Following the approaches of \citet{agt} and \citet{parametricityRevisited},
we let the syntactic type-safety proof for the static \slang drive its design.
In place of a cast calculus, gradual terms carry \textit{evidence} 
that they match their type, and computation steps evolve that evidence incrementally.
When evidence no longer supports the well-typedness of a term, 
execution fails with a runtime type error. 

\subsection{The Runtime Language}

\begin{figure}
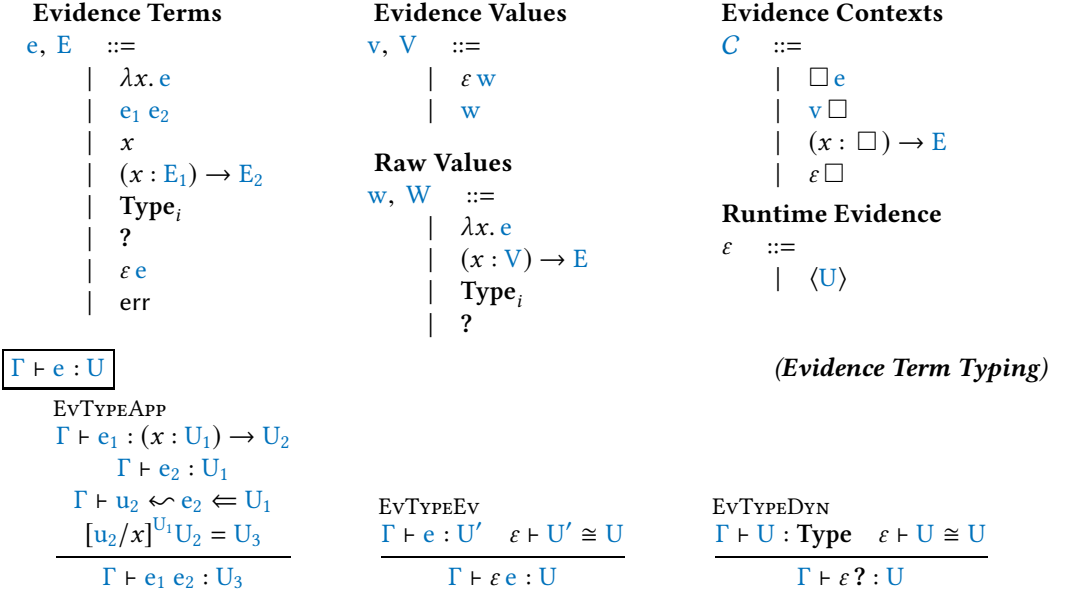
  
    \begin{minipage}[t]{0.32\textwidth} 
    \nonterms{et} 
    \end{minipage}
    \begin{minipage}[t]{0.32\textwidth}
    \nonterms{dummyev,rv}  
    \end{minipage}
    \rev{
    \begin{minipage}[t]{0.32\textwidth}
        \nonterms{evalContext}  
        \nonterms{epsilon} 
        \end{minipage}
        \drules{$\gradual{Gamma}  \vdash  \evterm{t}  \ottsym{:}  \gradual{U}$}{\ottcom{Evidence Term Typing}}{EvTypeApp, EvTypeEv,EvTypeDyn}
    }

    \caption{Evidence Term Syntax and Typing (select rules)} 
    \label{fig:evidenceTerms} 
\end{figure}
  
\Autoref{fig:evidenceTerms} gives the syntax for our runtime language.
It   mirrors the syntax for gradual terms, with two main changes.
In place of type ascriptions,
we have a special form for terms augmented with evidence,
following \citet{parametricityRevisited}.
We also have $\mathsf{err}$, an explicit term for runtime type errors.

Translation proceeds by augmenting our bidirectional typing rules to 
 output the translated term. Type ascriptions are dropped in the
\rrule{GSynthAnn} rule, and initial evidence of consistency is added in \rrule{GCheckSynth}.
\Autoref{subsec:initialEvidence} describes how to derive this initial evidence.
In the \rrule{GSynthDyn} rule, we annotate $ {\qm } $ with evidence $ {\qm } $,
so $ {\qm } $ is always accompanied by some evidence of its type.
\added{Similarly, functions of type $\qm$ are ascribed $\langle \qm \to \qm \rangle$.}

In \autoref{fig:evidenceTerms} we also define the class of syntactic values, which determines those terms that are
done evaluating.
We wish to allow values to be augmented with evidence, but not to have multiple evidence objects stacked
on a value.
To express this, we separate the class of values from the class of \textit{raw values},
which are never ascribed with evidence at the top level. 

Values are similar to, but not the same, as canonical forms.
In particular, there are no redexes in canonical terms, even beneath a $\lambda$,
whereas values may contain redexes within abstractions.

\subsection{Typing and Evidence}
\label{subsec:initialEvidence}

To establish progress and preservation, we need typing rules for evidence terms,
whose key rules we highlight in \autoref{fig:evidenceTerms}.
These are essentially the same as for gradual terms, with two major changes.
First, we no longer use bidirectional typing, since
our type system need not be syntax directed to prove \replaced{safety}{soundness}.
\added{Second, whereas gradual terms could be given any type that is consistent with their actual type,
we only allow this for dynamic terms directly ascribed with evidence, as seen in the}
\rev{ rule \rrule{EvTypeEv}.    
Thus, all applications of consistency are made explicit in the syntax of
evidence terms, and for a term $\varepsilon\ \evterm{t}$, the evidence $\varepsilon$ serves as a concrete witness
between the actual type of $\evterm{t}$ and the type at which it is used.        
The normalization relation is extended to evidence terms: it simply erases evidence ascriptions
and otherwise behaves like the normalization for gradual terms.
We can then define hereditary substitutions
of evidence terms into types, which is crucial for updating evidence after function applications}.

This raises the question: what is evidence?
\added{At a high level, the evidence attached to a term tracks the most precise type information
about this term that is dynamically available. 
As we can see in \autoref{fig:evidenceTerms}, evidence consists of a canonical type: we use brackets 
$\langle \rangle$ to syntactically distinguish evidence from canonical forms.
Ascribing a term with evidence behaves like a cast in a gradual cast calculus; the key difference is that evidence only ever increases in precision.
It serves as a witness for the consistency of two types, and by refining evidence at each step
(and failing when this is not possible), we ensure that each intermediate expression is (gradually) well-typed.} 
\rev{
\citet{agt} identify the correspondence between evidence and the middle type in a threesome calculus~\cite{Siek:2009:TWB:1570506.1570511}.} 

AGT provides a general formulation of evidence, applicable to multi-argument, asymmetric predicates. However, since equality is the only predicate we use,
the meet of two terms is sufficient to serve as evidence of their consistency.
\rev{We say that $ \varepsilon \vdash  \gradual{U}_{{\mathrm{1}}}  \cong  \gradual{U}_{{\mathrm{2}}} $ whenever $\varepsilon = \langle  \gradual{U}'  \rangle$ and $ \gradual{U}' \sqsubseteq   \gradual{U}_{{\mathrm{1}}}  \sqcap  \gradual{U}_{{\mathrm{2}}}  $.}
\added{There are two key operations on evidence. First, we need \textit{initial evidence}  for
elaborating gradual terms to evidence terms.}
\rev{
If a term synthesizes some $\gradual{U}$ and is checked against $\gradual{U}'$, then during elaboration
we can ascribe to it the evidence $ \langle \gradual{U}  \sqcap  \gradual{U}' \rangle $.  
Secondly, we need a way to combine two pieces of evidence at runtime, an operation referred to as \textit{consistent transitivity} in AGT:
if $ \langle  \gradual{U}  \rangle \vdash  \gradual{U}_{{\mathrm{1}}}  \cong  \gradual{U}_{{\mathrm{2}}} $, and $ \langle  \gradual{U}'  \rangle \vdash  \gradual{U}_{{\mathrm{2}}}  \cong  \gradual{U}_{{\mathrm{3}}} $,
then $ \langle   \gradual{U}  \sqcap  \gradual{U}'   \rangle \vdash  \gradual{U}_{{\mathrm{1}}}  \cong  \gradual{U}_{{\mathrm{3}}} $,
provided that the meet is defined.
So we can also use the precision meet to dynamically combine different pieces of evidence.}

\added{Evidence is combined using the meet operation, which is based on definitional (intensional) equality.
This means that if we have a type $\mathtt{A} : (\ottkw{Nat} \to \ottkw{Nat}) \to \TypeType_1$,
then $\mathtt{A} (\lambda x \ldotp x + x - x)$ and $\mathtt{A} (\lambda x \ldotp x)$ will not be consistent at runtime,
despite being extensionally equivalent.
Extensional equality is undecidable, so it cannot be used during typechecking.
Since definitional equality is decidable, we use it both for typechecking and at runtime. This also ensures that the type operations performed at runtime directly mirror those performed during typechecking. 
}

\subsection{Developing a Safe Semantics}

To devise our semantics, we imagine a hypothetical 
proof of progress and preservation.
Progress tells us which expressions we need reduction rules for,
and preservation tells us how to step in order to remain well-typed.

\subsubsection*{Double Evidence}

Since values do not contain terms of the form
$\myepsilon_{{\mathrm{2}}} \, \ottsym{(}  \myepsilon_{{\mathrm{1}}} \, \evterm{rv}  \ottsym{)}$, progress dictates that we need a reduction rule for such a case.
If $\cdot  \vdash  \evterm{rv}  \ottsym{:}  \gradual{U}$, $ \myepsilon_{{\mathrm{1}}} \vdash  \gradual{U}  \cong  \gradual{U}' $ and $ \myepsilon_{{\mathrm{2}}} \vdash  \gradual{U}'  \cong  \gradual{U}'' $,
then $ \myepsilon_{{\mathrm{1}}}  \sqcap  \myepsilon_{{\mathrm{2}}}  \vdash  \gradual{U}  \cong  \gradual{U}'' $, so we step to $\ottsym{(}   \myepsilon_{{\mathrm{1}}}  \sqcap  \myepsilon_{{\mathrm{2}}}   \ottsym{)} \evterm{rv}$.
If the meet is not defined, then a runtime error occurs.

\subsubsection*{Functions with Evidence}

\added{There are two complications for reducing applications with evidence.
The first is that in} 
\rev{$\lambda  \mathit{x}  \ldotp  \evterm{t}$, the variable $\mathit{x}$
may be free in evidence ascriptions within  $\evterm{t}$.
When performing a substitution, we need the type and normal form of the term replacing the variable.
We use the notation
${[  \mathit{x}  \Mapsto  \evterm{t}_{{\mathrm{1}}}  ]}^{ \gradual{u}  :  \gradual{U} }  \evterm{t}_{{\mathrm{2}}}$
to denote the syntactic replacement of $\mathit{x} $ by $\evterm{t}_{{\mathrm{1}}}$ in $\evterm{t}_{{\mathrm{2}}}$,
where free occurrences of $\mathit{x} $ in evidence within  $\evterm{t}_{{\mathrm{2}}}$ are replaced by $\gradual{u}$
(the normal form of $\evterm{t}_{{\mathrm{2}}}$)
using hereditary substitution at type $\gradual{U}$.
We use this operation to reduce applications.

A second issue is that,
while the simple rules dictate how to evaluate a $\lambda$-term applied to a value,
they do not determine how to proceed for applications of the form $ \ottsym{(} \langle \gradual{U} \rangle \, \lambda  \mathit{x}  \ldotp  \evterm{t}  \ottsym{)} \  \ottsym{(}  \langle \gradual{U}' \rangle \, \evterm{rv}  \ottsym{)} $. 
 In such a case, we know that $\cdot  \vdash  \ottsym{(}  \langle \gradual{U} \rangle \, \lambda  \mathit{x}  \ldotp  \evterm{t}   \ottsym{)}  \ottsym{:}  \gradual{U}_1$ and that $ \langle \gradual{U} \rangle\vdash  \gradual{U}_1  \cong  \gradual{U}_2 $ for some $\gradual{U}_2$.
Computing $  (\ottkw{dom}\  \langle \gradual{U} \rangle )  \sqcap  \langle \gradual{U}' \rangle $ yields evidence that the type of $\evterm{rv}$ is consistent with the domain of $\gradual{U}_1$,
so we ascribe  this evidence during substitution to preserve well-typedness.
The evidence-typing rules say that the type of an application is found by normalizing the argument value
and substituting into the codomain of the function type. To produce a result at this type,
we can normalize $\evterm{rv}$ and substitute it into the codomain of $\langle \gradual{U} \rangle$,
thereby producing evidence that the actual result is consistent with the return type.
 In the case where $\evterm{rv}$ is not ascribed with evidence, we can simply behave as if it were ascribed
$\langle  {\qm }  \rangle$ and proceed using the above process. } 

\subsubsection*{Applying The Unknown Term}
 
The syntax for values only admits application under binders,
so we must somehow reduce terms of the form $ \ottsym{(}  \myepsilon \, {\qm }  \ottsym{)} \  \evterm{v} $.
The solution is simple: if the function is unknown, so is its output.
Since the unknown term is always accompanied by evidence at runtime,
we calculate the result type by substituting the argument into the codomain
of the evidence associated with $ {\qm } $.

\subsubsection*{The Full Semantics}

\begin{figure}

     \ottdefnStep{}

    \caption{\lang: \rev{Dynamic Semantics}}
    \label{fig:gradual-semantics}
\end{figure}

All other well-typed terms are either values, or contain a redex as a subterm,
either of the simple variety or of the varieties described above.
Using contextual rules to account for these remaining cases, we have a
semantics that satisfies progress and preservation \textit{by construction}.
\Autoref{fig:gradual-semantics} gives the full set of rules.  

\subsection{Example: Running $\mathtt{head}$ of $\mathtt{nil}$}
\label{subsec:nilhead-run}

We return to the example from \autoref{subsec:type-headnil}, this time explaining its runtime behaviour.
Because of consistency, the term $\ottsym{(}     \ottkw{Nil}  \   \ottkw{Nat}   \dblcolon \ottsym{(}     \ottkw{Vec}  \   \ottkw{Nat}   \  {\qm }   \ottsym{)}   \ottsym{)}$
is given the evidence $\langle   \ottkw{Vec}  \   \ottkw{Nat}   \  {\qm } \rangle$, obtained by computing ${\langle   \ottkw{Vec}  \   \ottkw{Nat}   \   0  \sqcap    \ottkw{Vec}  \   \ottkw{Nat}   \  {\qm } \rangle}$.
Applying consistency to use this as an argument adds the evidence $\langle   \ottkw{Vec}  \   \ottkw{Nat}   \   1  \rangle$,
since we check $\ottsym{(}     \ottkw{Nil}  \   \ottkw{Nat}   \dblcolon \ottsym{(}     \ottkw{Vec}  \   \ottkw{Nat}   \  {\qm }   \ottsym{)}   \ottsym{)}$ against $\ottkw{dom}\ (    \ottkw{Vec}  \   \ottkw{Nat}   \   1   \to  \ottkw{Nat}  )$.
The rule \rrule{StepContext} dictates that we must evaluate the argument to a function before
evaluating the application itself. 
Our argument is $\langle   \ottkw{Vec}  \   \ottkw{Nat}   \   1  \rangle\langle   \ottkw{Vec}  \   \ottkw{Nat}   \   0  \rangle(  \ottkw{Nil}  \   \ottkw{Nat}  )$, and since the meet of the evidence
types is undefined, we step to $\mathsf{err}$ with \rrule{StepAscrFail}.

\section{Properties of \lang}  
\label{sec:lang-properties} 
 
\lang satisfies all the criteria for gradual languages set forth by \citet{refinedCriteria}.

\paragraph{Safety} First, \lang is type \replaced{safe}{sound}  
\textit{by construction}: the runtime semantics
are specifically crafted to maintain progress and preservation.
We can then obtain the standard safety result for gradual languages, namely that well-typed terms do not get stuck.
 
\begin{restatable}[Type safety]{theorem}{typeSound}
    \label{thm:typeSound}
        If $\cdot  \vdash  \evterm{t}  \ottsym{:}  \gradual{U}$, then either $\evterm{t} \longrightarrow^{*} \evterm{v}$ for some $\evterm{v}$,
        $\evterm{t} \longrightarrow^{*} \mathsf{err}$,
        or $\evterm{t}$ diverges.
\end{restatable}

\rev{
This means that gradually well-typed programs in \lang may fail with runtime type errors, but they will never get stuck. Among the three main approaches to deal with gradual types in the literature, \lang follows the original approach of \citet{gradualTypeInitial} and \citet{refinedCriteria}, which enforces types eagerly at boundaries, including at higher-order types. This is in contrast with first-order enforcement (a.k.a as transient semantics \citep{Vitousek:2017:BTL:3093333.3009849}), or simple type erasure (a.k.a as optional typing).\footnote{\citet{Greenman:2018:STS:3243631.3236766} present a detailed comparative semantic account of these three approaches.} In particular, while the transient semantics support open world soundness \citep{Vitousek:2017:BTL:3093333.3009849} when implemented on top of a (safe) dynamic language, it is unclear if and how this approach, which is restricted to checking type constructors, can scale to full-spectrum dependent types. \lang is a sound gradually-typed language that requires elaboration of the complete program in order to insert the pieces of evidence that support runtime checking.
}

\paragraph{Conservative Extension of\ \slang}    
It is easy to show that \lang is a conservative extension of \slang. This means that any fully-precise \lang
programs enjoy the soundness and \added{logical} consistency properties that \slang guarantees.
Any statically-typed term is well-typed in \lang by construction, thanks to AGT: on fully precise gradual types, $\alpha \circ \gamma$ is the identity.
Moreover, the \textit{only} additions are those pertaining to $ {\qm } $,
meaning that if we restrict ourselves to the static subset of terms (and types) without $ {\qm } $,
then we have all the properties of the static system. We formalize this as follows: 

\begin{restatable}{theorem}{staticEmbed}
    \label{thm:staticEmbed}
    \added{}
    \rev{
    If $\gradual{Gamma},\gradual{t},\gradual{U}$ are the embeddings of some $\static{Gamma}, \static{t},\static{U}$ into \lang,
    and $\gradual{Gamma}  \vdash  \gradual{t}  \Leftarrow  \gradual{U}$, then $\static{Gamma}  \vdash  \static{t}  \Leftarrow  \static{U}$.
    Moreover, if $\static{t} \longrightarrow^{*} \static{v} $ and  $\gradual{t}$ elaborates to $\evterm{t}$, then there exists some $\evterm{v}$ where $\evterm{t} \longrightarrow^{*} \evterm{v}$, where
    removing evidence from $\evterm{v}$ yields $\static{v}$ }.
\end{restatable}

\paragraph{Embedding of Untyped Lambda Calculus}
A significant property of \lang is that it can fully embed the untyped lambda calculus, including non-terminating terms.
Given an untyped embedding function $\lceil t \rceil$ 
that (in essence) annotates all terms with $ {\qm } $ we can show that any untyped term can be embedded in our system.
Since no type information is present, all evidence objects are formed using $ {\qm } $ or $\to$, and the meet operator
never fails and untyped programs behave normally in \lang.

\begin{restatable}{theorem}{untypedEmbed}
    \label{thm:untypedEmbed} 
    For any untyped $\lambda$-term $t$ and closing environment $\ \gradual{Gamma}$ that maps all variables to type $ {\qm } $, then $\gradual{Gamma}\vdash \lceil t \rceil  \Rightarrow   {\qm } $.
    \rev{Moreover, if $t$ is closed, then $t \longrightarrow^{*} v$ implies that $\lceil t \rceil$ elaborates to $\evterm{t}$ where  $\evterm{t} \longrightarrow^{*} \evterm{v} $ and stripping evidence from  $\evterm{v}$ yields $v$}.
            \end{restatable}

\paragraph{Gradual Guarantees}
\label{subsec:gg-properties}

 \lang smoothly supports the full spectrum between dependent and untyped programming---a property known as the gradual guarantee~\cite{refinedCriteria}, which comes in two parts. \added{We say that} \rev{$\gradual{Gamma} \sqsubseteq \gradual{Gamma}'$ if they contain the same variables, and for each $(x : \gradual{U}) \in \gradual{Gamma}$,
 $(x : \gradual{U}') \in \gradual{Gamma}'$  where $\gradual{U} \sqsubseteq \gradual{U}'$}.

\begin{restatable}[Gradual Guarantee]{theorem}{gradGuarantee}
    \label{thm:gradGuarantee}
    \quad  

    \textsc{(Static Guarantee)} Suppose $\gradual{Gamma}  \vdash  \gradual{t}  \Leftarrow  \gradual{U}$ and $ \gradual{U} \sqsubseteq  \gradual{U}' $.\ If $\gradual{Gamma} \sqsubseteq \gradual{Gamma}'$ and $\gradual{t} \sqsubseteq \gradual{t}'$, then $\gradual{Gamma}'  \vdash  \gradual{t}'  \Leftarrow  \gradual{U}'$.
    
    \textsc{(Dynamic Guarantee)} Suppose that $\cdot  \vdash  \evterm{t}_{{\mathrm{1}}}  \ottsym{:}  \gradual{U}$, $\cdot  \vdash  \evterm{t}'_{{\mathrm{1}}}  \ottsym{:}  \gradual{U}'$,  $\evterm{t}_{{\mathrm{1}}}\sqsubseteq\evterm{t}'_{{\mathrm{1}}}$, and $\gradual{U}\sqsubseteq\gradual{U}'$. If $\evterm{t}_{{\mathrm{1}}}\longrightarrow^{*}\evterm{t}_{{\mathrm{2}}}$, then $\evterm{t}'_{{\mathrm{1}}}\longrightarrow^{*}\evterm{t}'_{{\mathrm{2}}}$ where $\evterm{t}_{{\mathrm{2}}}\sqsubseteq\evterm{t}'_{{\mathrm{2}}}$.
                        \end{restatable}

AGT ensures that the gradual guarantee holds by construction.
Specifically, because approximate normalization and consistent transitivity
are monotone with respect to precision,
we can establish a weak bisimulation between the steps of the more and less precise versions~\citep{agt}.

A novel insight that arises from our work is that we need a restricted form of the dynamic gradual guarantee {\em for normalization} in order
to prove the static gradual guarantee.
To differentiate it from the standard one, we call it the {\em normalization gradual guarantee}.
Because an $\eta$-long term might be longer at a more precise type, we phrase the guarantee modulo $\eta$-equivalence:
\added{we say that} \rev{ $\gradual{U}_1 \sqsubseteq^\eta \gradual{U}_2$ if $\gradual{U}_1 =_\eta \gradual{U}'_1$,
$\gradual{U}_2 =_\eta \gradual{U}'_2$ and $\gradual{U}'_1 \sqsubseteq \gradual{U}'_2$}.

With these defined, we can state the normalization gradual guarantee:

\begin{restatable}[Normalization Gradual Guarantee]{lemma}{normGuarantee}
    \label{thm:normGuarantee} 
    Suppose $\gradual{Gamma}_{{\mathrm{1}}}  \vdash  \gradual{u}_{{\mathrm{1}}}  \leadsfrom\  \gradual{t}_{{\mathrm{1}}}  \Leftarrow  \gradual{U}_{{\mathrm{1}}}$. 
    If $\gradual{Gamma}_{{\mathrm{1}}} \sqsubseteq^\eta \gradual{Gamma}_{{\mathrm{2}}}$, $ \gradual{t}_{{\mathrm{1}}} \sqsubseteq  \gradual{t}_{{\mathrm{2}}} $,
        and $ \gradual{U}_{{\mathrm{1}}} \sqsubseteq^\eta  \gradual{U}_{{\mathrm{2}}} $, then $\gradual{Gamma}_{{\mathrm{2}}}  \vdash  \gradual{u}_{{\mathrm{2}}}  \leadsfrom\  \gradual{t}_{{\mathrm{2}}}  \Leftarrow  \gradual{U}_{{\mathrm{2}}}$ where $ \gradual{u}_{{\mathrm{1}}} \sqsubseteq^\eta  \gradual{u}_{{\mathrm{2}}} $.
\end{restatable}

\section{Extension: Inductive Types}
\label{sec:inductive}

Though \lang provides type safety and the gradual guarantees, its lack of inductive types means that
programming is cumbersome. \added{Church encodings allow for some induction, but are strictly less powerful
than proper inductive types.}
Additionally, 
induction principles, along with basic facts like $0 \neq 1$, cannot be proven in the purely negative Calculus of Constructions~\citep{stump_2017}.
However, we can type such a term if we introduce inductive types with eliminators, and allow types to be defined
in terms of such eliminations.

This section describes how to extend \lang with 
a few common inductive types---natural numbers, vectors, and an identity type for equality proofs---
along with their eliminators. While not as useful as user-defined types or pattern matching (both of which are important subjects for future work),
this specific development illustrates how our approach can be extended to a more full-fledged dependently-typed language.
Note that while we show how inductives can be added to the language, extending our metatheory
to include inductives is left as future work.

\paragraph{Syntax and Typing}
 
\rev{We augment the syntax for terms as follows:}
\begin{align*}
\gradual{t},\gradual{T} & ::= & \ldots & |  \ottkw{Nat}  |  0  |   \ottkw{Succ}\ \gradual{t}  |    \ottkw{Vec}  \  \gradual{T}  \  \gradual{t}  |     \ottkw{Eq}  \  \gradual{T}  \  \gradual{t}_{{\mathrm{1}}}  \  \gradual{t}_{{\mathrm{2}}}  |    \ottkw{Refl}  \  \gradual{T}  \  \gradual{t} 
|   \ottkw{Nil}  \  \gradual{t}  |      \ottkw{Cons}  \  \gradual{T}  \  \gradual{t}_{{\mathrm{1}}}  \  \gradual{t}_{{\mathrm{2}}}  \  \gradual{t}_{{\mathrm{3}}}  \\ &&&
|  \ottkw{natElim}\    \gradual{T} \  \gradual{t}_{{\mathrm{1}}}  \  \gradual{t}_{{\mathrm{2}}}  \  \gradual{t}_{{\mathrm{3}}}  
| \ottkw{vecElim}\      \gradual{T}_{{\mathrm{1}}} \  \gradual{t}_{{\mathrm{1}}}  \  \gradual{T}_{{\mathrm{2}}}  \  \gradual{t}_{{\mathrm{2}}}  \  \gradual{t}_{{\mathrm{3}}}  \  \gradual{t}_{{\mathrm{4}}}     
| \ottkw{eqElim}\      \gradual{T}_{{\mathrm{1}}} \  \gradual{T}_{{\mathrm{2}}}  \  \gradual{t}_{{\mathrm{1}}}  \  \gradual{t}_{{\mathrm{2}}}  \  \gradual{t}_{{\mathrm{3}}}  \  \gradual{t}_{{\mathrm{4}}} 
\end{align*} 

\begin{figure}
    \begin{minipage}{0.45\textwidth}
    \begin{align*} 
         \ottkw{Nat}  & : &&  \TypeType_{ \ottsym{1} } \\
         \ottkw{Vec}  & : &&   \TypeType_{ \ottnt{i} }  \to   \ottkw{Nat}  \to  \TypeType_{ \ottnt{i} }   \\
         \ottkw{Eq}  & : && \ottsym{(}  \mathtt{A}  \ottsym{:}   \TypeType_{ \ottnt{i} }   \ottsym{)}  \rightarrow    \mathtt{A}  \to   \mathtt{A}  \to  \TypeType_{ \ottnt{i} }   
        \end{align*}
    \end{minipage}
    \begin{minipage}{0.45\textwidth}
        \begin{align*} 
             0  & : &&  \ottkw{Nat} \\
        \ottkw{Succ} & : &&   \ottkw{Nat}  \to  \ottkw{Nat}  \\
         \ottkw{Nil}  & : &&   \ottsym{(}  \mathtt{A}  \ottsym{:}   \TypeType_{ \ottnt{i} }   \ottsym{)}  \rightarrow   \ottkw{Vec}  \  \mathtt{A}  \   0  \\
         \ottkw{Refl}  & : &&    \ottsym{(}  \mathtt{A}  \ottsym{:}   \TypeType_{ \ottnt{i} }   \ottsym{)}  \rightarrow  \ottsym{(}  \mathit{x}  \ottsym{:}  \mathtt{A}  \ottsym{)}  \rightarrow   \ottkw{Eq}  \  \mathtt{A}  \  \mathit{x}  \  \mathit{x} 
        \end{align*}
    \end{minipage}
    \begin{align*}
         \ottkw{Cons}  & : &&   \ottsym{(}  \mathtt{A}  \ottsym{:}   \TypeType_{ \ottnt{i} }   \ottsym{)}  \rightarrow  \ottsym{(}  \mathtt{n}  \ottsym{:}   \ottkw{Nat}   \ottsym{)}  \rightarrow  \ottsym{(}  \mathtt{hd}  \ottsym{:}  \mathtt{A}  \ottsym{)}  \rightarrow  \ottsym{(}  \mathtt{tl}  \ottsym{:}     \ottkw{Vec}  \  \mathtt{A}  \  \mathtt{n}   \ottsym{)}  \rightarrow   \ottkw{Vec}  \  \mathtt{A}  \  \ottsym{(} \ottkw{Succ}\   \mathtt{n}     \ottsym{)}  \\
        \ottkw{natElim} & : && \ottsym{(}  \mathtt{m}  \ottsym{:}  \ottsym{(}    \ottkw{Nat}  \to  \TypeType_{ \ottnt{i} }    \ottsym{)}  \ottsym{)}  \rightarrow   \ottsym{(}   \mathtt{m} \   0    \ottsym{)} \to \ottsym{(}   \ottsym{(}  \mathtt{n}  \ottsym{:}   \ottkw{Nat}   \ottsym{)}  \rightarrow    \mathtt{m} \  \mathtt{n}  \to \mathtt{m}  \  \ottsym{(} \ottkw{Succ}\   \mathtt{n}     \ottsym{)}   \ottsym{)}  
         \to  \ottsym{(}  \mathtt{n}  \ottsym{:}   \ottkw{Nat}   \ottsym{)}  \rightarrow  \mathtt{m} \  \mathtt{n}  \\
        \ottkw{vecElim} & : &&  \ottsym{(}  \mathtt{A}  \ottsym{:}   \TypeType_{ \ottnt{i} }   \ottsym{)}  \rightarrow  \ottsym{(}  \mathtt{n}  \ottsym{:}   \ottkw{Nat}   \ottsym{)}  \rightarrow  \ottsym{(}  \mathtt{m}  \ottsym{:}  \ottsym{(}      \ottkw{Vec}  \  \mathtt{A}  \  \mathtt{n}  \to  \TypeType_{ \ottnt{i} }    \ottsym{)}  \ottsym{)}  \rightarrow  \mathtt{m} \  \ottsym{(}     \ottkw{Vec}  \  \mathtt{A}  \   0    \ottsym{)} 
        \\ &&& \to \ottsym{(}   \ottsym{(}  \mathtt{n_2}  \ottsym{:}   \ottkw{Nat}   \ottsym{)}  \rightarrow  \ottsym{(}  \mathtt{h}  \ottsym{:}  \mathtt{A}  \ottsym{)}  \rightarrow  \ottsym{(}  \mathtt{tl}  \ottsym{:}     \ottkw{Vec}  \  \mathtt{A}  \  \mathtt{n_2}   \ottsym{)}  \rightarrow    \mathtt{m} \  \mathtt{vec}  \to \mathtt{m}  \  \ottsym{(}       \ottkw{Cons}  \  \mathtt{A}  \  \ottsym{(} \ottkw{Succ}\   \mathtt{n_2}     \ottsym{)}  \  \mathtt{hd}  \  \mathtt{tl}   \ottsym{)}   \ottsym{)}
        \\ &&& \to  \ottsym{(}  \mathtt{vec}  \ottsym{:}     \ottkw{Vec}  \  \mathtt{A}  \  \mathtt{n}   \ottsym{)}  \rightarrow  \mathtt{m} \  \mathtt{vec} \\
        \ottkw{eqElim} & : && \ottsym{(}  \mathtt{A}  \ottsym{:}   \TypeType_{ \ottnt{i} }   \ottsym{)}  \rightarrow  \ottsym{(}  \mathtt{m}  \ottsym{:}  \ottsym{(}  \mathit{x}  \ottsym{:}  \mathtt{A}  \ottsym{)}  \rightarrow  \ottsym{(}  \mathit{y}  \ottsym{:}  \mathtt{A}  \ottsym{)}  \rightarrow       \ottkw{Eq}  \  \mathtt{A}  \  \mathit{x}  \  \mathit{y}  \to  \TypeType_{ \ottnt{i} }    \ottsym{)}  \\
        &&&
        \rightarrow  \ottsym{(}     \ottsym{(}  \mathit{z}  \ottsym{:}  \mathtt{A}  \ottsym{)}  \rightarrow  \mathtt{m} \  \mathit{z}  \  \mathit{z}  \  \ottsym{(}     \ottkw{Refl}  \  \mathtt{A}  \  \mathit{z}   \ottsym{)}   \ottsym{)} 
        \to \ottsym{(}  \mathit{x}  \ottsym{:}  \mathtt{A}  \ottsym{)}  \rightarrow  \ottsym{(}  \mathit{y}  \ottsym{:}  \mathtt{A}  \ottsym{)}  \rightarrow  \ottsym{(}  \mathtt{p}  \ottsym{:}      \ottkw{Eq}  \  \mathtt{A}  \  \mathit{x}  \  \mathit{y}   \ottsym{)}  \rightarrow  \ottsym{(}     \mathtt{m} \  \mathit{x}  \  \mathit{y}  \  \mathtt{p}   \ottsym{)}
        \end{align*}
        \caption{Constructor and Eliminator Types}
        \label{fig:elimTypes}
    \end{figure}

The typing rules are generally straightforward. We omit the full rules, but we essentially type them
as functions that must be fully applied, with the types given in \autoref{fig:elimTypes}.
Each form checks its arguments against the specified types, and the rule \rrule{GCheckSynth} ensures that
typechecking succeeds so long as argument types are consistent with the expected types.
Adding these constructs to canonical forms is interesting.
Specifically, the introduction forms are added as atomic forms, and the 
eliminators become new variants of the canonical spines. Since $\ottkw{natElim}$ applied to a $ \ottkw{Nat} $
is a redex, canonical forms can eliminate variables. 
\added{Eliminators take one fewer argument than in the term version, since the last argument is always the rest of the spine in which the eliminator occurs.} 
\begin{align*}
\gradual{rr},\gradual{RR} & ::= & \ldots |  \ottkw{Nat}  |  0  | \ottkw{Succ}\ \gradual{u} | \text{etc.\ldots}\\
\gradual{e} & ::= & \ldots | \ottkw{natElim}\ \gradual{u}_{{\mathrm{1}}}\ \gradual{u}_{{\mathrm{2}}}\ \gradual{u}_{{\mathrm{3}}} | \ottkw{vecElim}\ \gradual{U}_{{\mathrm{1}}}\ \gradual{u}_{{\mathrm{1}}}\ \gradual{U}_{{\mathrm{2}}}\ \gradual{u}_{{\mathrm{2}}}\ \gradual{u}_{{\mathrm{3}}}\ | \text{etc.\ldots}
\end{align*}

\paragraph{Normalization}

We extend hereditary substitution to inductive types.
\added{Unfortunately, we must treat hereditary substitution as a \textit{relation} between normal forms.
The strictly-decreasing metric we previously used no longer holds for inductive types,
so we have not proved that hereditary substitution with inductives is a well-defined function; this is left as future work.}  
For introduction forms, we simply substitute in the subterms.
For eliminators, if we are ever replacing $\mathit{x}$ with $\gradual{u}'$ in $ \mathit{x} \gradual{e} \ \ottkw{natElim}\ \gradual{u}_{{\mathrm{1}}}\ \gradual{u}_{{\mathrm{2}}}\ \gradual{u}_{{\mathrm{3}}}$,
then we substitute in $ \mathit{x} \gradual{e} $ and see if we get $ 0 $ or $\ottkw{Succ}$ back.
If we get $ 0 $, we produce $\gradual{u}_{{\mathrm{2}}}$, and if we get $ \ottkw{Succ}\ \mathtt{n}$, we compute the recursive elimination
for $\mathtt{n}$ as $\gradual{u}'_{{\mathrm{2}}}$, and substitute $\mathtt{n}$ and $\gradual{u}'_{{\mathrm{2}}}$ for the arguments of $\gradual{u}_{{\mathrm{3}}}$. 
Vectors are handled similarly. An application of $\ottkw{eqElim}$ to $   \ottkw{Refl}  \  \mathtt{A}  \  \mathit{x} $ simply returns the
given value of type $   \mathtt{m} \  \mathit{x}  \  \mathit{x}  \  \ottsym{(}     \ottkw{Refl}  \  \mathit{x}  \  \mathtt{A}   \ottsym{)} $ as a value of type $   \mathtt{m} \  \mathit{x}  \  \mathit{y}  \  \mathtt{p} $.

How should we treat eliminations with $ {\qm } $ as a head?
Since $ {\qm } $ represents the set of all static values, the result of eliminating it is the 
abstraction of the eliminations of all possible values. Since these values may produce conflicting results,
the abstraction is simply $ {\qm } $, which is our result.
However, for equality, we have a special case. Each instance of $\ottkw{eqElim}$
can have only one possible result: the given value, considered as having the output type.
Then, we abstract a singleton set.
This means we can treat each application of $\ottkw{eqElim}$ to $ {\qm } $ as an application to $   \ottkw{Refl}  \  {\qm }  \  {\qm } $. 
This principle holds for any single-constructor inductive type. 

With only functions, we needed to return $ {\qm } $ any time we applied a dynamically-typed function.
However, with eliminators, we are always structurally decreasing on the value being eliminated.
For $\ottkw{natElim}$, we can eliminate  a $ {\qm } $-typed value provided it is $ 0 $ or $\ottkw{Succ}\ \gradual{u}$, 
but otherwise we must produce $ {\qm } $ for substitution to be total and type-preserving.
Other types are handled similarly.

\paragraph{Runtime Semantics}

The semantics are straightforward. 
Eliminations are handled as with hereditary substitution: eliminating $ {\qm } $ produces $ {\qm } $,
except with $ \ottkw{Eq} $, where $ {\qm } $
 behaves like $   \ottkw{Refl}  \  {\qm }  \  {\qm } $ when eliminating.
When applying eliminators or constructors, evidence is composed as with functions.

One advantage of \lang is that the meet operator on evidence \added{allows definitional equality checks to be moved to runtime}.
Thus, if we write $   \ottkw{Refl}  \  {\qm }  \  {\qm } $, but use it at type $    \ottkw{Eq}  \  \mathtt{A}  \  \mathit{x}  \  \mathit{y} $, then it is ascribed with evidence $\langle     \ottkw{Eq}  \  \mathtt{A}  \  \mathit{x}  \  \mathit{y}  \rangle$.
If we ever use this proof to transform a value of type  $ \mathtt{P} \  \mathit{x} $ into one of type $ \mathtt{P} \  \mathit{y} $, the meet operation on the evidence
ensures that the result actually has type $ \mathtt{P} \  \mathit{y} $.

Returning to the $\mathtt{head'}$ function from \autoref{subsec:gdtl-action}, in $     \mathtt{head'} \   \ottkw{Nat}   \   0   \  {\qm }  \  {\qm }  \  \mathtt{staticNil} $, the second $ {\qm } $
is ascribed with the evidence $\langle    \ottkw{Eq}  \   \ottkw{Nat}   \   0   \  \ottsym{(}   \ottkw{Succ}\ {\qm }    \ottsym{)} \rangle$.
The call to $\mathtt{rewrite}$ using this proof tries to convert a $ \ottkw{Vec} $ of length $ 0 $ into one of length $ 1 $,
which adds the evidence $\langle    \ottkw{Eq}  \   \ottkw{Nat}   \   0   \   1  \rangle$ to our proof term. Evaluation tries to compose the layered evidence,
but fails with the rule \rrule{StepAscrFail}, since they cannot be composed.

\section{Related Work}
\label{sec:related}

\paragraph{\slang} The static dependently-typed language \slang, from which \lang is derived, incorporates many features and techniques from the literature. The core of the language is very similar to that of \CCw~\citep{COQUAND198895}, albeit without an impredicative $\mathsf{Prop}$ sort. The core language of Idris~\citep{idrisPaper}, $\mathsf{TT}$, also features cumulative universes
with a single syntactic category for terms and types.  
Our use of canonical forms 
draws heavily from work on the  Logical Framework (LF)~\citep{Harper:1993:FDL:138027.138060, mechanizingLF}. 
The bidirectional type system we adopt is inspired by the tutorial of \citet{basicDepTutorial}. 
Our formulation of hereditary substitution~\citep{hereditary,churchCurry} in \slang is largely drawn from
that of~\citet{mechanizingLF}, particularly the type-outputting judgment for substitution on atomic forms, and the treatment of the variable type as an extrinsic argument.

\paragraph{Mixing Dependent Types and Non-termination}
Dependently-typed languages that admit non-termination either give up on logical consistency altogether ($\Omega$mega~\citep{DBLP:conf/cefp/SheardL07}, Haskell), or isolate a sublanguage of pure terminating expressions. This separation can be either enforced through the type system and/or a termination checker (Aura~\cite{DBLP:conf/icfp/JiaVMZZSZ08}, F$\star$~\cite{mumon}, Idris~\cite{idrisPaper}), or through a strict syntactic separation (Dependent ML~\cite{Xi1999}, ATS~\cite{Chen:2005:CPT:1086365.1086375}). The design space is very wide, reflecting a variety of sensible tradeoffs between expressiveness, guarantees, and flexibility.

The Zombie language~\citep{EPTCS76.9, Casinghino:2014:CPP:2535838.2535883} implements a flexible combination of programming and proving. 
The language is defined as a the programmatic fragment that ensures type safety but not termination, and a logical fragment 
(a syntactic subset of the programmatic one) that
guarantees logical consistency. Programmers must declare in which fragment a given definition lives, but mobile types and cross-fragment case expressions allow interactions between the fragments. 
Zombie embodies a different tradeoff from \lang: while the logical fragment is consistent as a logic, typechecking may diverge due to normalization of terms from the programmatic fragment. In contrast, \lang eschews logical consistency as soon as imprecision is introduced (with~$ {\qm } $), but approximate normalization ensures that typechecking terminates. 

In general, gradual dependent types as provided in \lang can be interpreted as a smooth, tight integration of such a two-fragment approach. Indeed, the subset of \lang that is fully precise corresponds to \slang, which is consistent as a logic. However, in the gradual setting, the fragment separation is fluid: it is driven by the precision of types and terms, which is free to evolve at an arbitrarily fine level of granularity. Also, the mentioned approaches are typically not concerned with accommodating the flexibility of a fully dynamically-typed language.

\paragraph{Mixing Dependent Types and Simple Types}

Several approaches have explored how soundly combine dependent types with non-dependently typed components. 
\citet{10.1007/1-4020-8141-3_34} support a two-fragment language, where runtime checks at the boundary ensure that dependently-typed properties hold. 
The approach is limited to properties that are expressible as boolean-valued functions. 
\citet{Tanter:2015:GCP:2816707.2816710} develop a cast framework for subset types in Coq, allowing one to assert a term of type $A$ to have the subset type $\{ a : A \mid P\;a \}$ for any decidable (or checkable) property $P$. They use an axiom to represent cast errors. 
\citet{Osera:2012:DI:2103776.2103779} present \textit{dependent interoperability} as a multi-language approach to combine both typing disciplines, mediated by runtime checks. Building on the subset cast framework (reformulated in a monadic setting instead of using axiomatic errors), \citet{dagandAl:icfp2016,dagand_tabareau_tanter_2018} revisit dependent interoperability in Coq via type-theoretic Galois connections that allow automatic lifting of higher-order programs. 
Dependent interoperability allows exploiting the connection between pre-existing types, such as $ \ottkw{Vec} $ and $ \ottkw{List} $, imposing the overhead of data structure conversions. The fact that $ \ottkw{List} $ is a less precise structure than $ \ottkw{Vec} $ is therefore defined {\em a posteriori}. In contrast, in \lang, one can simply define $  \ottkw{List}  \  \mathtt{A} $ as an alias for $   \ottkw{Vec}  \  \mathtt{A}  \  {\qm } $, thereby avoiding the need for deep structural conversions.

The work of \citet{lehmannTanter:popl2017} on gradual refinement types 
includes some form of dependency in types. 
Gradual refinement types range from simple types to logically-refined types, i.e. subset types where the refinement is drawn from an SMT-decidable logic. Imprecise logical formulae in a function type can refer to arguments and variables in context. This kind of value dependency is less expressive than the dependent type system considered here. Furthermore, \lang is the first gradual language to allow $ {\qm } $ to be used in both term and type position, and to fully embed the untyped lambda calculus.

\paragraph{Programming with Holes}
Finally, we observe that using $ {\qm } $ in place of proof terms in \lang is related to the concept of \textit{holes} in dependently-typed languages.
Idris~\citep{idrisPaper} and Agda~\citep{agdaPaper} both allow typechecking of programs with typed holes. The main difference between $ {\qm } $ and holes in these languages is that applying a hole to a value results in a stuck term,
while in \lang, applying $ {\qm } $ to a value produces another $ {\qm } $. 

Recently, \citet{DBLP:journals/pacmpl/OmarVCH19} describe Hazelnut, a language and programming system with \textit{typed holes} that fully supports evaluation in presence of holes, including reduction {\em around} holes. The approach is based on Contextual Modal Type Theory~\citep{Nanevski:2008:CMT:1352582.1352591}. It would be interesting to study whether the dependently-typed version of CMTT~\citep{Pientka:2008:PPE:1389449.1389469} could be combined with the evaluation approach of Hazelnut, and the IDE support, in order to provide a rich programming experience with gradual dependent types.

\section{Conclusion}
\label{sec:conclusion}

\lang  represents a glimpse of the challenging and potentially large design space induced by combining dependent types and gradual typing. Specifically, this work proposes approximate normalization as a novel technique for designing gradual dependently-typed languages, in a way that ensures decidable typechecking and naturally satisfies the gradual guarantees.

Currently, \lang lacks a number of features required of a practical dependently-typed programming language.  
While we have addressed the most pressing issue of supporting inductive types in Section~\ref{sec:inductive}, the metatheory of this extension, in particular the proof of strong normalization, is future work. 
It might also be interesting to consider pattern matching as the primitive notion for eliminating inductives, as in Agda, instead of elimination principles as in Coq; the equalities implied by dependent matches could be turned into runtime checks for gradually-typed values.

Future work includes supporting implicit arguments and higher-order unification, blame tracking~\citep{10.1007/978-3-642-00590-9_1}, and efficient runtime semantics with erasure of computationally-irrelevant arguments~\cite{DBLP:conf/types/BradyMM03}.
Approximate normalization might be made more precise by exploiting {\em termination contracts}~\citep{nguyen2018sizechange}.

\paragraph{{\bf Acknowledgments}}                         We thank the anonymous reviewers for their constructive feedback.

    \bibliography{myRefs}

  \pagebreak

  \appendix

\section{Complete Definitions} 
 
\subsection{\slang}

\begin{minipage}[t]{0.45\textwidth}
    \nonterms{t,simpleContext,u}
\end{minipage}
\begin{minipage}[t]{0.45\textwidth}
    \nonterms{rr,e,SGamma}
\end{minipage}

\ottdefnSSynth

\ottdefnSCheck

\ottdefnSNormSynth 

\ottdefnSNormCheck

\ottdefnEtaExpand

\ottdefnSCSynth

\ottdefnSCCheck 

\ottdefnSWF

\ottdefnStaticType

\ottdefnStaticTypeNorm

\ottdefnSHsub

\ottdefnSHsubR 

\ottdefnSimpleStep

\subsection{\lang Surface Language}

\begin{minipage}[t]{0.45\textwidth}
    \nonterms{gt,canonical}
\end{minipage}
\begin{minipage}[t]{0.45\textwidth}
    \nonterms{spine,atomic,Gamma}
\end{minipage}

\ottdefnGSynth
 
\ottdefnGCheck

\ottdefnGNSynth
 
\drules{$\gradual{Gamma}  \vdash  \gradual{u}  \leadsfrom\  \gradual{t}  \Leftarrow  \gradual{U}$}{\ottcom{Approximate Normalization Checking}}{GNCheckSynth, GNCheckApprox, GNCheckLevel, GNCheckPiType, GNCheckPiDyn, GNCheckLamPi, GNCheckLamPiRdxAlpha, GNCheckLamDyn} 

\ottdefnGradualType

\ottdefnGradualTypeNorm

\ottdefnGEtaLong

\ottdefnGEtaExpand

\ottdefnGEtaExpandC
 
\ottdefnGCSynth

\ottdefnGCCheck

\ottdefnWF

\ottdefnGradualEnvSub
 
\ottdefnGHsub

\ottdefnGHsubR

\ottdefnConsistent

\ottdefnConsistentType  

\ottdefnMeet

\ottdefnMorePrecise

\ottdefnDomain

\ottdefnCodSub

\ottdefnBodySub

\subsection{\lang Elaboration and Evidence Term Rules}

\begin{minipage}[t]{0.32\textwidth} 
    \nonterms{et} 
    \end{minipage}
    \begin{minipage}[t]{0.32\textwidth}
    \nonterms{dummyev,rv}  
    \end{minipage}
    \begin{minipage}[t]{0.32\textwidth}
        \nonterms{evalContext}  
        \nonterms{epsilon} 
        \end{minipage}

\drules{$\gradual{Gamma}  \vdash   \evterm{t} \ \leadsto \gradual{u}    \Rightarrow  \gradual{U}$}{\ottcom{Additional Synthesis Rule for Evidence Term Normalization}}{GNSynthEv} 

\ottdefnEvConsistent

\ottdefnGElabSynth 

\ottdefnGElabCheck 

\ottdefnEvType 

\ottdefnStep

\section{Proofs}

\subsection{Properties of Normalization and Hereditary Substitution}

The following formulation of type preservation for substitution is based on that of \citet{churchCurry}.

\begin{lemma}[Approximate Substitution Preserves Typing]
            Suppose $\gradual{Gamma}' \, \ottsym{(}  \mathit{x}  \ottsym{:}  \gradual{U}  \ottsym{)}  \gradual{Gamma}  \vdash   \mathit{y} \gradual{e}   \Rightarrow  \gradual{U}'$, $\gradual{Gamma}  \vdash  \gradual{u}  \Leftarrow  \gradual{U}$.
        Let $ {[  \gradual{u}  / { \mathit{x} } ]}^{ \gradual{U}  }  \gradual{Gamma}'  =  \gradual{Gamma}'' $, $ {[  \gradual{u}  / { \mathit{x} } ]}^{ \gradual{U}  }  \mathit{x} \gradual{e}  = { \gradual{u}'' } : { \gradual{U}''' } $,
        and $ {[  \gradual{u}  / { \mathit{x} } ]}^{ \gradual{U}  }  \gradual{U}'  =  \gradual{U}'' $.
        \begin{enumerate} 
            \item If $x \neq y$, then $\gradual{u}'' =  \mathit{y} \gradual{e}' $ and $\gradual{Gamma}'' \, \gradual{Gamma}  \vdash   \mathit{y} \gradual{e}'   \Rightarrow  \gradual{U}''$, and $\gradual{U}'' = \gradual{U}'''$
            \item If $x = y$, then $\gradual{Gamma}'' \, \gradual{Gamma}  \vdash  \gradual{u}''  \Leftarrow  \gradual{U}''$
        \end{enumerate}

        Similarly, suppose $\gradual{Gamma}' \, \ottsym{(}  \mathit{x}  \ottsym{:}  \gradual{U}  \ottsym{)}  \gradual{Gamma}  \vdash  \gradual{u}'  \Leftarrow  \gradual{U}'$, $\gradual{Gamma}  \vdash  \gradual{u}  \Leftarrow  \gradual{U}$.
        Let $ {[  \gradual{u}  / { \mathit{x} } ]}^{ \gradual{U}  }  \gradual{Gamma}'  =  \gradual{Gamma}'' $, $ {[  \gradual{u}  / { \mathit{x} } ]}^{ \gradual{U}  }  \gradual{u}'  =  \gradual{u}'' $,
        and $ {[  \gradual{u}  / { \mathit{x} } ]}^{ \gradual{U}  }  \gradual{U}'  =  \gradual{U}'' $. 

        \begin{enumerate}[resume]  
        \item Then $\gradual{Gamma}'' \, \gradual{Gamma}  \vdash  \gradual{u}''  \Leftarrow  \gradual{U}''$  
\end{enumerate}
\end{lemma}
\begin{proof}
    By induction on the derivation of $\gradual{u}''$. For most rules, the typing derivation is trivial
    to construct using our inductive hypotheses, so we give the interesting rules here.
    
    \begin{itemize}
     
    \item \rrule{GHsubDiffNil}:
    Since $\mathit{y}$ is unchanged by substitution, we must show that substitution on the environment produces the right type.
    If $\ottsym{(}  \mathit{y}  \ottsym{:}  \gradual{U}'  \ottsym{)} \, \in \, \gradual{Gamma}$, then $\gradual{U}'$ cannot contain $\mathit{x}$, so $\ottsym{(}  \mathit{y}  \ottsym{:}  \gradual{U}'  \ottsym{)} \, \in \, \gradual{Gamma}$
    and $ {[  \gradual{u}  / { \mathit{x} } ]}^{ \gradual{U}  }  \gradual{U}'  =  \gradual{U}' $.
    If $\ottsym{(}  \mathit{y}  \ottsym{:}  \gradual{U}'  \ottsym{)} \, \in \, \gradual{Gamma}'$, then $\ottsym{(}  \mathit{y}  \ottsym{:}  \gradual{U}''  \ottsym{)} \, \in \, \gradual{Gamma}''$ by the definition of substitution on environments.
    In both cases, $\mathit{y}$ synthesizes the desired type.
        
    \item \rrule{GHsubDiffCons}: 
    Suppose $\gradual{Gamma}' \, \ottsym{(}  \mathit{x}  \ottsym{:}  \gradual{U}  \ottsym{)}  \gradual{Gamma}  \vdash   \mathit{y}  \gradual{e} \  \gradual{u}'    \Rightarrow  \gradual{U}'$.
    Then $\gradual{Gamma}' \, \ottsym{(}  \mathit{x}  \ottsym{:}  \gradual{U}  \ottsym{)}  \gradual{Gamma}  \vdash   \mathit{y} \gradual{e}   \Rightarrow  \gradual{U}'_{{\mathrm{0}}}$, the domain of $\gradual{U}'_{{\mathrm{0}}}$ is defined,
    and $ [  \gradual{u}'  / {\_} ]  \ottkw{cod}\  \gradual{U}'_{{\mathrm{0}}}  =  \gradual{U}' $
    
    If $\gradual{U}'_{{\mathrm{0}}} =  {\qm } $, then $ \ottkw{dom}\  \gradual{U}'_{{\mathrm{0}}} =  {\qm } $,
    and $\gradual{Gamma}' \, \ottsym{(}  \mathit{x}  \ottsym{:}  \gradual{U}  \ottsym{)}  \gradual{Gamma}  \vdash  \gradual{u}'  \Leftarrow  {\qm }$,
    and $\gradual{U}' = \gradual{U}'' =  {\qm } $.
    If $\gradual{e}', \gradual{u}''$ are the substituted versions of $\gradual{e},\gradual{u}'$ respectively,
    by our hypothesis we know that
    $\gradual{Gamma}'' \, \gradual{Gamma}'  \vdash   \mathit{y} \gradual{e}'   \Rightarrow  {\qm }$
    and $\gradual{Gamma}'' \, \gradual{Gamma}'  \vdash  \gradual{u}''  \Leftarrow  {\qm }$.
    By the definition of codomain, $ [  \gradual{u}''  / {\_} ]  \ottkw{cod}\  {\qm }  =  {\qm } $,
    which gives us enough to construct our desired typing derivation.
    
    \item \rrule{GHsubRHead}: 
    In this case, $\gradual{u}' = \mathit{x}$, $\gradual{U}' = \gradual{U}$ and $\gradual{u}'' = \gradual{u}$.
    By our premise, $\gradual{Gamma}'' \, \ottsym{(}  \mathit{x}  \ottsym{:}  \gradual{U}  \ottsym{)}  \gradual{Gamma}  \vdash  \mathit{x}  \Rightarrow  \gradual{U}$, with $\gradual{U}' = \gradual{U}$
    and for the context to be well formed, this means that
    $\mathit{x}$ cannot occur in $\gradual{U}$, so $\gradual{U}'' = \gradual{U}$.
    By our premise, combined with the fact that  typing is preserved under context extension, 
    $\gradual{Gamma}'' \, \gradual{Gamma}  \vdash  \gradual{u}  \Leftarrow  \gradual{U}$. 
        
    \item \rrule{GHsubRLamSpine}: 
    In this case, $\gradual{u}' =  \mathit{x}  \gradual{e} \  \gradual{u}'_{{\mathrm{1}}}  $,
    $ {[  \gradual{u}  / { \mathit{x} } ]}^{ \gradual{U}_{{\mathrm{1}}}  }  \mathit{x} \gradual{e}  = { \ottsym{(}  \lambda  \mathit{y}  \ldotp  \gradual{u}_{{\mathrm{2}}}  \ottsym{)} } : { \ottsym{(}  \mathit{y}  \ottsym{:}  \gradual{U}_{{\mathrm{1}}}  \ottsym{)}  \rightarrow  \gradual{U}_{{\mathrm{2}}} } $,
    and $ {[  \gradual{u}  / { \mathit{x} } ]}^{ \gradual{U}_{{\mathrm{1}}}  }  \gradual{u}'  =  \gradual{u}'_{{\mathrm{2}}} $.  
    By our hypothesis, $\gradual{Gamma}'' \, \gradual{Gamma}  \vdash  \gradual{u}'_{{\mathrm{2}}}  \Leftarrow  \gradual{U}_{{\mathrm{2}}}$
    where $\gradual{Gamma}'' \, \gradual{Gamma}  \vdash  \ottsym{(}  \lambda  \mathit{y}  \ldotp  \gradual{u}_{{\mathrm{2}}}  \ottsym{)}  \Leftarrow  \ottsym{(}  \mathit{y}  \ottsym{:}  \gradual{U}_{{\mathrm{1}}}  \ottsym{)}  \rightarrow  \gradual{U}_{{\mathrm{2}}}$.
    However, the function must be typed using \rrule{GCCheckLamPi}, 
    meaning that $\ottsym{(}  \mathit{y}  \ottsym{:}  \gradual{U}_{{\mathrm{1}}}  \ottsym{)}  \gradual{Gamma}'' \, \gradual{Gamma}  \vdash  \gradual{u}_{{\mathrm{2}}}  \Leftarrow  \gradual{U}_{{\mathrm{2}}}$.
    We can once again apply our inductive hypothesis here to see that
    $\gradual{Gamma}'' \, \gradual{Gamma}  \vdash  \gradual{u}'_{{\mathrm{3}}}  \Leftarrow  \gradual{U}'$, proving our result.
    
    \item \rrule{GHsubRDynType}, \rrule{GHsubRDynSpine}, \rrule{GHsubRLamSpineOrd}: The result holds since $ {\qm } $ checks against any type
    that is itself well-typed. $ {\qm } $ checks against any $ \TypeType_{ \ottnt{i} } $. For  \rrule{GHsubRDynSpine},
    our hypothesis says that $\gradual{Gamma}'' \, \gradual{Gamma}'  \vdash  {\qm }  \Leftarrow  \ottsym{(}  \mathit{y}  \ottsym{:}  \gradual{U}'_{{\mathrm{1}}}  \ottsym{)}  \rightarrow  \gradual{U}'_{{\mathrm{2}}}$, so 
    $\gradual{Gamma}'' \, \gradual{Gamma}'  \vdash  \ottsym{(}  \mathit{y}  \ottsym{:}  \gradual{U}'_{{\mathrm{1}}}  \ottsym{)}  \rightarrow  \gradual{U}'_{{\mathrm{2}}}  \Leftarrow   \TypeType_{ \ottnt{i} } $ for some $\ottnt{i}$.
    This in turn means that $\ottsym{(}  \mathit{y}  \ottsym{:}  \gradual{U}'_{{\mathrm{1}}}  \ottsym{)}  \gradual{Gamma}'' \, \gradual{Gamma}'  \vdash  \gradual{U}'_{{\mathrm{2}}}  \Leftarrow   \TypeType_{ \ottnt{i} } $.
    We can again use our hypothesis to show that $\gradual{Gamma}'' \, \gradual{Gamma}  \vdash  \mathtt{V3}  \Leftarrow   \TypeType_{ \ottnt{i} } $.

    \end{itemize}
    
    Note that for the last two cases, the induction is well-founded: since $\gradual{Gamma}''$ and $\gradual{Gamma}$
    are not part of the derivation of the substitution, we can quantify them universally
    in our inductive hypothesis.

    \end{proof}

    \typeNorm*
    \begin{proof}
        We perform mutual induction, proving that if $\gradual{Gamma}  \vdash  \gradual{t}  \leadsto  \gradual{u}  \Rightarrow  \gradual{U}$, then $\gradual{Gamma} \vdash \gradual{u} \Leftarrow \gradual{U}$.
        All cases are simple, except for the following:

        \rrule{GNCheckSynth}: since $ \gradual{U}' \sqsubseteq  \gradual{U} $, the static gradual guarantee gives us that
        $\gradual{t}$ checks against $\gradual{U}$.

        \rrule{GNCheckApprox}: $ {\qm } $ checks against any type.

        \rrule{GNSynthApp}: while we only know that the normal form of $\gradual{t}_{{\mathrm{1}}}$ checks against the type that $\gradual{t}_{{\mathrm{1}}}$
        synthesizes, this is enough to fulfill the premises of preservation of typing by approximate substitution,
        which, combined with our hypothesis, gives us our result.

    \end{proof}
 
    \begin{lemma}
        \label{lem:subTotal}
        Suppose $\gradual{Gamma}  \vdash  \gradual{u}  \Leftarrow  \gradual{U}$ and $\gradual{Gamma}' \, \ottsym{(}  \mathit{x}  \ottsym{:}  \gradual{U}  \ottsym{)}  \gradual{Gamma}  \vdash  \gradual{u}'  \Leftarrow  \gradual{U}'$.
        Then there exists some $\gradual{u}''$ such that $ {[  \gradual{u}  / { \mathit{x} } ]}^{ \gradual{U}  }  \gradual{u}'  =  \gradual{u}'' $.
    
        Similarly, if $\gradual{Gamma}  \vdash  \gradual{u}  \Leftarrow  \gradual{U}$ and $\gradual{Gamma}' \, \ottsym{(}  \mathit{x}  \ottsym{:}  \gradual{U}  \ottsym{)}  \gradual{Gamma}  \vdash  \gradual{rr}  \Rightarrow  \gradual{U}'$.
        Then there exists some $\gradual{u}''$ such that $ {[  \gradual{u}  / { \mathit{x} } ]}^{ \gradual{U}  }  \gradual{rr}  =  \gradual{u}'' $,
    and if $\gradual{rr} =  \mathit{x}  \gradual{e} \  \gradual{u}'  $, $ {[  \gradual{u}  / { \mathit{x} } ]}^{ \gradual{U}  }  \mathit{x}  \gradual{e} \  \gradual{u}'   = { \gradual{u}'' } : { \gradual{U}' } $.
    \end{lemma}
    \begin{proof}
        We perform nested induction: first, on the multiset of level annotations on arrow types in $\gradual{U}$, then
        on the typing derivation for $\gradual{u}'$.
             
            \begin{itemize}
                \item \rrule{GCSynthType}, \rrule{GCCheckDyn}: take $\gradual{u}'' = \gradual{u}'$.
                \item \rrule{GCSynthVar}: take $\gradual{u}'' = \mathit{y}$ if $ \mathit{x}   \neq   \mathit{y} $, $\gradual{u}$ otherwise.
                \item \rrule{GCCheckLamDyn}, \rrule{GCCheckLamPi}, \rrule{GCCheckPi}, \rrule{GCCheckSynth}: 
                follows immediately from our hypothesis.
                \item \rrule{GCSynthApp}: This case is the most interesting. Suppose we are substituting into $ \mathit{x}  \gradual{e} \  \gradual{u}'  $.  
                If $ \mathit{x}   \neq   \mathit{y} $, then $ \mathit{y} \gradual{e} $ and $\gradual{u}'$ must both be well typed, so we can apply our hypothesis
                to find their substituted forms. 
    
                Assume then that $\mathit{x} = \mathit{y}$. By our premise, there must be some type such that
                $\gradual{Gamma}' \, \ottsym{(}  \mathit{x}  \ottsym{:}  \gradual{U}  \ottsym{)}  \gradual{Gamma}  \vdash  \ottsym{(}   \mathit{x} \gradual{e}   \ottsym{)}  \Rightarrow  \gradual{U}_{{\mathrm{3}}}$, and it must have a defined domain.
                
                If $\gradual{U}_{{\mathrm{3}}} =  {\qm } $, then by our hypothesis, there's some value where $ {[  \gradual{u}  / { \mathit{x} } ]}^{ \gradual{U}  }  \mathit{x} \gradual{e}  = { \gradual{u}'' } : { {\qm } } $,
                meaning we can apply \rrule{GHsubRDynType} to produce $ {\qm }  :  {\qm } $ as our result.
                Note that if $\gradual{U} =  {\qm } $, we necessarily have this case.                 
                Otherwise, for the domain to be defined, $\gradual{U}_{{\mathrm{3}}} = \ottsym{(}  \mathit{y}  \ottsym{:}  \gradual{U}_{{\mathrm{1}}}  \ottsym{)}  \rightarrow  \gradual{U}_{{\mathrm{2}}}$.
                By our hypothesis, there must be some value such that $ {[  \gradual{u}  / { \mathit{x} } ]}^{ \gradual{U}  }  \mathit{x} \gradual{e}  = { \gradual{u}'' } : { \gradual{U}'' } $,
                and moreover, preservation of typing says that $ {[  \gradual{u}  / { \mathit{x} } ]}^{ \gradual{U}  }  \ottsym{(}  \mathit{y}  \ottsym{:}  \gradual{U}_{{\mathrm{1}}}  \ottsym{)}  \rightarrow  \ottsym{(}  \gradual{U}_{{\mathrm{2}}}  \ottsym{)}  =  \gradual{U}'' $
                and that $\gradual{Gamma}'' \, \gradual{Gamma}  \vdash  \gradual{u}''  \Leftarrow  \gradual{U}''$, where $\gradual{Gamma}''$ is $\gradual{Gamma}'$ with $\gradual{u}$ substituted for $\mathit{x}$.
                But this means that $\gradual{U}'' = \ottsym{(}  \mathit{y}  \ottsym{:}  \gradual{U}'_{{\mathrm{1}}}  \ottsym{)}  \rightarrow  \gradual{U}'_{{\mathrm{2}}}$, since we must have constructed it using \rrule{GHsubPi}.
                Then, if $\gradual{U}''$ an arrow type, it is not atomic, meaning that $\gradual{u}''$ could only check against it
                using \rrule{GCCheckDyn} or \rrule{GCCheckLamPi}.     
                In both cases, since the whole application is well typed, we know that
                $\gradual{Gamma}' \, \ottsym{(}  \mathit{x}  \ottsym{:}  \gradual{U}  \ottsym{)}  \gradual{Gamma}  \vdash  \gradual{u}'  \Leftarrow  \gradual{U}_{{\mathrm{1}}}$, so there's some value where
                $ {[  \gradual{u}  / { \mathit{x} } ]}^{ \gradual{U}  }  \gradual{u}'  =  \gradual{u}'_{{\mathrm{2}}} $.
                Preservation of typing gives that $\gradual{Gamma}'' \, \gradual{Gamma}  \vdash  \gradual{u}'_{{\mathrm{2}}}  \Leftarrow  \gradual{U}'_{{\mathrm{1}}}$.
                If $\gradual{U}'_{{\mathrm{1}}}$ is not less than $\gradual{U}$ in our multiset order,
                then we can apply \rrule{GHsubRLamSpineOrd}. Otherwise
                $\gradual{U}'_{{\mathrm{1}}}$ is less than $\gradual{U}$ in our multiset order.
                                
                In the first case, $\gradual{u}'' =  {\qm } $. Then we can use $ {\qm } $ as our return value, and
                by our outer inductive hypothesis, there must be some value such that  $ {[  \gradual{u}'_{{\mathrm{2}}}  / { \mathit{y} } ]}^{ \gradual{U}'_{{\mathrm{1}}}  }  \gradual{U}_{{\mathrm{2}}}  =  \gradual{U}' $,
                giving us our return type. This allows us to build the derivation with \rrule{GHsubRDynSpine}.
     
                In the second case, $\gradual{u}'' = \ottsym{(}  \lambda  \mathit{y}  \ldotp  \gradual{u}_{{\mathrm{2}}}  \ottsym{)}$.
                We can decompose the typing of $\gradual{u}''$ to see that
                $\ottsym{(}  \mathit{y}  \ottsym{:}  \gradual{U}'_{{\mathrm{1}}}  \ottsym{)}  \gradual{Gamma}'' \, \gradual{Gamma}  \vdash  \gradual{u}_{{\mathrm{2}}}  \Leftarrow  \gradual{U}'_{{\mathrm{2}}}$,
                which lets us apply our outer inductive hypothesis see that
                $ {[  \gradual{u}'_{{\mathrm{2}}}  / { \mathit{y} } ]}^{ \gradual{U}'_{{\mathrm{1}}}  }  \gradual{u}_{{\mathrm{2}}}  =  \gradual{u}'_{{\mathrm{3}}} $, which we can use as our final result.
                Similarly,  $ {[  \gradual{u}'_{{\mathrm{2}}}  / { \mathit{y} } ]}^{ \gradual{U}'_{{\mathrm{1}}}  }  \gradual{U}'_{{\mathrm{2}}}  =  \gradual{U}'_{{\mathrm{3}}} $ gives us our final type,
                which allows us to construct the derivation using \rrule{GHsubRLamSpine}.
            \end{itemize}
          
    \end{proof}

    The fact that the normalization rules directly correspond to typing rules, combined with the totality
    of hereditary substitution, is enough to give us our result for normalization.
    \normTotal*

    \subsection{Type Safety}

    \begin{lemma}[Progress]
        If $\cdot  \vdash  \evterm{t}'  \ottsym{:}  \gradual{U}$ and $\evterm{t}'$ is not a value, then there exists some $\evterm{t}$
        such that $\evterm{t}'  \longrightarrow  \evterm{t}$.
        \end{lemma}
        \begin{proof}
            By induction on the derivation on the typing of $\evterm{t}'$.
            \begin{itemize}
                \item \rrule{EvTypeType}, \rrule{EvTypeLevel}, \rrule{EvTypeDyn}, \rrule{EvTypeLam}: must be values 
                \item \rrule{EvTypeEv},\rrule{EvTypePi}: either we have a value, or can step with our context rules.
                \item \rrule{EvTypeVar}: cannot be typed under empty environment
                \item \rrule{EvTypeApp}: then $\evterm{t}'= \evterm{t}_1 \  \evterm{t}_2 $, where $\cdot  \vdash  \evterm{t}_1  \ottsym{:}  \ottsym{(}  \mathit{x}  \ottsym{:}  \gradual{U}  \ottsym{)}  \rightarrow  \gradual{U}'$.
                and $\Gamma \vdash \evterm{t}_2 : \gradual{U}$.
                If either of $\evterm{t}_1$ and $\evterm{t}_2$ is not a value, then we can step by the context rule.
                Otherwise, both $\evterm{t}_1$ and $\evterm{t}_2$  are values,
                and by inversion on our typing rules, we know that $\evterm{t}_1 = \langle \gradual{U}_1 \rangle \evterm{rv}$,
                where $\gradual{U}_1 \cong \ottsym{(}  \mathit{x}  \ottsym{:}  \gradual{U}  \ottsym{)}  \rightarrow  \gradual{U}'$.
                This means that $\ottkw{dom}$ and $\ottkw{cod}$ are both defined on $\gradual{U}_1$.
                                If $\evterm{rv}$ is $\qm$, we can step by \rrule{StepAppDyn}.
                Otherwise, $\evterm{rv}$ is $\lambda x \ldotp \evterm{v}$.
                First consider when $\evterm{t}_2$ is annotated with evidence.
                We can then either apply 
                \rrule{StepAppEv} or  \rrule{StepAppFailTrans},
                depending on whether the meet with the evidence of $\evterm{t}_2$ is defined.
                In the case that $\evterm{t}_2$ has no evidence, we repeat the above process after applying \rrule{StepAppEvRaw}. 
            \end{itemize}
        \end{proof}
        
        \begin{lemma}[Preservation]
            If $ \cdot  \vdash  \evterm{t}'  \ottsym{:}  \gradual{U}$ and $\evterm{t}'  \longrightarrow  \evterm{t}$, then $\evterm{t} = \mathsf{err}$ or $\cdot  \vdash  \evterm{t}  \ottsym{:}  \gradual{U}$.
            \end{lemma}
        \begin{proof}
            By induction on the derivation of $\evterm{t}'  \longrightarrow  \evterm{t}$.
            \begin{itemize}
                \item \rrule{StepAscrFail},\rrule{StepAppFailTrans},\rrule{StepContextErr}: 
                trivial, since we step to $\mathsf{err}$.
                \item \rrule{StepAscr}: then $\evterm{t}' = \langle  \gradual{U}_{{\mathrm{1}}}  \rangle \, \ottsym{(}  \langle  \gradual{U}_{{\mathrm{2}}}  \rangle \, \evterm{t}''  \ottsym{)}$, and by inversion on typing, we know that
                $\cdot  \vdash  \evterm{t}''  \ottsym{:}  \gradual{U}'$ where $ \langle  \gradual{U}_{{\mathrm{2}}}  \rangle \vdash  \gradual{U}'  \cong  \gradual{U} $. 
                 Since $  \gradual{U}_{{\mathrm{1}}}  \sqcap  \gradual{U}_{{\mathrm{2}}}  \sqsubseteq  \gradual{U}_{{\mathrm{2}}} $ then $ \langle   \gradual{U}_{{\mathrm{1}}}  \sqcap  \gradual{U}_{{\mathrm{2}}}   \rangle \vdash  \gradual{U}'  \cong  \gradual{U} $, which gives us $\cdot  \vdash  \langle   \gradual{U}_{{\mathrm{1}}}  \sqcap  \gradual{U}_{{\mathrm{2}}}   \rangle \, \evterm{t}''  \ottsym{:}  \gradual{U}$.
                                                 \item \rrule{StepAppEv}: then $\evterm{t}' =  \ottsym{(}  \langle  \gradual{U}_{{\mathrm{1}}}  \rangle \, \ottsym{(}  \lambda  \mathit{x}  \ldotp  \evterm{t}''  \ottsym{)}  \ottsym{)} \  \ottsym{(}  \langle  \gradual{U}_{{\mathrm{2}}}  \rangle \, \evterm{rv}  \ottsym{)} $.
                 By inversion on typing, $\cdot  \vdash  \lambda  \mathit{x}  \ldotp  \evterm{t}''  \ottsym{:}  \gradual{U}'_{{\mathrm{1}}}$ where $ \langle  \gradual{U}_{{\mathrm{1}}}  \rangle \vdash  \gradual{U}'_{{\mathrm{1}}}  \cong  \gradual{U}'_{{\mathrm{2}}} $,
                 and the whole expression $\evterm{t}'$ types against the codomain of $\gradual{U}'_{{\mathrm{2}}}$.
                 We know that $\gradual{Gamma}  \vdash  \langle  \gradual{U}_{{\mathrm{2}}}  \rangle \, \evterm{rv}  \ottsym{:}   \ottkw{dom}\  \gradual{U}'_{{\mathrm{2}}} $, so $ \langle  \gradual{U}_{{\mathrm{2}}}  \rangle \vdash   \ottkw{dom}\  \gradual{U}'_{{\mathrm{2}}}   \cong  \gradual{U}_{{\mathrm{3}}} $ and $\cdot  \vdash  \evterm{rv}  \ottsym{:}  \gradual{U}_{{\mathrm{3}}}$.
                  
                 We then know that $ \langle   \gradual{U}_{{\mathrm{2}}}  \sqcap   \ottkw{dom}\  \gradual{U}_{{\mathrm{1}}}    \rangle \vdash   \ottkw{dom}\  \gradual{U}'_{{\mathrm{1}}}   \cong  \gradual{U}'_{{\mathrm{2}}} $.                  So $\cdot  \vdash  \langle   \gradual{U}_{{\mathrm{2}}}  \sqcap   \ottkw{dom}\  \gradual{U}_{{\mathrm{1}}}    \rangle \, \evterm{rv}  \ottsym{:}   \ottkw{dom}\  \gradual{U}'_{{\mathrm{1}}} $. 
                 By preservation of typing under substitution, substituting this into $\evterm{t}''$ has type $ {[  \evterm{rv}  / {\_} ]} \ottkw{cod}\  \gradual{U}'_{{\mathrm{1}}} $.
                 Finally, we then know that if $ \langle   {[  \evterm{rv}  / {\_} ]} \ottkw{cod}\  \gradual{U}_{{\mathrm{1}}}   \rangle \vdash   {[  \evterm{rv}  / {\_} ]} \ottkw{cod}\  \gradual{U}'_{{\mathrm{1}}}   \cong   {[  \evterm{rv}  / {\_} ]} \ottkw{cod}\  \gradual{U}'_{{\mathrm{2}}}  $.
                                 This means that our final result can be typed at $\gradual{U}'_{{\mathrm{2}}}$.
                 \item \rrule{StepAppDyn}: holds by preservation of typing under codomain substitution.
                 \item \rrule{StepAppEvRaw}: similar reasoning as for \rrule{StepAppEv}, except with less indirection since we know
                 $\cdot  \vdash  \evterm{rv}  \ottsym{:}   \ottkw{dom}\  \gradual{U}'_{{\mathrm{2}}} $.
                 \item \rrule{StepContext}: holds by our inductive hypothesis and preservation under substitution (since hole-filling is a special case of substitution).

            \end{itemize}
        \end{proof}

These together give us type safety.
\typeSound*

\subsection{Soundness and Optimality of $\alpha$ with respect to $\gamma$}

We now show that our AGT functions form a Galois-connection. 

\begin{theorem}[Soundness of $\alpha$]
    If $A \neq \emptyset$, then $A\subseteq \gamma(\alpha(A))$.
\end{theorem}
\begin{proof} 
    By induction on the structure of $\alpha(A)$.

    If $\alpha(A)= \mathit{x}  \gradual{e}' \  \gradual{u}'  $ then $A = \set{ \mathit{x}  \static{e} \  \static{u}   \mid  \mathit{x} \static{e}  \in B_1, \static{u} \in B_2}$ for some $B_1, B_2$ where $ \alpha(B_1) =  \mathit{x} \gradual{e}' , \alpha(B_2) = \gradual{u}'$.
    Then $\gamma( \mathit{x}  \gradual{e}' \  \gradual{u}'  ) = \set{ \mathit{x}  \static{e}' \  \static{u}'   \mid  \mathit{x} \static{e}'  \in \gamma( \mathit{x} \gradual{e}' ), \static{u}' \in \gamma(\gradual{u}')}$.
    By our hypothesis, $B_1 \subseteq \gamma( \mathit{x} \gradual{e}' )$ and $B_2 \subseteq \gamma(\gradual{u}')$, so 
    $\set{ \mathit{x}  \static{e} \  \static{u}   \mid  \mathit{x} \static{e}  \in B_1, \static{u} \in B_2} \subseteq \set{ \mathit{x}  \static{e}' \  \static{u}'   \mid  \mathit{x} \static{e}'  \in \gamma( \mathit{x} \gradual{e}' ), \static{u}' \in \gamma(\gradual{u}')}$.
    The cases for $\lambda$ and $\rightarrow$ are proved in the same way.

    If $\alpha(A)=\mathit{x}$, then $A = \set{\mathit{x}}$ and $\gamma(\alpha(A)) = \set{\mathit{x}}$. The same logic proves the case for $ \TypeType $.

    If $\alpha(A) =  {\qm } $, $\gamma(\alpha(A)) = \AllTerms \supseteq A$ by definition. 
\end{proof}

\begin{theorem}[Optimality of $\alpha$]
    If $A \neq \emptyset$, $A\subseteq \gamma(\gradual{u})$ then $\alpha(A) \sqsubseteq \gradual{u}$.
\end{theorem}
\begin{proof}
By induction on the structure of $\gradual{u}$.

Case $ \TypeType_{ \ottnt{i} } $: then $\gamma(\gradual{u}) = \set{ \TypeType_{ \ottnt{i} } }$, so if $A$ is nonempty $A = \set{ \TypeType_{ \ottnt{i} } }$.
The same reasoning holds if $\gradual{u}=\mathit{x}$.

Case $ \mathit{x}  \gradual{e} \  \gradual{u}  $: then $\gamma(\gradual{u}) = \set{ x e' u' \mid  \mathit{x} \static{e}'  \in \gamma( \mathit{x} \gradual{e} ), \static{u}' \in \gamma(\gradual{u})}$.
If $A$ is nonempty and $A \subseteq \gamma(\gradual{u})$, then there are some set $B_1, B_2$ such that 
$B_1 \subseteq \gamma( \mathit{x} \gradual{e} )$, $B_2 \subseteq \gamma(\gradual{u})$, and
$A = \set{ \mathit{x}  \static{e}_{{\mathrm{2}}} \  \static{u}_{{\mathrm{2}}}   \mid  \mathit{x} \static{e}_{{\mathrm{2}}}  \in B_1, \static{u}_{{\mathrm{2}}} \in B_2}$.
So $\alpha(A) =  \mathit{x}  \gradual{e}_{{\mathrm{3}}} \  \gradual{u}_{{\mathrm{3}}}  $ where $\alpha(B_1) =  \mathit{x} \gradual{e}_{{\mathrm{3}}} $, $\alpha(B_2) = \gradual{u}_{{\mathrm{3}}}$. 
Neither $B_1$ or $B_2$ can be empty if $A$ is non-empty, so by our inductive hypothesis,
$  \mathit{x} \gradual{e}_{{\mathrm{3}}}  \sqsubseteq   \mathit{x} \gradual{e}  $ and $ \gradual{u}_{{\mathrm{3}}} \sqsubseteq  \gradual{u} $. So $  \mathit{x}  \gradual{e}_{{\mathrm{3}}} \  \gradual{u}_{{\mathrm{3}}}   \sqsubseteq   \mathit{x}  \gradual{e} \  \gradual{u}   $, giving us our result.
The same reasoning holds for the $\lambda$ and $\to$ cases.

Case $?$: we have our result, since $ \gradual{u}' \sqsubseteq  {\qm } $ holds for all $gu'$.
\end{proof}

\subsection{Embeddings and Gradual Guarantees}

\untypedEmbed*
\begin{proof} 
    For the typing, we perform induction on $\lceil t \rceil$. 

    For a variable $x$, by our premise, its type in $\gradual{Gamma}$ is $ {\qm } $.
  
    If $\lceil t \rceil =  \gradual{t}' \  \gradual{t}'' $, then by our premise $\gradual{t}'$ and $\gradual{t}''$ both synthesize $ {\qm } $.
    This means that $\gradual{t}''$ checks against $ {\qm } $ and $ {\qm } $ checks against $ \TypeType_{ \ottnt{i} } $ for any $\ottnt{i}$.
    Finally, the codomain substitution of $ {\qm } $ is $ {\qm } $, so $ \gradual{t}' \  \gradual{t}'' $ synthesizes $ {\qm } $.
 
    If $ \lceil t \rceil =  \ottsym{(}  \lambda  \mathit{x}  \ldotp  \gradual{t}'  \ottsym{)} \dblcolon {\qm } $, then $ {\qm } $ is well-kinded with canonical form $ {\qm } $. We need to show that $\ottsym{(}  \lambda  \mathit{x}  \ldotp  \gradual{t}'  \ottsym{)}$
    checks against $ {\qm } $. To do so, we apply the rule \rrule{GCheckLamDyn}. By our hypothesis,
    $\gradual{t}'$ synthesizes $ {\qm } $ in context $\ottsym{(}  \mathit{x}  \ottsym{:}  {\qm }  \ottsym{)}  \gradual{Gamma}$, so it must check against $ {\qm } $.

    For the semantics, we  begin with some notations and facts.
    Let $\lfloor t \rfloor$ denote the elaboration of $\lceil t \rceil$.
    First, we note $\lfloor t \rfloor$ is simply $\lceil t \rceil$ with each $\qm$ ascription on a function replaced by the evidence $\langle \qm \to \qm \rangle$.
    This means that at any point during execution, any evidence will be $\qm$ or $\qm \to \qm$, so the meet, domain, and codomain are always defined.
     Next, we have $\lfloor [x \Mapsto v]t \rfloor =  [x \Mapsto \lfloor v \rfloor] \lfloor t \rfloor$ for any $x,v$ and $t$:
    this can be easily shown by induction on the structure of $t$.
    Similarly, we extend $\lfloor \_ \rfloor$ to contexts with $\lfloor \square \rfloor = \square$, with all other cases defined as for terms.
    Then, for any $t$ and $\mathcal{C}$, we have
    $\lfloor \mathcal{C}[t] \rfloor = \lfloor \mathcal{C}\rfloor[\lfloor t \rfloor] $. This again can be shown by  induction on the structure of $\mathcal{C}$.
    Finally we note that $\gradual{C}[\gradual{e}] \longrightarrow^{*} \gradual{C}[\gradual{e}']$ whenever
    $\gradual{e} \longrightarrow \gradual{e}'$: this is shown by induction on the length of the $\longrightarrow^{*}$ derivation.

    We then show that if  $t \longrightarrow t'$ then $\lfloor t \rfloor \longrightarrow^{*} \lfloor t' \rfloor$.
    We perform induction on the derivation of $t \longrightarrow t'$ using our simple semantics,
    ignoring the case for \rrule{SimpleStepAnn}, since by our premise $t$ is untyped.
        In the case for \rrule{SimpleStepApp}, we have $t = (\lambda x \ldotp t')v$, so
    $\lfloor t \rfloor = \varepsilon (\lambda x \ldotp \lfloor t' \rfloor) \lfloor v \rfloor$,
    where $\varepsilon \in \{ \langle \qm \to \qm \rangle, \langle \qm \rangle \}$.
    We can then apply either \rrule{StepAppEv} or \rrule{StepAppEvRaw}, and then apply substitution commuting with $\lfloor \_ \rfloor$
    to obtain our result.
    For \rrule{SimpleStepContext}, the result follows from the latter two facts we presented above.

    Induction on the number of steps then shows that $t \longrightarrow^{*} t'$ implies that $\lfloor t \rfloor \longrightarrow^{*} \lfloor t' \rfloor$,
    which in turn implies our result.

\end{proof}

\begin{lemma}
    \label{lem:sub-precise}
    Suppose $ \gradual{u}_{{\mathrm{1}}} \sqsubseteq^\eta  \gradual{u}'_{{\mathrm{1}}} $,  $ \gradual{u}_{{\mathrm{2}}} \sqsubseteq^\eta  \gradual{u}'_{{\mathrm{2}}} $ and $ \gradual{U} \sqsubseteq^\eta  \gradual{U}' $ and $ \gradual{U}_{{\mathrm{2}}} \sqsubseteq^\eta  \gradual{U}'_{{\mathrm{2}}} $,
    where $\gradual{Gamma}  \vdash  \gradual{u}_{{\mathrm{2}}}  \Leftarrow  \gradual{U}_{{\mathrm{2}}}$ and $\gradual{Gamma}'  \vdash  \gradual{u}'_{{\mathrm{2}}}  \Leftarrow  \gradual{U}'_{{\mathrm{2}}}$ for some $ \gradual{Gamma} \sqsubseteq^\eta  \gradual{Gamma}' $.
    If $ {[  \gradual{u}_{{\mathrm{1}}}  / { \mathit{x} } ]}^{ \gradual{U}  }  \gradual{u}_{{\mathrm{2}}}  =  \gradual{u}_{{\mathrm{3}}} $ and $ {[  \gradual{u}'_{{\mathrm{1}}}  / { \mathit{x} } ]}^{ \gradual{U}'  }  \gradual{u}'_{{\mathrm{2}}}  =  \gradual{u}'_{{\mathrm{3}}} $,
    then $ \gradual{u}_{{\mathrm{3}}} \sqsubseteq^\eta  \gradual{u}'_{{\mathrm{3}}} $.

    Suppose also that  $  \mathit{x} \gradual{e}  \sqsubseteq^\eta   \mathit{x} \gradual{e}'  $, where 
    $ {[  \gradual{u}_{{\mathrm{1}}}  / { \mathit{x} } ]}^{ \gradual{U}  }  \mathit{x} \gradual{e}  = { \gradual{u}_{{\mathrm{3}}} } : { \gradual{U}_{{\mathrm{3}}} } $, $ {[  \gradual{u}'_{{\mathrm{1}}}  / { \mathit{x} } ]}^{ \gradual{U}'  }  \mathit{x} \gradual{e}'  = { \gradual{u}'_{{\mathrm{3}}} } : { \gradual{U}'_{{\mathrm{3}}} } $,
    $\gradual{Gamma}  \vdash   \mathit{x} \gradual{e}   \Rightarrow  \gradual{U}_{{\mathrm{2}}}$ and $\gradual{Gamma}'  \vdash   \mathit{x} \gradual{e}'   \Rightarrow  \gradual{U}'_{{\mathrm{2}}}$.
    Then $ \gradual{u}_{{\mathrm{3}}} \sqsubseteq^\eta  \gradual{u}'_{{\mathrm{3}}} $ and $ \gradual{U}_{{\mathrm{3}}} \sqsubseteq^\eta  \gradual{U}'_{{\mathrm{3}}} $.
\end{lemma}
\begin{proof}
    First we note that if $\gradual{u}'_{{\mathrm{2}}}= {\qm } $, then $\gradual{u}'_{{\mathrm{3}}}= {\qm } $, and the result is trivially true.

    We then assume that $\gradual{u}'_{{\mathrm{2}}}\neq {\qm } $, and proceed by induction on the derivation of $ {[  \gradual{u}_{{\mathrm{1}}}  / { \mathit{x} } ]}^{ \gradual{U}  }  \gradual{u}_{{\mathrm{2}}}  =  \gradual{u}_{{\mathrm{3}}} $.

    \rrule{GHsubType}, \rrule{GHsubDiffNil}: In these cases, $\gradual{u}_{{\mathrm{2}}}=^\eta\gradual{u}'_{{\mathrm{2}}}=^\eta\gradual{u}_{{\mathrm{3}}}=^\eta\gradual{u}'_{{\mathrm{3}}}$. 

    \rrule{GHsubPi}, \rrule{GHsubLam}, : follows from the inductive hypothesis.

    \rrule{GHsubLam}: This is the case in which we must consider $\eta$-equality. Here, $\gradual{u}_{{\mathrm{2}}}=\lambda  \mathit{y}  \ldotp  \gradual{u}_{{\mathrm{4}}}$. 
    If $\gradual{u}'_{{\mathrm{2}}}=\lambda  \mathit{y}  \ldotp  \gradual{u}'_{{\mathrm{4}}}$, then the result follows from our inductive hypothesis. 
    Otherwise, it must be some $ \mathit{z} \gradual{e}' $ that $\eta$-expands to a (possibly) less-precise version of $\gradual{u}_{{\mathrm{2}}}$.
    Since both terms are well-formed, they are $\eta$-long with respect to $\gradual{U}_{{\mathrm{2}}}$ and $\gradual{U}'_{{\mathrm{2}}}$ respectively.
    So $\gradual{U}_{{\mathrm{2}}}$ must be $ {\qm } $ or an arrow-type, but $\gradual{U}'_{{\mathrm{2}}}$ cannot be an arrow type, so $\gradual{U}_{{\mathrm{2}}}= {\qm } $.
    Then we know that $\gradual{u}_{{\mathrm{4}}} \sqsubseteq^\eta \mathit{z}  \gradual{e}' \   \mathit{y}   $, so by our hypothesis,
    if $ {[  \gradual{u}_{{\mathrm{1}}}  / { \mathit{x} } ]}^{ \gradual{U}  }  \gradual{u}_{{\mathrm{4}}}  =  \gradual{u}_{{\mathrm{5}}} $ and $ {[  \gradual{u}'_{{\mathrm{1}}}  / { \mathit{x} } ]}^{ \gradual{U}'  }   \mathit{z}  \gradual{e}' \   \mathit{y}     =  \gradual{u}'_{{\mathrm{5}}} $
    then $\gradual{u}'_{{\mathrm{5}}}\sqsubseteq^\eta\gradual{u}'_{{\mathrm{5}}}$. We know that $\mathit{x}\neq\mathit{y}$ by our premise.
    If $\mathit{x}\neq \mathit{z}$, then we know that $\gradual{u}'_{{\mathrm{5}}}= \mathit{z}  \gradual{e}' \   \mathit{y}   $, so $\lambda  \mathit{y}  \ldotp  \gradual{u}'_{{\mathrm{5}}}\sqsubseteq^\eta \mathit{z} \gradual{e}' $.
    If $\mathit{x}=\mathit{z}$, then $[ \gradual{u}'_{{\mathrm{1}}} / \mathit{x} ]^{\gradual{U}'}  \mathit{z} \gradual{e}'  $ has type $ {\qm } $, so applying it to $\mathit{y}$
    produces $ {\qm } $. Then $\gradual{u}'_{{\mathrm{5}}}= {\qm } $, giving us our result.
 
    \rrule{GHsubSpine}:  Then $\gradual{u}_{{\mathrm{1}}} =  \mathit{x} \gradual{e} $.
    Since $ \gradual{U}_{{\mathrm{2}}} \sqsubseteq^\eta  \gradual{U}'_{{\mathrm{2}}} $, $\gradual{u}'_{{\mathrm{2}}}$ cannot possibly be $\eta$-expanded any further than $ \mathit{x} \gradual{e} $m
    so it must be equal to some $ \mathit{x} \gradual{e}' $. The result then follows from our hypothesis. 

    \rrule{GHsubDiffCons}: same logic as the previous case.

    \rrule{GHsubRHead}: Then $\gradual{u}_{{\mathrm{2}}}=\gradual{u}'_{{\mathrm{2}}}=\mathit{x}$, $\gradual{u}_{{\mathrm{3}}}=\gradual{u}_{{\mathrm{1}}}$ and $\gradual{u}'_{{\mathrm{3}}}=\gradual{u}'_{{\mathrm{1}}}$.
    Likewise, $\gradual{U}=\gradual{U}_{{\mathrm{2}}}$ and $\gradual{U}'=\gradual{U}'_{{\mathrm{2}}}$, so our result follows from our premises.
    The only other case is where $\gradual{u}'_{{\mathrm{2}}}$ is an $\eta$-expansion of $\mathit{x}$.

    \rrule{GHsubRDynSpine}: then $\gradual{u}_{{\mathrm{1}}}= \mathit{x}  \gradual{e} \  \gradual{u}_{{\mathrm{4}}}  $ and $\gradual{u}'_{{\mathrm{1}}} =  \mathit{x}  \gradual{e}' \  \gradual{u}'_{{\mathrm{4}}}  $, where $  \mathit{x} \gradual{e}  \sqsubseteq^\eta   \mathit{x} \gradual{e}'  $ and $ \gradual{u}_{{\mathrm{4}}} \sqsubseteq^\eta  \gradual{u}'_{{\mathrm{4}}} $.
    By our premise, $ {[  \gradual{u}_{{\mathrm{1}}}  / { \mathit{x} } ]}^{ \gradual{U}  }  \mathit{x} \gradual{e}  = { {\qm } } : { \gradual{U}_{{\mathrm{2}}} } $, so by our hypothesis, $ {[  \gradual{u}'_{{\mathrm{1}}}  / { \mathit{x} } ]}^{ \gradual{U}'  }  \mathit{x} \gradual{e}'  = { {\qm } } : { \gradual{U}'_{{\mathrm{2}}} } $, so
    $\gradual{e}'_{{\mathrm{3}}} =  {\qm } $, giving us our result.
 
    \rrule{GHsubRLamSpine}: follows from applying our inductive hypothesis to each sub-derivation. 
        
    \rrule{GHsubRDynType}: then $\gradual{u}_{{\mathrm{1}}}= \mathit{x}  \gradual{e} \  \gradual{u}_{{\mathrm{4}}}  $ and $\gradual{u}'_{{\mathrm{1}}} =  \mathit{x}  \gradual{e}' \  \gradual{u}'_{{\mathrm{4}}}  $, where $  \mathit{x} \gradual{e}  \sqsubseteq^\eta   \mathit{x} \gradual{e}'  $ and $ \gradual{u}_{{\mathrm{4}}} \sqsubseteq^\eta  \gradual{u}'_{{\mathrm{4}}} $.
    By our premise, $ {[  \gradual{u}_{{\mathrm{1}}}  / { \mathit{x} } ]}^{ \gradual{U}  }  \mathit{x} \gradual{e}  = { \gradual{u}_{{\mathrm{4}}} } : { {\qm } } $, so by our hypothesis, $ {[  \gradual{u}'_{{\mathrm{1}}}  / { \mathit{x} } ]}^{ \gradual{U}'  }  \mathit{x} \gradual{e}'  = { \gradual{u}'_{{\mathrm{4}}} } : { {\qm } } $,
    since $ {\qm } $ is the least precise type.
    Then $\gradual{u}_{{\mathrm{3}}}=\gradual{u}'_{{\mathrm{3}}}= {\qm } $.

\end{proof}

\begin{corollary}
    \label{cor:codsub-precise}
    If $ \gradual{u}_{{\mathrm{1}}} \sqsubseteq^\eta  \gradual{u}'_{{\mathrm{1}}} $, and $ \gradual{U} \sqsubseteq^\eta  \gradual{U}' $,
    $ [  \gradual{u}_{{\mathrm{1}}}  / {\_} ]  \ottkw{cod}\  \gradual{U}  =  \gradual{U}_{{\mathrm{2}}} $ and $ [  \gradual{u}'_{{\mathrm{1}}}  / {\_} ]  \ottkw{cod}\  \gradual{U}'  =  \gradual{U}'_{{\mathrm{2}}} $,
    then $ \gradual{U}_{{\mathrm{2}}} \sqsubseteq^\eta  \gradual{U}'_{{\mathrm{2}}} $.
\end{corollary}

\begin{corollary}
    \label{cor:bodysub-precise}
    If $ \gradual{u}_{{\mathrm{1}}} \sqsubseteq^\eta  \gradual{u}'_{{\mathrm{1}}} $,  $ \gradual{u}_{{\mathrm{2}}} \sqsubseteq^\eta  \gradual{u}'_{{\mathrm{2}}} $, $ \gradual{U} \sqsubseteq^\eta  \gradual{U}' $,
    $ {[  \gradual{u}_{{\mathrm{1}}}  / {\_} ]}^{ \gradual{U}  } \ottkw{body}\  \gradual{u}_{{\mathrm{2}}} =  \gradual{u}_{{\mathrm{3}}} $ and $ {[  \gradual{u}'_{{\mathrm{1}}}  / {\_} ]}^{ \gradual{U}'  } \ottkw{body}\  \gradual{u}'_{{\mathrm{2}}} =  \gradual{u}'_{{\mathrm{3}}} $,
    then $ \gradual{u}_{{\mathrm{3}}} \sqsubseteq^\eta  \gradual{u}'_{{\mathrm{3}}} $.
\end{corollary}

\normGuarantee*
\begin{proof}
    We first note that if $\gradual{t}_{{\mathrm{2}}}= {\qm } $, then $\gradual{u}_{{\mathrm{2}}}= {\qm } $, so $ \gradual{u}_{{\mathrm{1}}} \sqsubseteq  \gradual{u}_{{\mathrm{2}}} $.

    Assume then that $\gradual{t}_{{\mathrm{2}}}\neq {\qm } $. We proceed by induction on the derivation of the normalization of $\gradual{t}_{{\mathrm{1}}}$. We perform mutual induction
    to prove the same result for synthesis.

    \rrule{GNSynthAnn}: follows immediately from inductive hypothesis.
    
    \rrule{GNSynthType}, 
    \rrule{GNSynthDyn}: trivial, since $\gradual{t}_{{\mathrm{2}}}$ must be $ {\qm } $.

    \rrule{GNSynthVar}: trivial, since we consider precision modulo $\eta$-expansion

    \rrule{GNSynthApp}: follows from our inductive hypothesis,  \autoref{cor:bodysub-precise}, \autoref{cor:codsub-precise},
    and the transitivity of precision.

    \rrule{GNCheckSynth},\rrule{GNCheckLevel}, \rrule{GNCheckPiType}, \rrule{GNCheckPiDyn}, \rrule{GNCheckLamPi}, \rrule{GNCheckLamDyn}: follows immediately from our inductive hypothesis.

    \rrule{GNCheckSynth}: we note that decreasing the precision of a term can only decrease the precision of its synthesized type.
    Given this, we know that if $\gradual{t}$ synthesizes $\gradual{U}'$, then $\gradual{t}_{{\mathrm{2}}}$ synthesizes $\gradual{U}''$, where also
    $\gradual{U}''\not\sqsubseteq \gradual{U}$. So both normalize to $ {\qm } $.

\end{proof} 

We note that we can apply the exact same proof procedure to achieve the same result for elaborated terms.

\begin{lemma}
    If $ \evterm{v} \sqsubseteq  \evterm{t} $, then $\evterm{t}$ is a syntactic value.
\end{lemma}
\begin{proof}
    If $\evterm{t}= {\qm } $, then it is a value.
    Otherwise we proceed by induction on the structure of $\evterm{v}$,
    noting that a less precise value must have the same top-level constructor if it is not $ {\qm } $.
    For $\myepsilon \, \evterm{rv}$ and $\ottsym{(}  \mathit{x}  \ottsym{:}  \evterm{rv}  \ottsym{)}  \rightarrow  \evterm{T}$, it follows by applying our hypothesis to $\evterm{rv}$ and $\evterm{rV}$ respectively.
    $ \TypeType_{ \ottnt{i} } $ and $ {\qm } $ are trivial, and a function is a value regardless of its body, giving us our result.
      
\end{proof}

\gradGuarantee*
\begin{proof}[Proof of Static Guarantee]
    We prove by mutual induction with the following proposition: if $\gradual{Gamma}  \vdash  \gradual{t}  \Rightarrow  \gradual{U}$, $\gradual{Gamma} \sqsubseteq \gradual{Gamma}'$ and $\gradual{t} \sqsubseteq \gradual{t}'$, then $\gradual{Gamma}'  \vdash  \gradual{t}'  \Rightarrow  \gradual{U}'$   for some $\gradual{U}'$ where $ \gradual{U} \sqsubseteq  \gradual{U}' $.

    First we note that if $\gradual{t}'= {\qm } $, then our results holds trivially: $ {\qm } $ synthesizes the least precise type,
    and can check against any type using \rrule{GCheckSynth}.
    We hence assume that $\gradual{t}'\neq {\qm } $, and proceed by induction on the typing derivation for $\gradual{t}$.

    \rrule{GSynthAnn}: then $\gradual{t}= \gradual{t}_{{\mathrm{1}}} \dblcolon \gradual{T} $, and $\gradual{t}'= \gradual{t}'_{{\mathrm{1}}} \dblcolon \gradual{T}' $.
    By our premise, $\gradual{T}$ has some normal form $\gradual{U}_{{\mathrm{1}}}$.
    By \autoref{thm:normGuarantee}, $\gradual{T}'$ has a normal form $\gradual{U}'_{{\mathrm{1}}}$ where $ \gradual{U}_{{\mathrm{1}}} \sqsubseteq  \gradual{U}'_{{\mathrm{1}}} $.
    By our inductive hypothesis, $\gradual{Gamma}  \vdash  \gradual{t}'_{{\mathrm{1}}}  \Leftarrow  \gradual{U}'_{{\mathrm{1}}}$, allowing us to complete our typing derivation.

    \rrule{GSynthType}: then $\gradual{t}=\gradual{t}'= \TypeType_{ \ottnt{i} } $ and we can use an identical typing derivation.

    \rrule{GSynthVar}: then $\gradual{t}=\gradual{t}'=\mathit{x}$, and by our premise that $ \gradual{Gamma} \sqsubseteq  \gradual{Gamma}' $, the synthesized type $\gradual{U}$
    from $\gradual{Gamma}$ is at least as precise as $\gradual{U}'$ from $\gradual{Gamma}'$.

    \rrule{GSynthApp}: then $\gradual{t}= \gradual{t}_{{\mathrm{1}}} \  \gradual{t}_{{\mathrm{2}}} $, and $\gradual{t}'= \gradual{t}'_{{\mathrm{1}}} \  \gradual{t}'_{{\mathrm{2}}} $. The result then follows from our hypothesis,
    combined with the monotonicity of domain and codomain.

    \rrule{GSynthDyn}: vacuous.

    \rrule{GCheckSynth}: our hypothesis gives that $\gradual{Gamma}  \vdash  \gradual{t}  \Rightarrow  \gradual{U}_{{\mathrm{1}}}$ and $\gradual{Gamma}'  \vdash  \gradual{t}'  \Rightarrow  \gradual{U}'_{{\mathrm{1}}}$
    where $ \gradual{U}_{{\mathrm{1}}} \sqsubseteq  \gradual{U}'_{{\mathrm{1}}} $. Since we know $ \gradual{U} \sqsubseteq  \gradual{U}' $, we know that $ \gradual{U}'_{{\mathrm{1}}}  \cong  \gradual{U}' $, giving us our typing derivation.

    \rrule{GCheckLevel}: if $\gradual{Gamma}  \vdash  \gradual{t}  \Rightarrow   \TypeType_{ \ottnt{i} } $, by our hypothesis, either $\gradual{Gamma}'  \vdash  \gradual{t}'  \Rightarrow   \TypeType_{ \ottnt{i} } $,
    or $\gradual{Gamma}'  \vdash  \gradual{t}'  \Rightarrow  {\qm }$. In the first case, we can type $\gradual{t}'$ with \rrule{GCheckLevel},
    and in the second we can use \rrule{GCheckSynth}.

    \rrule{GCheckPi}: follows from our hypothesis and \autoref{thm:normGuarantee}.

    \rrule{GCheckLamPi}, \rrule{GCheckLamDyn}: follows from our hypothesis.

\end{proof}

\begin{proof}[Proof of Dynamic Guarantee]
    We prove a slightly stronger result:
        If $\evterm{t}_{{\mathrm{1}}}  \longrightarrow  \evterm{t}_{{\mathrm{2}}}$, where $\evterm{t}_{{\mathrm{2}}}\neq\mathsf{err}$ and $ \evterm{t}_{{\mathrm{1}}} \sqsubseteq  \evterm{t}'_{{\mathrm{1}}} $, then $\evterm{t}'_{{\mathrm{1}}}\longrightarrow^{*}\evterm{t}'_{{\mathrm{2}}}$ for some $\evterm{t}'_{{\mathrm{2}}}$ such that
        $ \evterm{t}_{{\mathrm{2}}} \sqsubseteq  \evterm{t}'_{{\mathrm{2}}} $. 
        
        We first note that if $\evterm{t}'_{{\mathrm{1}}}=  {\qm } $, then the result trivially holds since $ {\qm } \longrightarrow^{*} {\qm } $. We assume then that
        $\evterm{t}'_{{\mathrm{1}}}\neq {\qm } $, and proceed by
        induction on the derivation of $\evterm{t}_{{\mathrm{1}}}  \longrightarrow  \evterm{t}_{{\mathrm{2}}}$.
        
        \rrule{StepAscrFail}, \rrule{StepAppFailTrans}, \rrule{StepContextErr}: vacuous.
    
        \rrule{StepAscr}: then $\evterm{t}_{{\mathrm{1}}}=\myepsilon_{{\mathrm{1}}} \, \ottsym{(}  \myepsilon_{{\mathrm{2}}} \, \evterm{rv}  \ottsym{)}$ and $\evterm{t}'_{{\mathrm{1}}}=ep1'(ep2' rv')$.
        If $ \evterm{t}_{{\mathrm{1}}} \sqsubseteq  \evterm{t}'_{{\mathrm{1}}} $, then $ \myepsilon_{{\mathrm{1}}} \sqsubseteq  \myepsilon'_{{\mathrm{1}}} $, $ \myepsilon_{{\mathrm{2}}} \sqsubseteq  \myepsilon'_{{\mathrm{2}}} $ and $ \evterm{rv}_{{\mathrm{1}}} \sqsubseteq  \evterm{rv}'_{{\mathrm{1}}} $.
        Our premise gives that $ \myepsilon_{{\mathrm{1}}}  \sqcap  \myepsilon_{{\mathrm{2}}} $ is defined, and since we define meet and precision on sets of static values,
        $  \myepsilon_{{\mathrm{1}}}  \sqcap  \myepsilon_{{\mathrm{2}}}  \sqsubseteq   \myepsilon'_{{\mathrm{1}}}  \sqcap  \myepsilon'_{{\mathrm{2}}}  $. So $\evterm{t}_{{\mathrm{1}}}  \longrightarrow  \langle   \myepsilon'_{{\mathrm{1}}}  \sqcap  \myepsilon'_{{\mathrm{2}}}   \rangle \, \evterm{rv}'$ and $ \langle   \myepsilon_{{\mathrm{1}}}  \sqcap  \myepsilon_{{\mathrm{2}}}   \rangle \, \evterm{rv} \sqsubseteq  \langle   \myepsilon'_{{\mathrm{1}}}  \sqcap  \myepsilon'_{{\mathrm{2}}}   \rangle \, \evterm{rv}' $.

        \rrule{StepAppEv}: Then $\evterm{t}_{{\mathrm{1}}}=  \ottsym{(}  \myepsilon_{{\mathrm{1}}} \, \ottsym{(}  \lambda  \mathit{x}  \ldotp  \evterm{t}_{{\mathrm{3}}}  \ottsym{)}  \ottsym{)} \  \ottsym{(}  \myepsilon_{{\mathrm{2}}} \, \evterm{rv}  \ottsym{)} $.
        We have two cases for our precision relation to hold. 
        
        In the first, $\evterm{t}_{{\mathrm{2}}} =  \ottsym{(}  \myepsilon'_{{\mathrm{1}}} \, {\qm }  \ottsym{)} \  \ottsym{(}  \myepsilon'_{{\mathrm{2}}} \, \evterm{rv}'  \ottsym{)} $ where $ \myepsilon_{{\mathrm{1}}} \sqsubseteq  \myepsilon'_{{\mathrm{1}}} $ and $ \myepsilon_{{\mathrm{2}}} \sqsubseteq  \myepsilon'_{{\mathrm{2}}} $. By our premise, $ [  \myepsilon_{{\mathrm{2}}} \, \evterm{rv}  / {\_} ]  \ottkw{cod}\  \myepsilon_{{\mathrm{1}}}  =  \myepsilon_{{\mathrm{4}}} $ for some $\myepsilon_{{\mathrm{4}}}$, so there must be some  $\myepsilon'_{{\mathrm{4}}}$ where $ [  \myepsilon'_{{\mathrm{2}}} \, \evterm{rv}'  / {\_} ]  \ottkw{cod}\  \myepsilon'_{{\mathrm{1}}}  =  \myepsilon'_{{\mathrm{4}}} $,
         and by \autoref{cor:codsub-precise}, $ \myepsilon_{{\mathrm{4}}} \sqsubseteq  \myepsilon'_{{\mathrm{4}}} $. So we can step $ \ottsym{(}  \myepsilon'_{{\mathrm{1}}} \, {\qm }  \ottsym{)} \  \ottsym{(}  \myepsilon'_{{\mathrm{2}}} \, \evterm{rv}'  \ottsym{)}   \longrightarrow  \myepsilon'_{{\mathrm{4}}} \, {\qm }$ by \rrule{StepAppDyn},
         and since $ \myepsilon_{{\mathrm{4}}} \sqsubseteq  \myepsilon'_{{\mathrm{4}}} $ and $ {\qm } $ is less precise than all terms, we have our result.
    
         In the second case, $\evterm{t}'_{{\mathrm{1}}} =  \ottsym{(}  \myepsilon_{{\mathrm{1}}} \, \ottsym{(}  \lambda  \mathit{x}  \ldotp  \evterm{t}'_{{\mathrm{3}}}  \ottsym{)}  \ottsym{)} \  \ottsym{(}  \myepsilon'_{{\mathrm{2}}} \, \evterm{rv}  \ottsym{)} $. Then we can step $\evterm{t}'_{{\mathrm{1}}}$ using \rrule{StepAppEv}.
         By definition of $\ottkw{dom}$, \autoref{cor:codsub-precise} and the monotonicity of the precision meet,
         the evidences created to step $\evterm{t}'_{{\mathrm{1}}}$ are all no more precise than the corresponding ones for $\evterm{t}_{{\mathrm{1}}}$.
         Since syntactic substitution preserves precision, we have our result.

        \rrule{StepAppEvRaw}: By the same argument as \rrule{StepAppEv}, except that in the second case we need not apply monotonicity of the meet.

        \rrule{StepAppDyn}: Then $\evterm{t}= \ottsym{(}  \myepsilon_{{\mathrm{1}}} \, {\qm }  \ottsym{)} \  \evterm{v} $, so for precision to hold, $\evterm{t}'_{{\mathrm{1}}}$ must be $ \ottsym{(}  \myepsilon'_{{\mathrm{1}}} \, {\qm }  \ottsym{)} \  \evterm{v}' $, where $ \myepsilon_{{\mathrm{1}}} \sqsubseteq  \myepsilon'_{{\mathrm{1}}} $ and $ \evterm{v} \sqsubseteq  \evterm{v}' $.
        By \autoref{cor:codsub-precise}, if $ [  \evterm{v}  / {\_} ]  \ottkw{cod}\  \myepsilon_{{\mathrm{1}}}  =  \myepsilon_{{\mathrm{2}}} $ and $ [  \evterm{v}  / {\_} ]  \ottkw{cod}\  \myepsilon'_{{\mathrm{1}}}  =  \myepsilon'_{{\mathrm{2}}} $, then $ \myepsilon_{{\mathrm{2}}} \sqsubseteq  \myepsilon'_{{\mathrm{2}}} $,
        so we can step $\evterm{t}'_{{\mathrm{1}}}  \longrightarrow  \myepsilon'_{{\mathrm{2}}} \, {\qm }$ and our precision result holds.
    
        \rrule{StepContext}: If $\evterm{t}_{{\mathrm{1}}}= \evterm{C}[  \evterm{t}_{{\mathrm{3}}}  ] $ and $\evterm{t}'_{{\mathrm{1}}}= \evterm{C}[  \evterm{t}'_{{\mathrm{3}}}  ] $, by our premise and inductive hypothesis,
        we have $\evterm{t}_{{\mathrm{3}}}  \longrightarrow  \evterm{t}_{{\mathrm{4}}}$, $\evterm{t}'_{{\mathrm{3}}}  \longrightarrow  \evterm{t}'_{{\mathrm{4}}}$ and $ \evterm{t}_{{\mathrm{4}}} \sqsubseteq  \evterm{t}'_{{\mathrm{4}}} $.
        If $\evterm{t}_{{\mathrm{1}}}= \evterm{t}_{{\mathrm{3}}} \  \evterm{t}_{{\mathrm{5}}} $, then $\evterm{t}'_{{\mathrm{1}}}= \evterm{t}'_{{\mathrm{3}}} \  \evterm{t}'_{{\mathrm{5}}} $ and $ \evterm{t}_{{\mathrm{5}}} \sqsubseteq  \evterm{t}'_{{\mathrm{5}}} $, so $  \evterm{t}_{{\mathrm{4}}} \  \evterm{t}_{{\mathrm{5}}}  \sqsubseteq   \evterm{t}'_{{\mathrm{4}}} \  \evterm{t}'_{{\mathrm{5}}}  $,
        and we can step to $ \evterm{C}[  \evterm{t}'_{{\mathrm{4}}}  ] $ using \rrule{StepContext}, and preserve the precision relation.
        Similar reasoning shows the same result for the other possible frames,
        though we note that the fact that precision preserves the value property is required for the frames involving values.
    \end{proof}

  \end{document}